\documentclass[12pt,letterpaper]{article}

\usepackage{graphics,epsfig,graphicx,color}
\graphicspath{{/EPSF/}{../Figures/}{Figures/}}

\usepackage{amssymb,latexsym}

\usepackage{setspace}

\usepackage{amsmath}

\usepackage{mathrsfs,amsfonts}

\usepackage{amsthm,amsxtra}

\newif\ifPDF
\ifx\pdfoutput\undefined
	\PDFfalse
\else
	\ifnum\pdfoutput > 0
		\PDFtrue
	\else
		\PDFfalse
	\fi
\fi

\ifPDF
	\usepackage{pdftricks}
	\begin{psinputs}
		\usepackage{pstricks}
		\usepackage{pstcol}
		\usepackage{pst-plot}
		\usepackage{pst-tree}
		\usepackage{pst-eps}
		\usepackage{multido}
		\usepackage{pst-node}
		\usepackage{pst-eps}
	\end{psinputs}
\else
	\usepackage{pstricks}
\fi

\ifPDF
	\usepackage[debug,pdftex,colorlinks=true, 
	linkcolor=blue, bookmarksopen=false,
	plainpages=false,pdfpagelabels]{hyperref}
\else
	\usepackage[dvips]{hyperref}
\fi

\usepackage{fancyhdr}

\newtheorem{theorem}{Theorem}[section]
\newtheorem{lemma}[theorem]{Lemma}

\newtheorem{proposition}[theorem]{Proposition} 
\newtheorem{remark}[theorem]{Remark} 
\newtheorem{corollary}[theorem]{Corollary}

\numberwithin{equation}{section}

\newcommand{\aver}[1]{\langle {#1} \rangle}

\newcommand{\diag}{{\rm diag}}

\newcommand{\dsum}{\displaystyle\sum}
\newcommand{\dint}{\displaystyle\int}

\newcommand{\bxi}{\boldsymbol \xi}

\newcommand{\bbR}{\mathbb R} \newcommand{\bbS}{\mathbb S}

 \newcommand{\bn}{\mathbf n}

 \newcommand{\bv}{\mathbf v} 
 \newcommand{\bx}{\mathbf x} 
\newcommand{\by}{\mathbf y}

 \newcommand{\bH}{\mathbf H}

\newcommand{\cA}{\mathcal A} 
\newcommand{\cC}{\mathcal C} \newcommand{\cD}{\mathcal D} 
\newcommand{\cE}{\mathcal E} 
 \newcommand{\cH}{\mathcal H}
\newcommand{\cI}{\mathcal I} 
\newcommand{\cK}{\mathcal K} \newcommand{\cL}{\mathcal L}
\newcommand{\cM}{\mathcal M}

\newenvironment{keywords}
{\noindent{\bf Key words.}\small}{\par\vspace{1ex}}

\title{Quantitative photoacoustic imaging in radiative transport regime}

\author{
	Alexander V. Mamonov\thanks{
		Institute for Computational Engineering and Sciences, 
		University of Texas, 
		Austin, TX 78712;
		\href{mailto:mamonov@ices.utexas.edu}{mamonov@ices.utexas.edu}
	}\and
	Kui Ren\thanks{
		Department of Mathematics, 
		University of Texas, 
		Austin, TX 78712;
		\href{mailto:ren@math.utexas.edu}{ren@math.utexas.edu}
	}
}
    
\begin{document}

\maketitle



\begin{abstract}
The objective of quantitative photoacoustic tomography (QPAT) is to reconstruct optical and thermodynamic properties of heterogeneous media from data of absorbed energy distribution inside the media. There have been extensive theoretical and computational studies on the inverse problem in QPAT, however, mostly in the diffusive regime. We present in this work some numerical reconstruction algorithms for multi-source QPAT in the radiative transport regime with energy data collected at either single or multiple wavelengths. We show that when the medium to be probed is non-scattering, explicit reconstruction schemes can be derived to reconstruct the absorption and the Gr\"uneisen coefficients. When data at multiple wavelengths are utilized, we can reconstruct simultaneously the absorption, scattering and Gr\"uneisen coefficients. We show by numerical simulations that the reconstructions are stable.
\end{abstract}






\section{Introduction}
\label{SEC:intro}

In photoacoustic tomography (PAT) experiment, we send near infra-red (NIR) light into a biological tissue. The tissue absorbs part of the incoming light and heats up due to the absorbed energy. The heating then results in expansions of the tissue and the expansion generates compressive (acoustic) waves. We then measure the time-dependent acoustic signal that arrives on the surface of the tissue. From the knowledge of these acoustic measurements, we are interested in reconstructing the absorption and scattering properties of the tissue, as well as the thermodynamic Gr\"uneisen parameter which measures the photoacoustic efficiency of the tissue. We refer interested readers to~\cite{Ammari-Book08,Bal-IO12,Beard-IF11,CoLaArBe-JBO12,CoLaBe-SPIE09,Kuchment-MLLE12,KuKu-HMMI10,LiWa-PMB09,PaSc-IP07,Wang-DM04,Wang-IEEE08,XuWa-RSI06} for overviews of the field of photoacoustic imaging.

The propagation of NIR light in biological tissues is accurately modeled by the radiative transport equation which describes the distribution of photons in the phase space. To be precise, let $\Omega\in\bbR^d$ ($d\ge 2$) be the domain of interest with smooth boundary $\partial\Omega$ and $\bbS^{d-1}$ be the unit sphere in $\bbR^d$. We denote by $X=\Omega\times \bbS^{d-1}$ the phase space and $\Gamma_\pm=\{(\bx,\bv): (\bx,\bv)\in\partial\Omega\times \bbS^{d-1}\ \text{s.t.}\ \pm\bn(\bx)\cdot \bv >0\}$ its incoming and outgoing boundaries, $\bn(\bx)$ being the unit outer normal vector at $\bx\in\partial\Omega$. The radiative transport equation for photon density $u(\bx,\bv)$ can then be written as~\cite{Arridge-IP99,Bal-IP09,Ren-CiCP10}:
\begin{equation}\label{EQ:ERT}
	\begin{array}{rcll}
  	\bv\cdot\nabla u(\bx,\bv) + 
  	\sigma_a(\bx)u(\bx,\bv)
  	&=& \sigma_s(\bx)K(u)(\bx,\bv)
  	&\mbox{ in }\ X\\
       u(\bx,\bv) &=& g(\bx,\bv) &\mbox{ on }\ \Gamma_{-},
	\end{array}
\end{equation}
Here the positive functions $\sigma_a(\bx)$ and $\sigma_s(\bx)$ are the absorption and the scattering coefficients, respectively. The function $g(\bx,\bv)$ is the incoming illumination source. The scattering operator $K$ is defined as 
\[
	K(u)(\bx,\bv)=
		\dint_{\bbS^{d-1}}\cK(\bv, \bv')u(\bx,\bv')d\bv' - u(\bx,\bv)
\]
where $d\bv$ is the normalized measure on $\bbS^{d-1}$ in the sense that $\int_{\bbS^{d-1}}d\bv=1$, and the kernel $\cK(\bv, \bv')$ describes the way that photons traveling in direction $\bv'$ getting scattered into direction $\bv$, and satisfies the normalization condition $\int_{\bbS^{d-1}}\cK(\bv, \bv')d \bv'=1,\  \forall\ \bv\in \bbS^{d-1}$. In practical applications in biomedical optics, $\cK$ is often taken to be the Henyey-Greenstein phase function~\cite{Arridge-IP99,HeGr-AJ41,WeVa-Book95} which depends only on the product $\bv\cdot\bv'$; see equation~\eqref{EQ:HG} in Section~\ref{SEC:Num}.

The photon energy that is absorbed at location $\bx\in\Omega$ per unit volume, $E(\bx)$, is the product of the absorption coefficient and the fluence distribution:
\begin{equation}\label{EQ:Energy}
	E(\bx) = \int_{\bbS^{d-1}} \sigma_a(\bx) u(\bx,\bv)d \bv.
\end{equation}
The heating due to this absorbed energy generates an initial pressure field, denoted by $H$, in the tissue that depends on the thermodynamic property of tissue and is proportional to $E$:
\begin{equation}\label{EQ:Data}
	H(\bx) =\Upsilon(\bx) E(\bx)\equiv \Upsilon(\bx) \int_{\bbS^{d-1}} \sigma_a(\bx) u(\bx,\bv)d \bv
\end{equation}
where the positive function $\Upsilon(\bx)$ is the nondimensional Gr\"uneisen coefficient which in the current formulation measures the photoacoustic efficiency of the tissue. To simplify the presentation, we use the short notation $\aver{f}_\bv$ to denote the integral of $f$ over the $\bv$ variable in the rest of this work.

The initial pressure field $H$ then propagates according to the acoustic wave equation, with the wave speed $c(\bx)$,
\begin{equation}\label{EQ:Acous}
	\begin{array}{cl}
  	\dfrac{1}{c^2(\bx)}\dfrac{\partial ^2 p}{\partial t^2}-\Delta p = 0, & \mbox{ in }\ \bbR_+\times \bbR^d\\
	p(0,\bx)= H(\bx), \quad \dfrac{\partial p}{\partial t}(0,\bx) =0 &\mbox{ in }\ \bbR^d.
	\end{array}
\end{equation}
The time-dependent pressure signal $p(t,\bx)$ is then measured on the surface of the tissue for long enough time, say $t_\infty$, and the objective is to reconstruct the coefficients $\sigma_a$, $\sigma_s$ and the Gr\"uneisen coefficient $\Upsilon$ from this measurement. Note that the reason that we can write the transport equation in stationary case while using the wave equation in time-dependent case, is that the two phenomena occur on two significantly different time scales~\cite{BaJoJu-IP10}.

The reconstruction problem in photoacoustic tomography can be split into two steps. In the first step, we need to reconstruct the initial pressure field $H(\bx)$ from measured acoustic signal on the boundary, $\left. p(t,\bx)\right|_{(0,T)\times\partial\Omega}$. This is a well known inverse problem for the acoustic wave equation that has been thoroughly studied in the past a few years under various scenarios; see for instance~\cite{AgKuKu-PIS09,AgQu-JFA96,BuMaHaPa-PRE07,CoArBe-IP07,FiHaRa-SIAM07,FiPaRa-SIAM04,FiRa-IP07,GoHiKu-Prep11,Haltmeier-SIAM11,Haltmeier-SIAM11B,HaScBuPa-IP04,HaScSc-M2AS05,KuKu-EJAM08,Kunyansky-IP08,Kunyansky-IP11,Kunyansky-IP07,Kunyansky-IP07B,Natterer-Prep11,Nguyen-IPI09,PaSc-IP07,RoRaNt-IEEE10,XuWa-PRE05} for analytical reconstruction formulas with constant wave speed, ~\cite{Hristova-IP09,HrKuNg-IP08,QiStUhZh-SIAM11,StUh-IP09,Steinhauer-Prep09B,Steinhauer-Prep09A,WuTaLi-JAP11} for reconstruction under variable wave speed, and ~\cite{AmBrGaWa-CM11,AmBrJuWa-LNM11,DeRaNt-PMB11,KoSc-Prep10,TrLaZhNoLyBeCo-PPUIS11,TrZhCo-IP10, AmGoLe-IP10,BuRoJeDiRaNt-MP11,CoBe-JOSA09,DeRaNt-PMB11,PaNuHaBu-IP07,PaNuBu-PMB09,QiStUhZh-SIAM11,StUh-TAMS12,StUh-IP11,TaLi-OE10,XuWaAmKu-MP04} for reconstructions under even more complicated environments. 

This work is concerned with the second step of photoacoustic tomography, call quantitative photoacoustic tomography (QPAT). The objective is to reconstruct the absorption and the diffusion coefficients, $\sigma_a$ and $\sigma_s$, in the transport equation~\eqref{EQ:ERT} and the Gr\"uneisen coefficient $\Upsilon$ from the result of the first step, the data $H$ in~\eqref{EQ:Data}. This step has recently attracted significant attention from both mathematical~\cite{AmBoJuKa-SIAM10,AmBoJuKa-SIAM11,BaJoJu-IP10,BaNoErWaCu-JBO11,BaRe-IP11,BaRe-CM11,BaRe-IP12,BaUh-IP10,BaUh-Prep11}, computational~\cite{BaBaVaRo-JOSA08,CoArBe-JOSA09,CoArKoBe-AO06,CoTaAr-CM11,GaZhOs-Prep10,GaOsZh-LNM12,LaCoZhBe-AO10,ReGaZh-SIAM12,RiNt-PRE05,ShCoZe-AO11,YuWaJi-OE07,Zemp-AO10} and modeling and experimental~\cite{BaSc-PRL10,RaNt-MP07} perspectives. Most of existing work on this step, however, is done in the diffusive regime~\cite{BaRe-IP11,BaRe-IP12,BaUh-IP10}, i.e., based on the diffusion approximation to the transport equation~\eqref{EQ:ERT}. Transport-based QPAT is only studied in~\cite{BaJoJu-IP10,CoTaAr-CM11,YaSuJi-PMB10} where the Gr\"uneisen coefficient $\Upsilon$ has been assumed to be a constant and known. 

To setup the problem appropriately, for the rest of the paper, we assume that all the coefficients that we are interested in are positive and bounded functions. More precisely, we assume: 
\begin{itemize}
\item[{\bf (A1)}] the coefficients $0<c_0\le \Upsilon(\bx), \sigma_a(\bx),\sigma_s(\bx) \in L^1(\Omega)\cap L^\infty(\Omega)$ for some $c_0>0$;
\item[{\bf (A2)}] the scattering kernel $\cK(\bv, \bv')\in L^1(\bbS^{d-1}\times \bbS^{d-1})$;
\item[{\bf (A3)}] the illumination, modeled by the boundary condition $g$, is in $L^1(\Gamma_-,d\bxi)$, with measure $d\bxi=|\bv\cdot\bn(\bx)|d\mathfrak m(\bx)d\bv$, $d\frak m(\bx)$ being the usual Lebesgue measure on $\partial\Omega$.
\end{itemize}
 With these assumptions, it is well-known that the radiative transport problem~\eqref{EQ:ERT} is well-posed and thus admits a unique solution $u\in L^1(X)$~\cite{DaLi-Book93-6}. This means that the data $H$ in~\eqref{EQ:Data} is a well-defined function in $L^1(\Omega)$. In practical applications, the objects to be imaged are often embedded into phantoms of similar optical and acoustic properties to get regularly shaped imaging domains. We thus assume will assume that the domain $\Omega$ is convex. This assumption will simplify some of the presentation but is not essential for the results obtained.

We conclude this section with the two remarks. First, if both the absorption coefficient $\sigma_a(\bx)$ and the scattering coefficient $\sigma_s(\bx)$ are known and only the Gr\"uneisen coefficient has to be reconstructed, we can simply solve the transport equation~\eqref{EQ:ERT} and compute the energy $E(\bx)$. Then $\Upsilon(\bx)$ is reconstructed as $\Upsilon=\frac{H}{E}$. Thus we need only one interior data set and one transport solver to solve the inverse problem. This is a trivial case. We will not discuss this case in the rest of the paper. 

Second, in practical application of PAT in biomedical imaging, the absorption and the scattering coefficients $\sigma_a$ and $\sigma_s$ are often isotropic, i.e. independent of the angular variable $\bv$. We thus restrict ourselves to the case of isotropic coefficients in this work. Mathematically this assumption is essential, as we can see from the following result that it is not possible to reconstruct uniquely anisotropic coefficients.
\begin{proposition}
	Let $(\Upsilon,\sigma_a,\sigma_s)$ and $(\tilde\Upsilon,\tilde\sigma_a,\sigma_s)$ be two sets of coefficients, and $H$ and $\tilde H$ the corresponding data sets. Let $z(\bx)\in\cC^1(\bar\Omega)$ be an arbitrary positive function with boundary value $\left. z \right|_{\partial\Omega}=1$. 
Then
\begin{equation}\label{EQ:Invariance}
	\tilde\sigma_a = (\sigma_a -\bv\cdot\nabla \ln z) z \quad \mbox{and}\quad \tilde\Upsilon\tilde\sigma_a =\Upsilon \sigma_a
\end{equation}
implies $\tilde H = H$.
\end{proposition}
\begin{proof}
	Let  $u$ be the solution of the transport equation for coefficients $(\sigma_a,\sigma_s)$ with boundary value $g$. It is straightforward to verify that $uz$ is the solution of the transport equation for coefficients $(\sigma_a-\bv\cdot\nabla \ln z,\sigma_s)$ with the same boundary value $g$ (because $\left. z \right|_{\partial\Omega}=1$). Now, clearly 
$\tilde\Upsilon\dint_{\bbS^{d-1}}(\sigma_a-\bv\cdot\nabla\ln z)uz\ d\bv=\Upsilon\dint_{\bbS^{d-1}}\sigma_a u\ d\bv$ if 
\eqref{EQ:Invariance} holds.
\end{proof}

The rest of the paper is structured as follows. We first study in Section~\ref{SEC:Non-scattering} the inverse transport problem for non-scattering media for applications in quantitative sectional QPAT. We present some analytical reconstruction strategies in this setting. We then study in Section~\ref{SEC:Scattering} the same inverse problem for scattering media. We linearize the nonlinear inverse problem using the tool of Born approximation and present some numerical procedures to solve the linear and nonlinear inverse problems. In Section~\ref{SEC:MultiSpec} we consider the QPAT problem in the case when illuminations with multiple wavelengths are available.  To validate the reconstruction strategies, we provide in Section~\ref{SEC:Num} several numerical simulations with synthetic data generated under different scenarios. We conclude the paper with some additional remarks in Section~\ref{SEC:Concl}.

\section{QPAT of non-scattering media}
\label{SEC:Non-scattering}

We start by considering the QPAT problem for non-scattering media. In this case, photons propagate in the media along straight trajectories, without changing directions of propagation. The scattering mechanism in the transport equation~\eqref{EQ:ERT} is thus dropped by setting the scattering coefficient $\sigma_s=0$. We have the following transport model for light propagation:
\begin{equation}\label{EQ:ERT Free}
	\begin{array}{rcll}
  	\bv\cdot\nabla u(\bx,\bv) + 
  	\sigma_a(\bx)u(\bx,\bv) &=& 0
  	&\mbox{ in }\ X\\
       u(\bx,\bv) &=& g(\bx,\bv) &\mbox{ on }\ \Gamma_{-} .
	\end{array}
\end{equation}
The fact that photons travel in straight lines allows us to illuminate only part of the domain, for instance a plane cut through a three-dimensional medium. This is the fundamental principle of sectional photoacoustic tomography~\cite{ZhSe-OE11,ElScSc-arXiv11}.

To simplify the presentation, for any point $\bx\in\Omega$ and direction $\bv\in\bbS^{d-1}$, let us define 
\[
	\tau_\pm(\bx,\bv)=\inf\{s\in\bbR_+|\bx\pm s\bv\notin \Omega\}, 
\]
and $\tau(\bx,\bv)=\tau_+(\bx,\bv)+\tau_-(\bx,\bv)$. It is easy to see that $\tau_+(\bx,\bv)$ (resp. $\tau_-(\bx,\bv)$) is the distance it takes a particle at $\bx$ traveling in direction $\bv$ (resp. $-\bv$) to reach the boundary of the domain. In the same spirit, for any point $(\bx,\bv)$ on the incoming boundary $\Gamma_-$, we define
\[
	\tau_+(\bx,\bv)=\sup\{s\in\bbR_+|\bx\pm s\bv \in \Omega\},
\]
and set $\tau_-(\bx,\bv)=0$. Thus $\tau_+(\bx,\bv)$ is the distance for a photon coming into the domain at $\bx$ in direction $\bv$ to exit the domain.

It is straightforward to show the following representation of the interior data.
\begin{lemma}
	The interior data $H(\bx)$ generated with source $g$ and coefficients $(\Upsilon,\sigma_a)$ can be written as
\begin{equation}\label{EQ:H Expr}
	H(\bx)=\Upsilon(\bx)\sigma_a(\bx)\int_{\bbS^{d-1}}g(\bx-\tau_-(\bx,\bv)\bv,\bv)e^{-\int_0^{\tau_-(\bx,\bv)}\sigma_a(\bx-\tau_-(\bx,\bv)\bv+s\bv)ds}d\bv. 
\end{equation}	
\end{lemma}
\begin{proof}
The transport equation~\eqref{EQ:ERT Free} can be integrated along direction $\bv$ as an ODE to obtain the solution:
\begin{equation}\label{EQ:Solu}
	u(\bx,\bv)=g(\bx-\tau_-(\bx,\bv)\bv,\bv)e^{-\int_0^{\tau_-(\bx,\bv)}\sigma_a(\bx-\tau_-(\bx,\bv)\bv+s\bv)ds} .
\end{equation}
The result then follows from using this solution in data $H$ introduced in~\eqref{EQ:Data}.
\end{proof}

In fact, this simple representation of the transport solution allows us to obtain explicit procedure to reconstruct the two unknown coefficients $\Upsilon$ and $\sigma_a$ if we have the luxury of using the right illumination source $g(\bx,\bv)$.

\begin{figure}[ht]
\centering
\includegraphics[angle=0,width=0.3\textwidth]{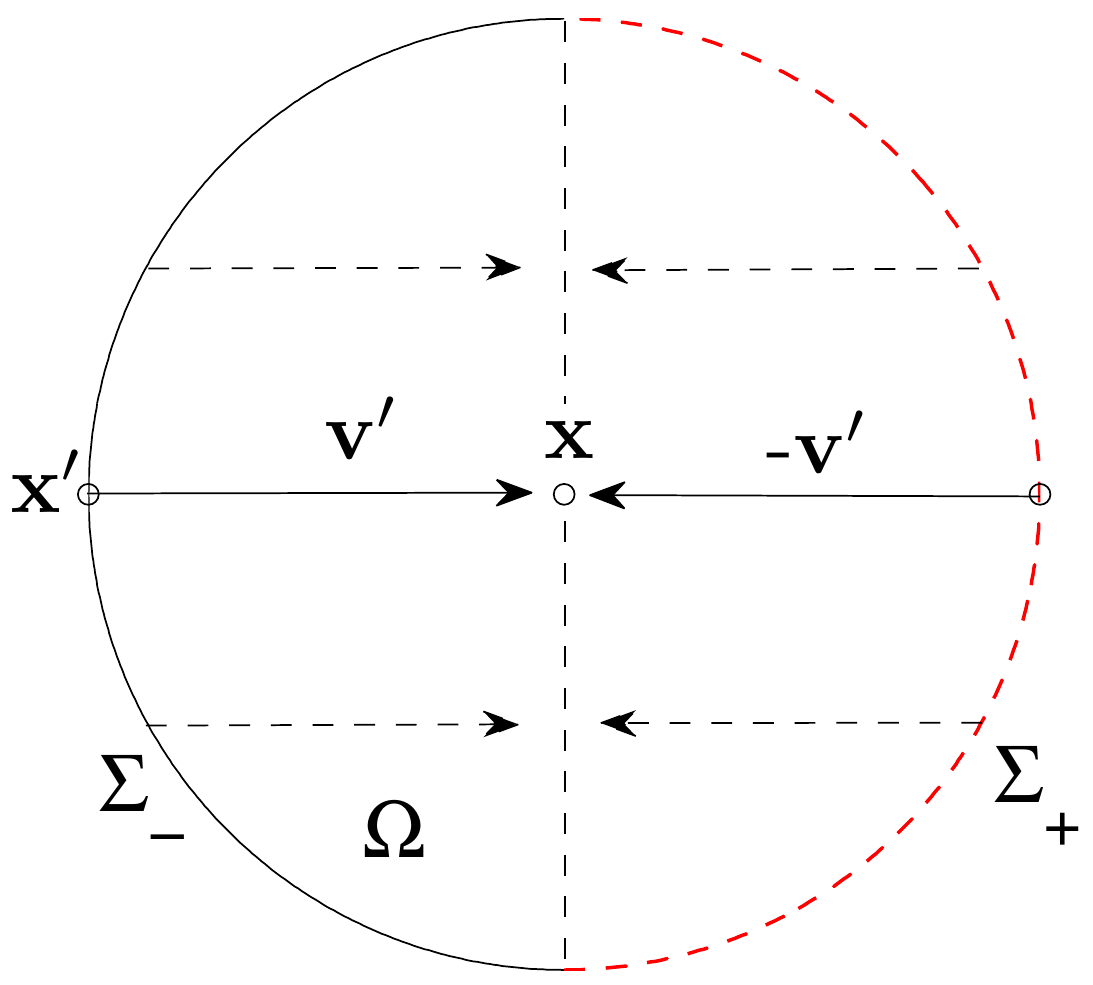}\hskip 0.1\textwidth
\includegraphics[angle=0,width=0.3\textwidth]{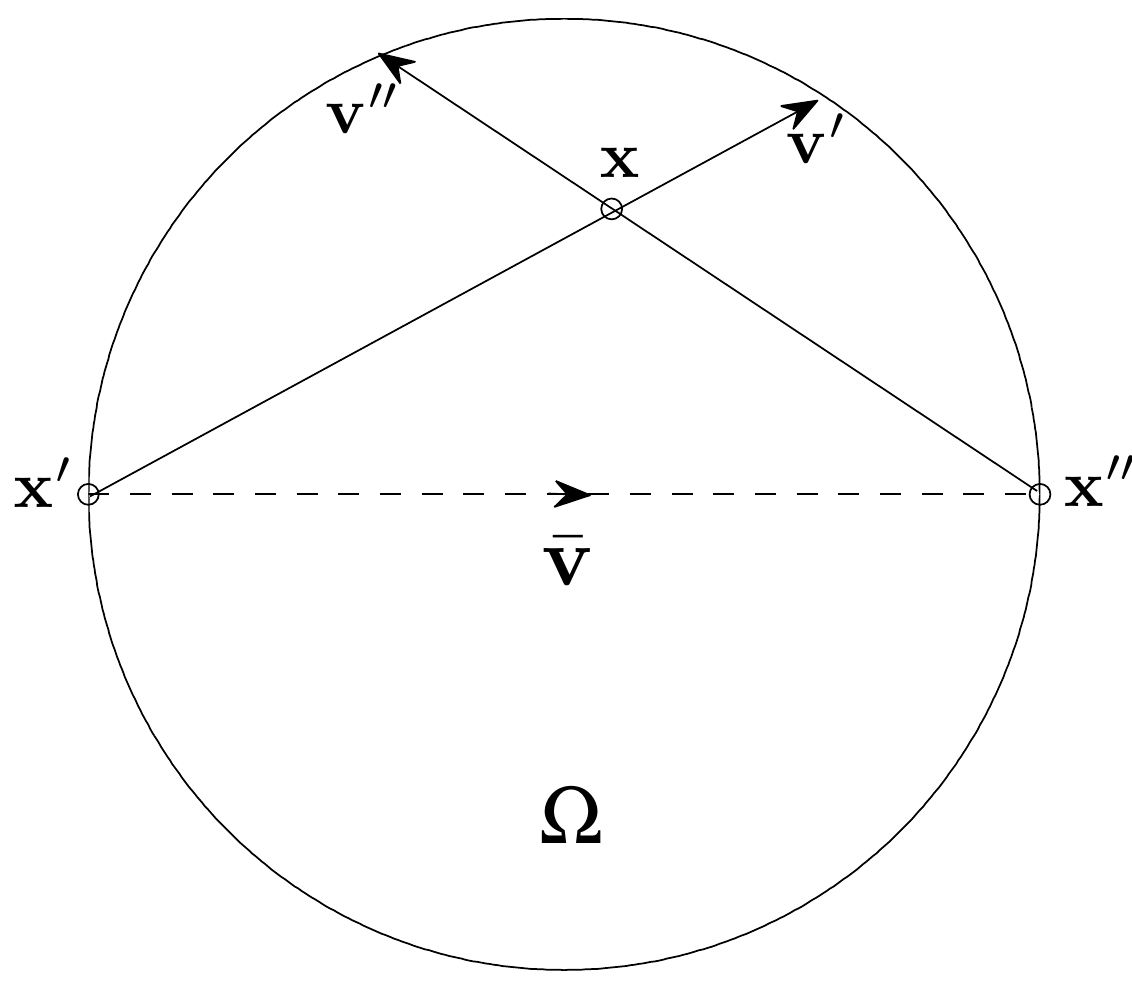}
\caption{Two illumination schemes that allow the simultaneous reconstruction of $\Upsilon$ and $\sigma_a$ using two data sets in a circular domain. Left: two collimated sources supported on $\Sigma_-$ (solid part of $\partial\Omega$) and $\Sigma_+$ (dashed part of $\partial\Omega$) respectively; Right: two point sources located at $\bx'$ and $\bx''$ respectively.}
\label{FIG:Sec Setup}
\end{figure}

\subsection{Reconstruction with collimated sources}
\label{SUBSEC:rec-collimated}

The first case where we can derive an analytical reconstruction formula is when collimated sources, i.e., sources focused on a specific direction, are used. Let $\bv'$ be the direction that the source points to, and define $\Sigma_\pm(\bv')\equiv\{\bx\in\partial\Omega| \pm\bn(\bx)\cdot \bv'>0\}$, $\Sigma_-(\bv')$ (resp. $\Sigma_+(\bv')$) being the part of the boundary  $\partial\Omega$ where lines in direction $\bv'$ enters (resp. leaves) the domain $\Omega$. As mentioned in the Introduction, we assume that the domain $\Omega$ is convex. In this case, each point $\bx\in\Sigma_-(\bv')$ is uniquely mapped to a point $\by\in\Sigma_+(\bv')$ as: $\by=\bx+\tau_+(\bx,\bv')\bv'$ and vise versa: $\bx=\by+\tau_+(\by,-\bv')(-\bv')$.

The collimated illumination source that we choose has the form
\begin{equation}\label{EQ:Collimated S}
	g(\bx,\bv)=\frak g(\bx)\delta(\bv-\bv'),\ \ \bx\in\Sigma_-(\bv') .
\end{equation}
This leads to the following expression for the interior data:
\[
	H(\bx)=\frak g(\bx-\tau_-(\bx,\bv')\bv')\Big(\Upsilon(\bx)\sigma_a(\bx)e^{-\int_0^{\tau_-(\bx,\bv')}\sigma_a(\bx-\tau_-(\bx,\bv')\bv'+s\bv')ds}\Big) .
\]
Let $\bx'\in\partial\Omega$ be the track back of $\bx$ to the boundary in $-\bv$ direction, i.e., $\bx'=\bx-\tau_-(\bx,\bv')\bv'$, then $H$ can be written as
\begin{equation}\label{EQ:A}
	\dfrac{H(\bx'+\tau_-(\bx,\bv')\bv')}{\frak g(\bx')\Upsilon(\bx'+\tau_-(\bx,\bv')\bv')}=\sigma_a(\bx'+\tau_-(\bx,\bv')\bv')e^{-\int_0^{\tau_-(\bx,\bv')}\sigma_a(\bx'+s\bv')ds} .
\end{equation}

\noindent{\bf (i) Reconstruction of $\sigma_a$.} If $\Upsilon$ is known, then we can take the log of the above formula and differentiate with respect to $\tau_-$ to obtain the following equation for $\sigma_a$:
\begin{equation}
	\begin{array}{rcll}
  	\bv'\cdot\nabla \sigma_a(\bx) -\sigma_a^2(\bx) - 
        \left( \bv'\cdot\nabla \ln \dfrac{H}{\Upsilon \frak g} \right) \sigma_a(\bx) &=& 0,
  	&\mbox{in}\ \ \Omega \\
       \sigma_a(\bx) &=& \dfrac{H(\bx)}{\Upsilon(\bx) \frak g(\bx)}, &\mbox{on}\ \ \Sigma_-
	\end{array}
\end{equation}
where $\frak g$ in the term $\bv'\cdot\nabla \ln \dfrac{H}{\Upsilon \frak g}$ is evaluated at $\bx-\tau_-(\bx,\bv')\bv'$. This is an initial value problem for a first order ordinary differential equation (because the direction $\bv'$ is fixed). It can be solved uniquely~\cite{Friedrichs-Book85} to reconstruct the absorption coefficient $\sigma_a$ along the lines in direction $\bv'$.

There is an equivalent reconstruction procedure for $\sigma_a$. When a collimated source is used, photons travel only along the direction that the source is pointing to, say $\bv'$. Thus, the solution $u$ to the transport equation is only non-zero in this direction. In other words, $\aver{u}_\bv\equiv\int_{\bbS^{d-1}} u(\bx,\bv)d\bv = u(\bx,\bv')$, and hence the data can be expressed as $H=\Upsilon \sigma_a(\bx) u(\bx,\bv')$. We can thus replace the absorption term $\sigma_a(\bx) u(\bx,\bv)$ in the transport equation with the term $H/\Upsilon$ to obtain the following transport equation for the photon density $u$:
\begin{equation}
	\begin{array}{rcll}
  	\bv\cdot\nabla u(\bx,\bv) + \dfrac{H}{\Upsilon} &=& 0
  	&\mbox{ in }\ X\\
       u(\bx,\bv) &=& \frak g(\bx)\delta(\bv-\bv') &\mbox{on}\ \ \Gamma_{-} .
	\end{array}
\end{equation}
This transport equation can be solved uniquely for $u$~\cite{DaLi-Book93-6}:
\begin{equation}
	u(\bx'+t\bv,\bv)=\left\{
	\begin{array}{ll}
		\frak g(\bx') - \dint_0^t\dfrac{H}{\Upsilon}(\bx'+s\bv)ds, & \bx'\in\Sigma_-\ \mbox{and} \ \bv=\bv'\\
		0, & \bx'\in\Sigma_-\ \mbox{and} \ \bv\neq \bv' .
	\end{array}\right .
\end{equation}
Once $u$ is obtained, we can reconstruct the absorption  $\sigma_a=H/(\Upsilon u)$, provided that we select the boundary condition $\frak g$ such that the denominator does not vanish.

\medskip
\noindent{\bf (ii) Reconstruction of ($\Upsilon$, $\sigma_a$).} In fact, we can use two data sets generated with collimated sources to determine the Gr\"uneisen coefficient and the absorption coefficients simultaneously. The configuration of the two sources is illustrated in Fig.~\ref{FIG:Sec Setup} (left). Let $g_1(\bx,\bv)=\frak g(\bx)\delta(\bv-\bv')$ be the first collimated source supported on $\Sigma_-$, and $H_1$ be the corresponding interior data which takes the form of $H$ in ~\eqref{EQ:A}. Then the second source we use is supported on $\Sigma_+$ with identical intensity distribution but pointing in the opposite direction, i.e., $-\bv'$. More precisely,
\begin{equation}
	g_2(\bx,\bv)=\frak g(\bx-\tau_+(\bx,\bv')\bv')\delta(\bv+\bv'), \ \ \bx\in\Sigma_+(\bv') . 
\end{equation}
If we use this source in the expression for the data $H$ in~\eqref{EQ:H Expr}, we obtain the following expression for the interior data $H_2$:
\begin{equation}\label{EQ:B}
	\dfrac{H_2(\bx'+\tau_-(\bx,\bv')\bv')}{\frak g(\bx')\Upsilon(\bx'+\tau_-(\bx,\bv')\bv')}=\sigma_a(\bx'+\tau_-(\bx,\bv')\bv')e^{-\int_0^{\tau_+(\bx,\bv')}\sigma_a(\bx'+\tau(\bx,\bv') v'-s\bv')ds} .
\end{equation}

We now take the logarithm of the ratio of ~\eqref{EQ:A} and ~\eqref{EQ:B} and use the relation $\tau(\bx,\bv)=\tau_+(\bx,\bv)+\tau_-(\bx,\bv)$ to obtain
\begin{equation}
	\ln \dfrac{H_2}{H_1}(\bx'+\tau_-(\bx,\bv')\bv')=-\int_0^{\tau_+(\bx',\bv')}\sigma_a(\bx'+s\bv')ds + 2\int_0^{\tau_-(\bx,\bv')}\sigma_a(\bx'+s\bv')ds.
\end{equation}
Differentiation of this result with respect to $\tau_-(\bx,\bv')$ (equivalent to the directional differentiation $\bv'\cdot\nabla_\bx$) will allow us to reconstruct the quantity $2\sigma_a(\bx'+\tau_-(\bx,\bv')\bv')=2\sigma_a(\bx)$ a.e. . The coefficient $\Upsilon$ can then be reconstructed by solving the transport equation with the reconstructed $\sigma_a$ and $g_1$ (resp. $g_2$), and computing $\Upsilon=H_1/(\sigma_a\aver{u_1}_\bv)$ (resp. $\Upsilon=H_2/(\sigma_a\aver{u_2}_\bv)$).

\subsection{Reconstruction with point sources}
\label{SUBSEC:rec-point}

The second case where we can derive an analytical reconstruction method is when the illumination source is a point source in the spatial variable:
\begin{equation}\label{EQ:Point S}
	g(\bx,\bv)=\frak g(\bv)\delta(\bx-\bx').
\end{equation}
This is the most commonly used type of source in optical imaging, such as optical tomography~\cite{Arridge-IP99}. For any point $\bx\in\Omega$, let us define $\bv'=\dfrac{\bx-\bx'}{|\bx-\bx'|}$ as the unit vector pointing from $\bx'$ to $\bx$. Then the interior data $H$ simplifies to
\[
	H(\bx)=\frak g(\bv')\big(\Upsilon(\bx)\sigma_a(\bx)e^{-\int_0^{\tau_-(\bx,\bv')}\sigma_a(\bx-\tau_-(\bx,\bv')\bv'+s\bv')ds}\big) .
\]
Following the presentation in the previous section, we can rewrite this as
\begin{equation}\label{EQ:C}
	\dfrac{H(\bx'+\tau_-(\bx,\bv')\bv')}{\frak g(\bv')\Upsilon(\bx'+\tau_-(\bx,\bv')\bv')}=\sigma_a(\bx'+\tau_-(\bx,\bv')\bv')e^{-\int_0^{\tau_-(\bx,\bv')}\sigma_a(\bx'+s\bv')ds} .
\end{equation}
This is the same type of formula as ~\eqref{EQ:A} except that the unit direction vectors $\bv'$ are different at different points $\bx\in\Omega$.

\medskip
\noindent{\bf (i) Reconstruction of $\sigma_a$.} If $\Upsilon$ is known, we can differentiate~\eqref{EQ:C} as before to obtain an equation for the absorption coefficient $\sigma_a$. For each $\bv'\in\bbS^{d-1}_-(\bx')\equiv\{\bv\in\bbS^{d-1}|\bn(\bx')\cdot\bv<0\}$, define $\Omega_{\bv'} \equiv \left\{ \bx\in\Omega \, \left| \, \frac{\bx-\bx'}{|\bx-\bx'|} =\bv' \right. \right\}$, then
\begin{equation}
	\begin{array}{rcll}
  	\bv'\cdot\nabla \sigma_a(\bx) -\sigma_a^2(\bx) - 
        \left( \bv'\cdot\nabla \ln \dfrac{H(\bx)}{\Upsilon(\bx) \frak g(\bv')} \right) \sigma_a(\bx) &=& 0,
  	&\mbox{in}\ \ \Omega_{\bv'}  \\
       \sigma_a(\bx) &=& \sigma_a(\bx'), &\mbox{at}\ \ \bx' .
	\end{array}
\end{equation}
For a fixed $\bv'$, we can solve this first order ODE to find $\sigma_a(\bx)$ along the line segment $\Omega_{\bv'}$ if we know the value of $\sigma_a(\bx)$ at $\bx'$. Due to the singularity of the source at $\bx'$, we can not reconstruct this value from data $H$ as before. 

Point sources emit photons that travel in straight lines away from the source location. Thus the transport solution at each spatial position $\bx$ is only non-zero in direction $\bv'$. Hence $H=\Upsilon(\bx) \sigma_a(\bx) u(\bx,\bv')$. The transport equation can again be simplified to:
\begin{equation}
	\begin{array}{rcll}
  	\bv\cdot\nabla u(\bx,\bv) + \dfrac{H}{\Upsilon} &=& 0
  	&\mbox{ in }\ X\\
       u(\bx,\bv) &=& \frak g(\bv)\delta(\bx-\bx') &\mbox{on}\ \ \Gamma_{-} .
	\end{array}
\end{equation}
This transport equation can be conveniently solved in the polar coordinates with the origin at $\bx'$. The absorption coefficient can then be reconstructed as $\sigma_a=H/(\Upsilon u)$ away from $\bx'$.

\medskip
\noindent{\bf (ii) Reconstruction of ($\Upsilon$, $\sigma_a$).} If both coefficients are unknown, we can use two data sets generated with point sources to reconstruct them. The setup is depicted in Fig.~\ref{FIG:Sec Setup} (right). Let $g_1(\bx,\bv)=\frak g(\bv)\delta(\bx-\bx')$ and $g_2(\bx,\bv)=\frak g(\bv)\delta(\bx-\bx'')$ be the two point sources used to produce the interior data. We denote by $d(\bx',\bx'')=|\bx'-\bx''|$ the distance between the two points and by $\bar\bv=\dfrac{\bx''-\bx'}{|\bx''-\bx'|}$ the unit vector pointing from $\bx'$ to $\bx''$. For any point $\bx\in\Omega$, we define $\bv''=\dfrac{\bx-\bx''}{|\bx-\bx''|}$ as the unit vector pointing from $\bx''$ to $\bx$. It is then straightforward to verify that
\begin{equation}\label{EQ:Triangle}
	\tau_-(\bx,\bv')\bv'\cdot\bar\bv+\tau_-(\bx,\bv'')\bv''\cdot(-\bar\bv)=d(\bx',\bx'') .
\end{equation}
 
The first source $g_1$ produces data $H_1$ that is given in~\eqref{EQ:C}, while the second source $g_2$ produces the data $H_2$ that is given by
\[
	H_2(\bx)=\frak g(\bv'')\big(\Upsilon(\bx)\sigma_a(\bx)e^{-\int_0^{\tau_-(\bx,\bv'')}\sigma_a(\bx-\tau_-(\bx,\bv'')\bv''+s\bv'')ds}\big) .
\]
We can rewrite $H_2$ into
\begin{equation}\label{EQ:D}
	\dfrac{H_2(\bx''+\tau_-(\bx,\bv'')\bv'')}{\frak g(\bv'')\Upsilon(\bx''+\tau_-(\bx,\bv'')\bv'')}=\sigma_a(\bx''+\tau_-(\bx,\bv'')\bv'')e^{-\int_0^{\tau_-(\bx,\bv'')}\sigma_a(\bx''+s\bv'')ds}.
\end{equation}
Taking the ratio of~\eqref{EQ:D} and ~\eqref{EQ:C}, we obtain, after taking logarithm on both sides,
\begin{multline}
	\ln \left( \dfrac{H_2(\bx''+\tau_-(\bx,\bv'')\bv'')\frak g(\bv')\Upsilon(\bx'+\tau_-(\bx,\bv')\bv')}{H_1(\bx'+\tau_-(\bx,\bv')\bv')\frak g(\bv'')\Upsilon(\bx''+\tau_-(\bx,\bv'')\bv'')} \right)=\\
	\int_0^{\tau_-(\bx,\bv')}\sigma_a(\bx'+s\bv')ds -\int_0^{\tau_-(\bx,\bv'')} \sigma_a(\bx''+s\bv'')ds .
\end{multline}
We can now differentiate this result with respect to $\tau_-(\bx,\bv')$ (equivalent to the directional differentiation $\bv'\cdot\nabla_\bx$) and use relation~\eqref{EQ:Triangle} to obtain the quantity $\left( 1-\dfrac{\bv'\cdot\bar \bv}{\bv''\cdot\bar \bv} \right) \sigma_a(\bx)$. Note that for any point $\bx\in\Omega$ located on the line determined by $\bx'$ and $\bx''$, $\frac{\bv'\cdot\bar \bv}{\bv''\cdot\bar \bv}=-1$ (because $\bx$ locates between $\bx'$ and $\bx''$). The result in this case reduces to that in the case of two collimated sources: $(1-\frac{\bv'\cdot\bar \bv}{\bv''\cdot\bar \bv})\sigma_a(\bx) = 2\sigma_a(\bx)$. 

\subsection{Reconstruction with cone-limited sources}
\label{SUBSEC:rec-cone}

The reconstruction procedures in the previous sections work only when the illumination function $g$ takes the required special forms. For more general sources, we do not have similar explicit reconstruction procedures. In fact, it is not clear whether or not data from an arbitrary source would uniquely determine the absorption coefficient. We consider here a reconstruction method for a source that is slightly more general than the sources used in Sections \ref{SUBSEC:rec-collimated} and \ref{SUBSEC:rec-point} but still possess certain causality properties of those sources.

Let $\bv^0$ be a selected direction pointing inside the domain $\Omega$. We intend to construct a source such that the transport equation~\eqref{EQ:ERT Free} with the source is casual along direction $\bv^0$. Such causality would allow us to derive a direct layer peeling method that solves the inverse problem in one pass in a non-iterative manner.

\begin{figure}[ht]
\centering
\includegraphics[width=0.55\textwidth]{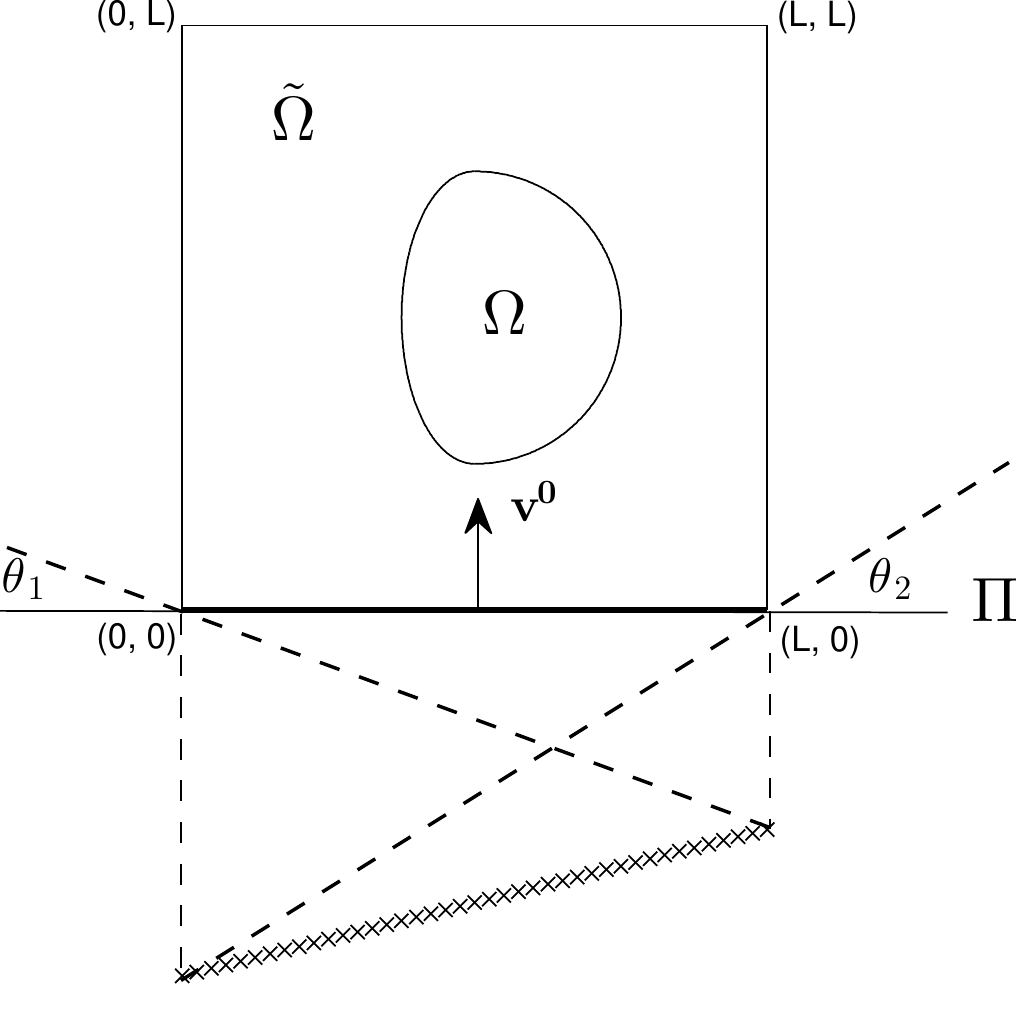}\\
\caption{Domain $\Omega$ and its extension to a hyper-cube $\tilde\Omega$. The physical sources located on a line marked with $\times$ are separated from $\Omega$ by a hyper-plane $\Pi$. The boundary condition is non-zero only on the 
bottom side of $\tilde\Omega$ (thick black part) that is orthogonal to $\bv^0$. The half-aperture $\theta^0$ of the cone $V(\bv^0, \theta^0)$ 
is given by $\theta^0 = \pi/2 - \min \{ \theta_1, \theta_2 \}$.}
\label{FIG:cone-square}
\end{figure}
We require the source $g(\bx,\bv)$ satisfy the following condition:
\begin{equation}\label{EQN:cone-source}
	g(\bx,\bv) = 0 \quad \mbox{for any} \quad \bv \notin V(\bv^0, \theta^0) = 
\{ \bv \in \mathbb{S}^{d-1} \;|\; \bv \cdot \bv^0 \geq \cos \theta^0 \},
\end{equation}
where the cone $V(\bv^0, \theta^0)$ is determined by the selected direction $\bv^0 \in \mathbb{S}^{d-1}$ and the half-aperture $0 < \theta^0 < \pi/2$. A practically important case that satisfies the assumption~\eqref{EQN:cone-source} is when the physical source can be separated from the object occupying the domain $\Omega$ by a hyper-plane. This limits the angle at which $\Omega$ is visible from the source, thus the rays entering $\Omega$ are limited to some cone of directions $V(\bv^0, \theta^0)$ with $\bv^0$ orthogonal to the separating hyper-plane. Then the resulting effective boundary conditions satisfy~\eqref{EQN:cone-source}. In fact, in such case we can extend $\Omega$ to a $d$-dimensional hyper-cube $\tilde\Omega$ such that $g(\bx,\bv)$ is non-zero only on one side of the cube which is orthogonal to $\bv^0$, as shown in Fig. \ref{FIG:cone-square}.

To simplify the presentation of the method we consider here the case $d=2$, although the algorithm remains essentially
the same for $d=3$. We assume that $\tilde\Omega$ is the square $[0,L]^2$ with $\sigma_a (\bx) = 0$ 
for $\bx \in \tilde\Omega \setminus \Omega$. We also assume that the rays enter the domain from the bottom side 
$\{(x,y) \; | \; x \in (0,L), y=0\}$, so $\bv^0 = (0,1)$. The source $g(\bx,\bv)$ is only non-zero on the bottom side. 
Thus the equation~\eqref{EQ:ERT Free} is causal in the $y$ direction; see Fig. \ref{FIG:cone-square} for the details of the setup. We use the parametrization $\bv = (\cos \theta, \sin \theta)$ so that $\bv \in V(\bv^0, \theta^0)$ becomes 
$\theta \in [\pi/2 - \theta^0, \pi/2 + \theta^0]$. This implies that we can divide~\eqref{EQ:ERT Free} by 
$\sin \theta$ $(> 0)$ so that the free transport equation together with its boundary condition can be written as
\begin{equation}
\begin{array}{rcl}\label{EQN:ERT-cone}
u_y (x,y,\theta) & = & - \cot\theta \; u_x (x,y,\theta) 
- \csc\theta \; \sigma_a(x,y) \; u(x,y,\theta),\\
u(x, 0, \theta) & = & g(x, \theta), \quad x\in(0,L),\ \theta \in \left[ \frac{\pi}{2}-\theta^0, \frac{\pi}{2}+\theta^0 \right] \\
u(0, y, \theta) & = & 0, \quad y\in(0,L),\  \theta \in \left[ \frac{\pi}{2}-\theta^0, \frac{\pi}{2} \right) \\
u(L, y, \theta) & = & 0, \quad y\in(0,L),\  \theta \in \left( \frac{\pi}{2}, \frac{\pi}{2}+\theta^0 \right]
\end{array}
\end{equation}
where we observe that no boundary condition is needed on the top side of the domain $\tilde\Omega$, $\{(x,y) \; | \; x\in(0,L),y = L \}$. The interior data $H(x,y)$ now takes the form
\[
	H(x,y) = \Upsilon(x,y) \sigma_a(x,y) \int\limits_{\pi/2-\theta^0}^{\pi/2+\theta^0} u(x,y,\theta) d \theta \equiv \Upsilon(x,y) \sigma_a(x,y) \aver{u(x,y,\theta)}_{\theta^0}.
\]

The role of the temporal variable in (\ref{EQN:ERT-cone}) is played by $y$ (the system is causal in $y$) so we can 
apply a first order time stepping procedure to obtain a semi-discrete inversion scheme. Let us discretize (\ref{EQN:ERT-cone}) in $y$ on a grid with nodes $y_j$, $j=0,\ldots,N_y$, where $y_0 = 0$, $y_{N_y}=L$ and the grid steps are $h_k = y_k - y_{k-1}$, $k=1,\ldots,N_y$. For the semi-discrete quantities we use notation
\begin{equation}
u^{(j)} (x, \theta) \approx u(x, y_j, \theta), 
\quad \sigma_a^{(j)}(x) \approx \sigma_a (x, y_j),
\quad j=0,\ldots,N_y.
\label{EQN:semidiscrete}
\end{equation}
To reconstruct $\sigma_a$ with $\Upsilon$ known, we apply a forward Euler time stepping scheme to (\ref{EQN:ERT-cone}) to obtain the following explicit reconstruction procedure.
\begin{enumerate}
\item Initialize $u^{(0)}(x, \theta) = g(x, \theta)$, and 
\[
\sigma_a^{(0)}(x) = \frac{H(x,0)}{\Upsilon(x,0) \left< g(x,\theta) \right>_{\theta^0}}, \quad x \in [0,L].
\]
\item For $j=1,\ldots,N_y$ compute 
\begin{equation}
u^{(j)}(x, \theta) = \left( \cI - h_j \cot\theta\; \partial_x - 
h_j \csc\theta \; \sigma_a^{(j-1)}(x)\right) u^{(j-1)}(x, \theta)
\label{EQN:fwd-euler}
\end{equation}
for $x \in (0,L)$ and $\theta \in \left[ \pi/2-\theta^0, \pi/2+\theta^0 \right]$. At the boundary $x \in \{0,L\}$ set
\begin{equation} 
\begin{array}{rcl}
u^{(j)}(0,\theta) & = & 0, \quad \theta \in \left[ \frac{\pi}{2}-\theta^0, \frac{\pi}{2} \right) \\
u^{(j)}(L,\theta) & = & 0, \quad \theta \in \left( \frac{\pi}{2}, \frac{\pi}{2}+\theta^0 \right]
\end{array}
\label{EQN:uxbc}
\end{equation}
and compute the reconstruction
\begin{equation}
\sigma_a^{(j)}(x) = \frac{H(x,y_j)}{ \Upsilon(x,y_j) \left< u^{(j)}(x,\theta) \right>_{\theta^0} }, \quad x \in [0,L].
\label{EQN:1coeff-sigma}
\end{equation}
\end{enumerate}

The above method is quite stable in practice as demonstrated in the numerical examples in Section~\ref{SUBSEC:num-nonscatter}. However, the discretization of $\partial_x$ in ~\eqref{EQN:fwd-euler} should be compatible with the grid refinement in $y$ 
to retain stability of the forward Euler. If a coarser discretization in $y$ compared to that in $x$ is desired, then~\eqref{EQN:fwd-euler} can be replaced by a semi-implicit scheme
\[
\left( \cI + h_j \cot\theta \; \partial_x \right) u^{(j)}(x, \theta) =  
\left( \cI - h_j \csc\theta \; \sigma_a^{(j-1)}(x)\right) u^{(j-1)}(x, \theta),
\]
which requires solving a first order differential (difference if discretized) equation in $x$ at every step 
$j$ with boundary conditions (\ref{EQN:uxbc}). Generalizing this method to reconstruct simultaneously two coefficients $\Upsilon$ and $\sigma_a$ remains a topic of future study.

\medskip
\medskip
We conclude Section~\ref{SEC:Non-scattering} by summarizing the reconstruction results for non-scattering media that we have introduced into the following uniqueness and stability theorem.
\begin{theorem}
	Let $(\Upsilon,\sigma_a)$ and $(\tilde\Upsilon,\tilde\sigma_a)$ be two sets of coefficients in~\eqref{EQ:Data} and ~\eqref{EQ:ERT Free} satisfying the assumptions in {\bf (A1)}, and $g_1(\bx,\bv)$ and $g_2(\bx,\bv)$ be two source functions of the form ~\eqref{EQ:Collimated S} or ~\eqref{EQ:Point S}. Let $(H_1, H_2)$ and $(\tilde H_1, \tilde H_2)$ be the corresponding interior data. Then (a) $H_1=\tilde H_1$ implies $\sigma_a(\bx)=\tilde\sigma_a(\bx)$ if $\Upsilon=\tilde\Upsilon$; and (b) $(H_1,H_2)=(\tilde H_1,\tilde H_2)$ implies $(\Upsilon,\sigma_a)=(\tilde\Upsilon,\tilde\sigma_a)$. Moreover, the following stability result holds:
\begin{multline}\label{EQ:Stab1}	\|\Upsilon(\bx)\sigma_a(\bx)e^{-\int_0^{\tau_-(\bx,\bv')}\sigma_a(\bx-\tau_-(\bx,\bv')\bv'+s\bv')ds} 
	- \tilde\Upsilon(\bx)\tilde\sigma_a(\bx)e^{-\int_0^{\tau_-(\bx,\bv')}\tilde\sigma_a(\bx-\tau_-(\bx,\bv')\bv'+s\bv')ds}\|_{L^\infty(\Omega)}\\
	\le C_l \|H_l-\tilde H_l\|_{L^\infty(\Omega)},\ \ l=1,2
\end{multline}
$C_l$ being a constant that depends on $g_l$ but is independent of the data $H_l$ and $\tilde H_l$.	
\end{theorem}

\section{QPAT of scattering media}
\label{SEC:Scattering}

We now consider the QPAT problem for scattering media. When the scattering effect is very weak, the results obtained in the previous section could be used to obtain a good approximation of the reconstruction. We thus assume that the scattering effects are sufficiently strong so that neglecting them would deteriorate significantly the quality of the reconstructions. On the other hand, we assume that the scattering effect is not strong enough for us to model the light propagation process with the diffusion model for which explicit reconstruction strategies have been proposed~\cite{BaRe-IP11,BaRe-IP12,BaUh-IP10}.

\subsection{Stability and uniqueness results}
\label{SUBSEC:Uniq}

Unlike the non-scattering regime where we have unique and stable reconstruction using only a small number  of interior data sets, in scattering media, we have only stability and uniqueness results from information of the full albedo operator:
\begin{equation}\label{EQ:Albedo}
	\cA:
	\begin{array}{lcl} 
		L^1(\Gamma_-,d\bxi) &\rightarrow& L^1(\Omega)\\
		g(\bx,\bv) & \mapsto & \cA g=H(\bx)
	\end{array} .
\end{equation}
It is the clear that $\cA$ depends on all three coefficients $\Upsilon$, $\sigma_a$ and $\sigma_s$. We denote by $\|\cA\|_{\cL(L^1(\Gamma_-,d\bxi); L^1(\Omega))}$ the norm of $\cA$ as a linear operator from $L^1(\Gamma_-,d\bxi)$ to $L^1(\Omega)$.

\medskip
We then have the following stability estimate following the analysis in~\cite[Theorem 2.3]{BaJoJu-IP10}.
\begin{theorem}\label{Thm:Stab} 
	Let $(\Upsilon,\sigma_a,\sigma_s)$ and $(\tilde\Upsilon,\tilde\sigma_a,\tilde\sigma_s)$ be two set of coefficients satisfying the regularity assumptions in {\bf (A1)}, and $\cA$, $\tilde\cA$ the corresponding albedo operators. Denote by $\sigma=\sigma_a+\sigma_s$ and $\tilde\sigma=\tilde\sigma_a+\tilde\sigma_s$. Then the following result holds
	\begin{multline}\label{EQ:Stab}
		\int_0^{\tau_+(\bx',\bv')}|\Upsilon(\bx'+t\bv')\sigma_a(\bx'+t\bv')e^{-\int_0^t\sigma(\bx'+s\bv')ds} 
				- \tilde\Upsilon(\bx'+t\bv')\tilde\sigma_a(\bx'+t\bv')e^{-\int_0^t\tilde\sigma(\bx'+s\bv')ds}| dt \\ \le \|\cA-\tilde\cA\|_{\cL(L^1(\Gamma_-,d\bxi); L^1(\Omega))}
	\end{multline}	
	for any point $(\bx',\bv')\in\Gamma_-$.
\end{theorem}
The result is a slight modification of the result in~\cite{BaJoJu-IP10} where the singular part of the Schwartz kernel of the albedo operator is analyzed in detail. The only change that we have to make is to include the Gr\"uneisen coefficient $\Upsilon$ in the final step of the analysis. We thus would not repeat the lengthy analysis here but refer to~\cite{BaJoJu-IP10}.

It turns out that the stability result~\eqref{EQ:Stab} leads to the unique reconstruction of the coefficients in some simplified situations. More precisely, if one of the three coefficients is known, we can reconstruct the other two coefficients uniquely as summarized in the following corollary.
\begin{corollary}\label{THM:Uniq}
	Let $(\Upsilon,\sigma_a,\sigma_s)$ and $(\tilde\Upsilon,\tilde\sigma_a,\tilde\sigma_s)$ be two set of coefficients satisfying the regularity assumptions in {\bf (A1)}. Then the following statements hold: (a) If $\Upsilon=\tilde\Upsilon$, then $\cA=\tilde\cA$ implies $(\sigma_a,\sigma_s)=(\tilde\sigma_a,\tilde\sigma_s)$; (b) If $\sigma_s=\tilde\sigma_s$, then $\cA=\tilde\cA$ implies $(\Upsilon, \sigma_a)= (\tilde\Upsilon, \tilde\sigma_a)$; (c) If $\sigma_a=\tilde\sigma_a$, then $\cA=\tilde\cA$ implies $(\Upsilon, \sigma_s)= (\tilde\Upsilon, \tilde\sigma_s)$.
\end{corollary}
\begin{proof}
	Theorem~\ref{THM:Uniq} implies that we can reconstruct uniquely
	\begin{equation}
		\cM(\bx',\bv',t)=\Upsilon(\bx'+t\bv')\sigma_a(\bx'+t\bv')e^{-\int_0^t\sigma(\bx'+s\bv')ds}, 
			\quad (\bx',\bv')\in\Gamma_- .
	\end{equation}
	We can reconstruct the same quantity for the a different point on $(\bx'+\tau_+(\bx',\bv')\bv',-\bv')$ on $\Gamma_-$ with $t$ replaced by $\tau_+(\bx',\bv')-t$: $\cM(\bx'+\tau_+(\bx',\bv')\bv',-\bv',\tau_+(\bx',\bv')-t)$. Taking the log of the ratio of the two quantites, we obtain
	\begin{multline}
		\ln \frac{\cM(\bx'+\tau_+(\bx',\bv')\bv',-\bv',\tau_+(\bx',\bv')-t)}{\cM(\bx',\bv',t)}=\\
		-\int_0^{\tau_+(\bx',\bv')-t}\sigma(\bx'+\tau_+(\bx',\bv')-s\bv')ds+\int_0^t\sigma(\bx'+s\bv')ds. 
	\end{multline}
	 Taking the derivative of the above quantity gives us $2\sigma(\bx'+t\bv')$ for a.e. $(\bx',\bv')\in\Gamma_-$ and $t>0$. This in turn gives us $\sigma(\bx)$ for a.e. $\bx\in\Omega$. Once $\sigma(\bx)$ is uniquely determined, we reconstruct $\Upsilon\sigma_a\equiv \mu$ uniquely from $\cM(\bx',\bv',t)$. (a) If $\Upsilon$ is known already, then $\sigma_a$ is known from $\mu$, and thus $\sigma-\sigma_a$ would give us $\sigma_s$; (b) If $\sigma_s$ is known, $\sigma-\sigma_s$ would give us $\sigma_a$. $\Upsilon$ is then known from $\mu$; (c) If $\sigma_a$ is known, then $\mu$ would give us $\Upsilon$ and $\sigma$ would give us $\sigma_s$. This completes the proof.
\end{proof}
\begin{remark}\rm
	Case (a) was presented in~\cite{BaJoJu-IP10}. The construction of $\cM(\bx',\bv',t)$ and $\cM(\bx'+\tau_+(\bx',\bv')\bv',-\bv',\tau_+(\bx',\bv')-t)$ is in the same spirit as the construction of ~\eqref{EQ:A} and ~\eqref{EQ:B} (or ~\eqref{EQ:C} and ~\eqref{EQ:D}) using collimated (or point) sources in the previous section.
\end{remark}

The above stability and uniqueness results, however, do not guarantee the unique reconstruction of all three coefficients simultaneously. In fact, theory based on the diffusion model in~\cite{BaRe-IP11} states that one can not simultaneously reconstruct all three coefficient uniquely unless additional information is available~\cite{BaRe-IP12}. Moreover, unlike the situation in the non-scattering regime, no explicit reconstruction method can be derived in the scattering media. We thus have to rely mainly on other computational methods for the reconstructions. We now derive a numerical reconstruction method based on the Born approximation of the nonlinear inverse problem.

\subsection{Linearized reconstruction with Born approximation}
\label{SUBSEC:Born}

We linearize around some \emph{known}, not necessarily constant, background optical properties $\sigma_{a0}(\bx)$ and $\sigma_{s0}(\bx)$. To be more precise, we assume
\begin{equation}\label{EQ:Pert Sigma}
	\sigma_a(\bx)=\sigma_{a0}(\bx)+\tilde\sigma_a(\bx), \qquad 
	\sigma_s(\bx)=\sigma_{s0}(\bx)+\tilde\sigma_s(\bx),
\end{equation}
where the perturbations are small in the sense that $\left \| \dfrac{\tilde\sigma_a(\bx)}{\sigma_{a0}} \right \|_{L^\infty(\Omega)}\ll 1$ and $\left\| \dfrac{\tilde\sigma_s(\bx)}{\sigma_{s0}} \right\|_{L^\infty(\Omega)}\ll 1$. Note again that $\Upsilon(\bx)$ is not assumed in perturbative form since the inverse problem of reconstructing $\Upsilon$ is linear (because the data $H$ is linearly proportional to the unknown $\Upsilon$).

The solution of the radiative transport problem with coefficients $\sigma_a$ and $\sigma_s$ can then be written as
\begin{equation}\label{EQ:Pert U}
	u(\bx,\bv)=U(\bx,\bv)+\tilde u(\bx,\bv),
\end{equation}
where $U(\bx,\bv)$ is the solution of the transport equation~\eqref{EQ:ERT} with the known background coefficients $\sigma_{a0}$ and $\sigma_{s0}$, and $\tilde u(\bx,\bv)$ is the perturbation in the solution caused by the perturbations $\tilde\sigma_a(\bx)$ and $\tilde\sigma_s(\bx)$ in the coefficient. The equation satisfied by the perturbation $\tilde u(\bx,\bv)$, to the first order, is 
\begin{equation}\label{EQ:ERT Pert}
  \begin{array}{rcll}
  	\bv\cdot\nabla \tilde u(\bx,\bv) 
  	+ \sigma_{a0}(\bx)\tilde u &=&
  	\sigma_{s0}(\bx)K(\tilde u) -\tilde\sigma_a(\bx) U + \tilde\sigma_s(\bx)K(U) &\mbox{ in }\ X\\
    \tilde u(\bx,\bv) &=& 0 &\mbox{ on }\ \Gamma_{-}.
  \end{array}
\end{equation}
It is can be shown, using the fact that $u$ is Fr\'echet differentiable with respective to $\sigma_a$ and $\sigma_s$~\cite{Dorn-IP98,Dorn-OE00,ReBaHi-SIAM06}, that the terms omitted are indeed high order terms.

We now introduce the (adjoint) Green's function $G(\bx,\bv;\by)$ for the transport problem with background optical properties, the solution of the following adjoint transport equation 
\begin{equation}\label{EQ:Green}
  \begin{array}{rcll}
  	-\bv\cdot\nabla G(\bx,\bv;\by) + 
  	\sigma_{a0}(\bx) G(\bx,\bv;\by)&=&
  	\sigma_{s0}(\bx)K(G)(\bx,\bv;\by) - \delta(\bx-\by)&\mbox{ in }\ X\\
    G(\bx,\bv;\by) &=& 0 &\mbox{ on }\ \Gamma_{+},
  \end{array}
\end{equation}
where $\delta$ is the usual Dirac delta function. Note that since the transport operator is not self-adjoint, the boundary condition is now put on $\Gamma_{+}$. Rigorously speaking, this equation only holds in the weak sense. We need to multiply it with a test function that is $\cC_+$. That requires that the coefficients $\sigma_{a0}$ and $\sigma_{s0}$ are at least continuous in $\bar\Omega$.

We can now multiply~\eqref{EQ:ERT Pert} by $\Upsilon G$ and integrate over $X$, and multiply~\eqref{EQ:Green} by $\Upsilon \tilde u$ and integrate over $X$, and subtract the results to show that
\begin{multline}\label{EQ:LIN}
	\Upsilon(\bx)\aver{\tilde u(\bx,\bv)}_\bv 
	=\int_\Omega \Upsilon(\by)\tilde\sigma_a(\by)
	\aver{U(\by,\bv) G(\by,\bv;\bx)}_\bv d\by \\
	-\int_\Omega \Upsilon(\by)\tilde\sigma_s(\by) \aver{K(U)(\by,\bv) G(\by,\bv;\bx)}_\bv d\by.
\end{multline}

The perturbations in~\eqref{EQ:Pert Sigma} and~\eqref{EQ:Pert U} imply that the interior data $H$ is now given, to the first order, by 
\begin{equation}\label{EQ:Data Pert}
	H=\Upsilon\sigma_{a0}\aver{U}_\bv + \Upsilon \tilde\sigma_a \aver{U}_\bv 
		+\Upsilon \sigma_{a0} \aver{\tilde u}_\bv .
\end{equation}
Combining~\eqref{EQ:Data Pert} and~\eqref{EQ:LIN}, we can show that 
\begin{equation}\label{EQ:LIN INTGRL}
	\dfrac{H}{\sigma_{a0}\aver{U}_\bv}(\bx)=\cI(\Upsilon)+\cL^a(\Upsilon\tilde\sigma_a)+\cL^s(\Upsilon\tilde\sigma_s),
\end{equation}
where $\cI$ is the identity operator and the operators $\cL^a$ and $\cL^s$ are defined, respectively, as
\begin{eqnarray}
	\label{EQ:La}
	\cL^a(\Upsilon\tilde\sigma_a) &=& \dfrac{\Upsilon \tilde\sigma_a}{\sigma_{a0}}  
	      +\int_\Omega \Upsilon(\by)\tilde\sigma_a(\by)
		\dfrac{\aver{U G}_\bv}{\aver{U}_\bv}(\by;\bx) d\by\\
	\label{EQ:Ls}
	\cL^s(\Upsilon\tilde\sigma_s) &=& -\int_\Omega \Upsilon(\by)\tilde\sigma_s(\by) 
		\dfrac{\aver{K(U)G}_\bv}{\aver{U}_\bv}(\by;\bx) d\by .	 
\end{eqnarray}
Here we have normalized the data by $\sigma_{a0}\aver{U}_\bv$. This can be done when $\aver{U}_\bv(\bx)\neq 0$ which can be guaranteed by selecting appropriate illumination source $g(\bx,\bv)$. The normalization makes the sizes of the kernels of the integral operators $\cL^a$ and $\cL^s$ on the same order at all locations. This is important due to the fact that the strength of optical signals usually varies over several orders of magnitude across the domain of interest. The normalization results in well-balanced matrix elements when the integral equation~\eqref{EQ:LIN INTGRL} is discretized for numerical solution.

Equation~\eqref{EQ:LIN INTGRL} is a linear integral equation for the three variables $\Upsilon$, $\Upsilon\tilde\sigma_a$ and $\Upsilon\tilde\sigma_s$. The kernels for the operator $\cL^a$ and that for $\cL^s$, $\aver{U G}_\bv/\aver{U}_\bv$ and $\aver{K(U)G}_\bv/\aver{U}_\bv$ are known since they only involve the solutions of the forward and adjoint transport equations with background optical properties $\sigma_{a0}$ and $\sigma_{s0}$. It remains to solve~\eqref{EQ:LIN INTGRL} to reconstruct the unknowns ($\Upsilon$, $\Upsilon\tilde\sigma_a$ and $\Upsilon\tilde\sigma_s$), and then the real coefficients ($\Upsilon$, $\tilde\sigma_a$ and $\tilde\sigma_s$).

In practice, we do not have the full albedo data, but only a finite number of data sets generated from different sources. Let us assume that there are $N_g$ data sets collected from $N_g$ illuminations. The data sets are denoted by $\{H_i\}_{i=1}^{N_g}$ and the corresponding illuminations are denoted by $\{g_i\}_{i=1}^{N_g}$. The system of linear integral equation with ${N_g}$ data sets can be written as
\begin{equation}\label{EQ:LIN INTGRL Sys}
	\cL
	\left(\begin{array}{c}
					\Upsilon\\
					\Upsilon\tilde\sigma_a\\
					\Upsilon\tilde\sigma_s
				\end{array}
	\right)
	\equiv
	\left(\begin{array}{ccc}
					\cI & \cL^a_1 & \cL^s_1\\
					\vdots  & \vdots & \vdots\\
					\cI & \cL^a_i & \cL^s_i\\
					\vdots  & \vdots & \vdots\\
					\cI & \cL^a_{N_g} & \cL^s_{N_g}
				\end{array}
	\right)
	\left(\begin{array}{c}
					\Upsilon\\
					\Upsilon\tilde\sigma_a\\
					\Upsilon\tilde\sigma_s
				\end{array}
	\right)
	=\dfrac{1}{\sigma_{a0}}
	\left(\begin{array}{c}
					H_1/\aver{U_1}_\bv\\
					\vdots\\
					H_i/\aver{U_i}_\bv\\
					\vdots\\
					H_{N_g}/\aver{U_{N_g}}_\bv
				\end{array}
	\right)
	\equiv
	\left(\begin{array}{c}
					h_1\\
					\vdots\\
					h_i\\
					\vdots\\
					h_{N_g}
				\end{array}
	\right)
\end{equation}
where $\cL^a_i$ and $\cL^s_i$ are the same operators defined in~\eqref{EQ:La} and ~\eqref{EQ:Ls} respectively, with $U$ replaced with $U_i$. 

The integral formulation can be discretized to get a linear system of equations. Let us assume that we have a numerical procedure, say a quadrature rule, to discretize the integral equation~\eqref{EQ:LIN}, and assume that we discretize $\tilde\sigma_a(\bx)$ on a mesh of $N_\Omega$ nodes, $\{\by_k\}_{k=1}^{N_\Omega}$. Then $\cL$ will be matrix consists of ${N_g}\times 3$ blocks, each having size $N_\Omega\times N_\Omega$. The $lk$ elements of the matrices are given by
\begin{equation}
	(\cL^a_i)_{lk} = \dfrac{\delta_{lk}}{\sigma_{a0}(\bx_l)}+\xi_k\dfrac{\aver{U_i(\by_k,\bv) G(\by_k,\bv;\bx_l)}_\bv}{\aver{U_i(\bx_l,\bv)}_\bv},\quad
	(\cL^s_i)_{lk} = -\xi_k\dfrac{\aver{K(U_i)(\by_k,\bv) G(\by_k,\bv;\bx_l)}_\bv}{\aver{U_i(\bx_l,\bv)}_\bv},
\end{equation}
where $\xi_k$ ($1\le k\le N_\Omega$) 
is the weight of the quadrature on the $k$-th element. For simplicity of notation, we will not differentiate from now on integral operators and their discrete equivalences, the matrices.

The system~\eqref{EQ:LIN INTGRL Sys} can be over- or under-determined and so is often solved in regularized least-square sense. More specifically, we solve the problem by solve the following minimization problem
\begin{equation}\label{EQ:Min}
\min_{\Upsilon,\Upsilon\tilde\sigma_a,\Upsilon\tilde\sigma_s}
\|\cL
\begin{pmatrix}
\Upsilon\\
\Upsilon\tilde\sigma_a\\
\Upsilon\tilde\sigma_s
\end{pmatrix}				
- \begin{pmatrix}
h_1\\ \vdots\\ h_{N_g}
\end{pmatrix} \|_{l^2}^2 
+ \rho\|\cD
\begin{pmatrix}
\Upsilon\\
\Upsilon\tilde\sigma_a\\
\Upsilon\tilde\sigma_s
\end{pmatrix} \|_{l^2}^2
\end{equation}
where the first term is the data fidelity term and the second term is a Tikhonov regularization term with the strength of regularization given by $\rho$. The operator $\cD$ is the discrete differentiation. The regularization term is needed due to the presence of noise in practice even though the inverse problem here is very stable compared to similar inverse transport problems in diffuse optical tomography~\cite{Arridge-IP99,Bal-IP09,Ren-CiCP10}.

\begin{remark}\rm
The $l^2$ least squares formulation and Tikhonov regularization strategy~\eqref{EQ:Min} are adopted here mainly for their computational efficiency. One can employ different data fidelity and regularization terms such as those based on total variation (TV) and $l^1$ norms~\cite{GaZh-OE10A,GaZh-OE10B}. Different selection can be effective for different problems. We refer interested reader to~\cite{GaZh-OE10A,GaZh-OE10B} for detail numerical analysis and comparison of performances of different approaches. The primary focus here is the properties of the inverse transport problems in QPAT, not the details of the numerical implementation.
\end{remark}

It is important to notice that the norm of the operator $\cL^s$ is in general small compared to that of $\cL^a$ and the identity operator $\cI$, especially when the transport process is diffusive and the scattering is very isotropic. In those cases, $U(\bx,\bv)$ is roughly independent of $\bv$ so that $\|K(U)\|_{L^1(X)}\ll 1$. Thus $\aver{K(U)G}_\bv$ is very small. In this case, when noise is present in the measured data, the reconstruction of the scattering coefficient $\sigma_s$ is easily corrupted by noise. That was observed in the past in optical tomography~\cite{Arridge-IP99,ArSc-IP09,Ren-CiCP10}. We thus separate the reconstruction into the following two cases.

\subsubsection{Reconstructing ($\Upsilon$, $\sigma_a$)}

To reconstruct $\tilde\sigma_a$ and $\Upsilon$ assuming $\sigma_s$ is known, we solve the simplified least-squares problem
\begin{equation}\label{EQ:Min GammaMua}
	\min_{\Upsilon,\Upsilon\tilde\sigma_a}
	\|\left(\begin{array}{cc}
					\cI & \cL_1^a \\
					\vdots & \vdots \\
					\cI & \cL_{N_g}^{a}
				\end{array}
	\right)
	\left(\begin{array}{c}
					\Upsilon\\
					\Upsilon\tilde\sigma_a
				\end{array}
	\right)
	-
	\left(\begin{array}{c}
					h_1\\
					\vdots\\
					h_{N_g}
				\end{array}
	\right) \|_{l^2}^2 +\rho\|\cD\left(\begin{array}{c}
					\Upsilon\\
					\Upsilon\tilde\sigma_a
				\end{array}
	\right)\|_{l^2}^2 .
\end{equation}
We recover from the solution the coefficients $\Upsilon$ and $\tilde\sigma_a$. The solution of the least-squares problem is simply
\begin{equation}\label{EQ:UpsilonMua Sol}
	\left(\begin{array}{c}
					\Upsilon\\
					\Upsilon\tilde\sigma_a
				\end{array}
	\right)
	=
	\Big[\left(\begin{array}{cc}
					{N_g}\cI & \sum_{i=1}^{N_g}\cL_i^a \\
					\sum_{i=1}^{N_g}\cL_i^{a*} & \sum_{i=1}^{N_g}\cL_i^{a*}\cL_i^a
				\end{array}
	\right)+\rho\cD^*\cD\Big]^{-1}
	\left(\begin{array}{ccc}
					\cI & \ldots & \cI \\
					\cL_1^{a*} & \ldots & \cL_{N_g}^{a*}
				\end{array}
	\right)
	\left(\begin{array}{c}
					h_1\\
					\vdots\\
					h_{N_g}
				\end{array}
	\right), 
\end{equation}
where the superscript $*$ is used to denote the adjoint of an operator, and $\cI^*=\cI$. The inversion of the regularized normal operator is achieved by a linear solver, not direct inversion.

\subsubsection{Reconstructing ($\sigma_a$, $\sigma_s$) or ($\Upsilon$, $\sigma_s$)}

To reconstruct either ($\sigma_a$, $\sigma_s$) or ($\Upsilon$, $\sigma_s$), we need to deal with an unbalanced system of equations caused by the smallness of the scattering component $\cL^s$. We consider here only the case of reconstructing ($\sigma_a$, $\sigma_s$). To obtain a similar algorithm to reconstruct ($\Upsilon$, $\sigma_s$), we only need to replace the operator $\cL^a$ in the algorithm by $\cI$.
We solve the simplified least-squares problem
\begin{equation}\label{EQ:Min MuaMus}
	\min_{\Upsilon\tilde\sigma_a,\Upsilon\tilde\sigma_s}
	\|\left(\begin{array}{cc}
					\cL_1^a & \cL_1^s \\
					\vdots & \vdots \\
					\cL_{N_g}^a & \cL_{N_g}^s
				\end{array}
	\right)
	\left(\begin{array}{c}
					\Upsilon\tilde\sigma_a\\
					\Upsilon\tilde\sigma_s
				\end{array}
	\right)
	-
	\left(\begin{array}{c}
					h_1\\
					\vdots\\
					h_{N_g}
				\end{array}
	\right) \|_{l^2}^2 +\rho\|\cD\left(\begin{array}{c}
					\Upsilon\tilde\sigma_a\\
					\Upsilon\tilde\sigma_s
				\end{array}
	\right)\|_{l^2}^2 .
\end{equation}
The normal equation for this minimization problem is
\begin{equation}
	\left(\begin{array}{cc}
					\sum_{i=1}^{N_g}\cL_i^{a*}\cL_i^a +\rho\cD_a^*\cD_a & \sum_{i=1}^{N_g}\cL_i^{a*}\cL_i^s \\
					\sum_{i=1}^{N_g}\cL_i^{s*}\cL_i^a & \sum_{i=1}^{N_g}\cL_i^{s*}\cL_i^s+\rho\cD_s^*\cD_s 
				\end{array}
	\right)
	\left(\begin{array}{c}
					\Upsilon\tilde\sigma_a\\
					\Upsilon\tilde\sigma_s
				\end{array}
	\right)
	=
	\left(\begin{array}{c}
					\sum_{i=1}^{N_g}\cL_i^{a*}h_i\\
					\sum_{i=1}^{N_g}\cL_i^{s*}h_i
				\end{array}
	\right), 
\end{equation}
where $\cD_a$ and $\cD_s$ are such that $\cD=\diag (\cD_a,\cD_s)$. Instead of solving this normal equation directly to get a solution similar to~\eqref{EQ:UpsilonMua Sol}, we perform Gauss elimination of the system to obtain
\begin{equation}
	\left(\begin{array}{cc}
					\sum_{i=1}^{N_g}\cL_i^{a*}\cL_i^a +\rho\cD_a^*\cD_a & \sum_{i=1}^{N_g}\cL_i^{a*}\cL_i^s \\
					0 & \cL^{red} 
				\end{array}
	\right)
	\left(\begin{array}{c}
					\Upsilon\tilde\sigma_a\\
					\Upsilon\tilde\sigma_s
				\end{array}
	\right)
	=
	\left(\begin{array}{c}
					\sum_{i=1}^{N_g}\cL_i^{a*}h_i\\
					h^{red}
				\end{array}
	\right), 
\end{equation}
where
\begin{equation*}
	\cL^{red}=\sum_{i=1}^{N_g}\cL_i^{s*}\cL_i^s+\rho\cD_s^*\cD_s -\sum_{i=1}^{N_g}\cL_i^{s*}\cL_i^a\big(\sum_{i=1}^{N_g}\cL_i^{a*}\cL_i^a +\rho\cD_a^*\cD_a\big)^{-1}\sum_{i=1}^{N_g}\cL_i^{a*}\cL_i^s,
\end{equation*}
and
\begin{equation*}
	h^{red}=\sum_{i=1}^{N_g}\cL_i^{s*}h_i -\Big[\sum_{i=1}^{N_g}\cL_i^{s*}\cL_i^a\big(\sum_{i=1}^{N_g}\cL_i^{a*}\cL_i^a +\rho\cD_a^*\cD_a\big)^{-1}\sum_{i=1}^{N_g}\cL_i^{a*}\cL_i^s\Big]\sum_{i=1}^{N_g}\cL_i^{a*}h_i.
\end{equation*}

This motivates the following two-step reconstruction procedure. 

\paragraph{Step I.} We reconstruct the coefficient $\Upsilon\tilde\sigma_s$ by solving the reduced linear system
\begin{equation}
	\cL^{red}\Upsilon\tilde\sigma_s =h^{red}.
\end{equation}
In the construction of the operator $\cL^{red}$, we need to apply the inverse of the operator $\sum_{i=1}^{N_g}\cL_i^{a*}\cL_i^a +\rho\cD_a^*\cD_a$. We construct this inverse as accurate as possible. For large problems, we avoid constructing the inverse operator directly. Instead, we construct it implicitly. More precisely, to apply the inverse operator to any object $z$ to get $\hat z=(\sum_{i=1}^{N_g}\cL_i^{a*}\cL_i^a +\rho\cD_a^*\cD_a)^{-1}z$, we solve the linear system $\big(\sum_{i=1}^{N_g}\cL_i^{a*}\cL_i^a +\rho\cD_a^*\cD_a\big)\hat z=z$ to maximum accuracy.

\paragraph{Step II.} Once the coefficient $\Upsilon\tilde\sigma_s$ is reconstructed, we can eliminate it from the system to obtain a linear system for the reconstruction of $\Upsilon\tilde\sigma_a$: $\big(\sum_{i=1}^{N_g}\cL_i^{a*}\cL_i^a +\rho\cD_a^*\cD_a\big)\Upsilon\tilde\sigma_a=\sum_{i=1}^{N_g}\cL_i^{a*}h_i-(\sum_{i=1}^{N_g}\cL_i^{a*}\cL_i^s)\Upsilon\tilde\sigma_s$.

\subsection{Minimization-based nonlinear reconstruction scheme}
\label{SUBSEC:Nonl Recon}

To solve the full nonlinear inverse problem, we reformulate the problem into a regularized least-squares formulation. More precisely, we minimize the following objective functional:
\begin{equation}\label{EQ:OBJ}
	\Phi(\Upsilon,\sigma_a,\sigma_s) 
		\equiv \dfrac{1}{2}\sum_{i=1}^{N_g}\|\Upsilon\sigma_a \aver{u_i}_\bv-H_i^*\|_{L^2(\Omega)}^2 
\end{equation}
where $u_i$ is the solution of the transport equation with the $i$th illumination, i.e., $u_i$ solves
\begin{equation}\label{EQ:ERT gi}
  \begin{array}{rcll}
  	\bv\cdot\nabla u_i(\bx,\bv) + 
  	\sigma_a(\bx) u_i(\bx,\bv)&=&
  	\sigma_s(\bx)K(u_i)(\bx,\bv)& \mbox{in}\ \ X\\
    u_i(\bx,\bv) &=& g_i(\bx,\bv) & \mbox{on}\ \ \Gamma_{-} .
  \end{array}
\end{equation}
The interior data collected for illumination $g_i$ is denoted by $H_i^*$.

It is known that the functional~\eqref{EQ:OBJ} is Fr\'echet differentiable with respect to the unknowns provided that the coefficients satisfy the assumptions in {\bf (A1)}; see for instance~\cite{Dorn-IP98,Dorn-OE00,ReBaHi-SIAM06}. We can thus use gradient-based minimization technique to solve this problem. To calculate the Fr\'echet derivatives of the objective functional, we use the method of adjoint equations~\cite{Vogel-Book02}. Let us denote by $w_i$ the solution of the following adjoint transport equation
\begin{equation}\label{EQ:Adjoint}
  \begin{array}{rcll}
  	-\bv\cdot\nabla w_i(\bx,\bv) + 
  	\sigma_a(\bx) w_i(\bx,\bv)&=&
  	\sigma_s(\bx)K(w_i)(\bx,\bv) + \Upsilon\sigma_a(\Upsilon\sigma_a \aver{u_i}_\bv-H_i^*) & \mbox{ in }\ X\\
    w_i(\bx,\bv) &=& 0 & \mbox{ on }\ \Gamma_{+} .
  \end{array}
\end{equation}
It is then straightforward to follow the standard calculations in~\cite{Dorn-IP98,Dorn-OE00,ReBaHi-SIAM06} to show that the Fr\'echet derivatives can be computed as follows.
\begin{theorem}
	The Fr\'echet derivatives of $\Phi$ are given by
	\begin{eqnarray}
	\langle \frac{\partial\Phi}{\partial \Upsilon}, \widehat \Upsilon \rangle = 
	\dsum_{i=1}^{N_g}\langle (\Upsilon\sigma_a \aver{u_i}_\bv-H_i^*) \sigma_a\aver{u_i}_\bv, \widehat \Upsilon \rangle_{L^2(X)},\\		
	\langle \frac{\partial\Phi}{\partial \sigma_a}, \widehat{\sigma_a} \rangle =
	\dsum_{i=1}^{N_g}\aver{(\Upsilon\sigma_a \aver{u_i}_\bv-H_i^*)\Upsilon\aver{u_i}_\bv - \aver{u_i  w_i}_\bv, \widehat {\sigma_a} }_{L^2(X)},\\
\label{EQ:Der sigma}\langle \frac{\partial\Phi}{\partial \sigma_s}, \widehat{\sigma_s} \rangle =\dsum_{i=1}^{N_g} \langle \aver{K(u_i)  w_i}_\bv, \widehat {\sigma_s} \rangle_{L^2(X)}
\end{eqnarray}
where $\langle\cdot,\cdot\rangle_{L^2(X)}$ denotes the usual inner product in $L^2(X)$.
\end{theorem}
Note that due to the fact that the problem of reconstructing $\Upsilon$ is a linear inverse problem, the Fr\'echet derivative of the objective functional with respect to $\Upsilon$ is independent of the adjoint solutions $\{w_i\}_{i=1}^{N_g}$. If we know $\sigma_a$ and $\sigma_s$ but are interested in reconstructing $\Upsilon$, then we can simply solve for $\Upsilon$ such that $\langle \frac{\partial\Phi}{\partial \Upsilon}, \widehat \Upsilon \rangle=0$ for any test function $\widehat\Upsilon$. This in turn gives us $\Upsilon\sigma_a \aver{u_i}_\bv-H_i^* =0,\ 1\le i\le {N_g}$, which holds for every point $\bx$ in $\Omega$. The least-squares solution of this overdetermined system is $\Upsilon=\sum_{i=1}^{N_g}\big(\aver{u_i}_\bv H_i^*\big)/\sum_{i=1}^{N_g}\big(\sigma_a\aver{u_i}_\bv^2\big)$. Thus we need only solve the $N_g$ transport equations and evaluate $\aver{u_i}_\bv$ to reconstruct $\Upsilon$.

We use the limited memory version of the BFGS quasi-Newton method to solve the minimization problem. The details of the implementation are documented in~\cite{ReBaHi-SIAM06,Ren-CiCP10}. In our numerical experiments, we only  attempt to reconstruct two of the three coefficients. We observe that the algorithm converges very fast, even from initial guesses that are relatively far from the true coefficients. This confirms the theory developed in the diffusive regime~\cite{BaRe-IP11,BaRe-IP12,BaUh-IP10} that is the problem is very well-conditioned when only two coefficients are sought. In the numerical simulation, we can add a small amount of regularization to the problem when the noise in the data is significant. This is done by adding a term of Tikhonov functional, such as $\dfrac{\rho}{2}\|(\Upsilon,\sigma_a)-(\bar\Upsilon,\bar\sigma_a)\|_{(\cH^1(\Omega))^2}^2$ in the reconstruction of $(\Upsilon,\sigma_a)$ where $\bar\Upsilon$, $\bar\sigma_a$, and $\bar\sigma_s$ are values obtained with \emph{a priori} knowledge, in the objective functional~\eqref{EQ:OBJ}.

\section{The multi-spectral QPAT}
\label{SEC:MultiSpec}

In practical applications of PAT, light of different wavelengths can be used to probe the properties of the medium at those wavelengths~\cite{ChDaBaMoCoSmChLe-PMB05,CoWa-IJBI06,CoGiHe-JBO10,CoArBe-JOSA09,CoLaBe-SPIE09,KiFlYaKaHi-JBO10,LaCoZhBe-AO10,Razansky-NP09,Razansky-OL07,ShCoZe-AO11,SrPoJiDePa-AO05,WaDaSrJiPoPa-JBO08,YuJi-OL09,ZaSvAxScArAn-OE09}. This is the idea of multi-spectral QPAT.

Let us denote by $\Lambda$ the set of wavelengths that can be accessed in practice, then the radiative transport equation in this case takes the form
\begin{equation}\label{EQ:ERT Spectra}
	\begin{array}{rcll}
  	\bv\cdot\nabla u(\bx,\bv,\lambda) + 
  	\sigma_a(\bx,\lambda)u(\bx,\bv,\lambda)
  	&=& \sigma_s(\bx,\lambda)K(u)(\bx,\bv,\lambda)
  	&\mbox{ in }\ X\times\Lambda\\
       u(\bx,\bv,\lambda) &=& g(\bx,\bv,\lambda) &\mbox{ on }\ \Gamma_{-}\times\Lambda,
	\end{array}
\end{equation}
where the coefficients $\sigma_a(\bx,\lambda)$ and $\sigma_s(\bx,\lambda)$ are now functions of the wavelength. The interior datum collected in this case is now also a function of the wavelength
\begin{equation}
	H(\bx,\lambda) = \Upsilon(\bx,\lambda)\int_{\bbS^{d-1}} \sigma_a(\bx,\lambda) u(\bx,\bv,\lambda) d\bv .
\end{equation}
Due to the fact that the equations for different wavelengths are de-coupled, we could not expect to reconstruct unknowns at one wavelength from data collected at other wavelengths, unless other \emph{a priori} information is supplied. 

If the medium is non-scattering, the results in Section~\ref{SEC:Non-scattering} enable us to reconstruct $\Upsilon(\bx,\lambda)$ and $\sigma_a(\bx,\lambda)$ with two well-chosen illuminations $g_1(\bx,\bv,\lambda)$ and $g_2(\bx,\bv,\lambda)$. The illuminations can be either collimated sources of the form $\frak g(\bx,\lambda)\delta(\bv-\bv')$ or point source in space of the form $\frak g(\bv,\lambda)\delta(\bx-\bx')$. We can simply perform wavelength by wavelength reconstructions using the methods developed in Section~\ref{SEC:Non-scattering}. 

\subsection{Uniqueness of reconstruction}

The theory developed in~\cite{BaRe-IP12} for the diffusion equation confirms that we can reconstruct all three coefficients $\Upsilon$, $\sigma_a$, and $\sigma_s$ simultaneously with multi-spectral interior data under practically reasonable assumptions on the coefficients. Following~\cite{BaRe-IP12}, we take the following standard model for the unknown coefficients: 
\begin{equation}\label{EQ:Model II}
	\sigma_a(\bx,\lambda)=\sum_{k=1}^K\alpha_k(\lambda)\sigma_a^k(\bx), \qquad 
	\sigma_s(\bx,\lambda)=\beta(\lambda)\sigma_s(\bx), \qquad 
	\Upsilon(\bx,\lambda)=\gamma(\lambda)\Upsilon(\bx),
\end{equation}
where the functions $\{\alpha_k(\lambda)\}_{k=1}^K$, $\beta(\lambda)$ and $\gamma(\lambda)$ are assumed to be \emph{known}. In other words, all three coefficient functions can be written as products of functions of different variables. Moreover, the absorption coefficient contains multiple components.  This is the parameter model proposed in~\cite{CoArBe-JOSA09,CoLaBe-SPIE09,KiFlYaKaHi-JBO10,LaCoZhBe-AO10,WaDaSrJiPoPa-JBO08} to reconstruct chromophore concentrations from photoacoustic measurements.

We have the following uniqueness result with the multi-spectral data.
\begin{corollary}\label{THM:Uniq Spec}
	Let $(\Upsilon,\sigma_a,\sigma_s)$ and $(\tilde\Upsilon,\tilde\sigma_a,\tilde\sigma_s)$ be two sets of coefficients of the form~\eqref{EQ:Model II}, satisfying the regularity assumptions in {\bf (A1)}. Assume that we have data from $M (\ge K)$ different wavelengths $\{\lambda_m\}_{m=1}^M$ such that the matrix $\boldsymbol\alpha=(\alpha_k(\lambda_m)), 1\le k\le K, 1\le m\le M$ has rank $K$. Let $\cA_\lambda$ be the albedo operator that depends on wavelength. Then $\cA_\lambda=\tilde\cA_\lambda$ implies $\{\Upsilon(\bx),\{\sigma_a^k(\bx)\}_{k=1}^K, \sigma_s(\bx)\}= \{\tilde \Upsilon(\bx),\{\tilde \sigma_a^k(\bx)\}_{k=1}^K, \tilde\sigma_s(\bx)\}$.
\end{corollary}
\begin{proof}
	Corollary~\ref{THM:Uniq} implies that we can reconstruct uniquely
	\begin{equation}
		\mu(\bx,\lambda)=\gamma(\lambda)\Upsilon(\bx)\sigma_a(\bx,\lambda), \qquad
		\sigma(\bx,\lambda)=\sigma_a(\bx,\lambda)+\beta(\lambda)\sigma_s(\bx)		
	\end{equation}
	Take two wavelengths $\lambda_1$ and $\lambda_2$. We then have 
	\begin{eqnarray}
		\frac{\mu(\bx,\lambda_1)}{\mu(\bx,\lambda_2)}=\frac{\gamma(\lambda_1)\sigma_a(\bx,\lambda_1)}{\gamma(\lambda_2)\sigma_a(\bx,\lambda_2)},\\
		\beta(\lambda_2)\sigma(\bx,\lambda_1)-\beta(\lambda_1)\sigma(\bx,\lambda_2)=\beta(\lambda_2)\sigma_a(\bx,\lambda_1)-\beta(\lambda_1)\sigma_a(\bx,\lambda_2) .		
	\end{eqnarray}
	If $\beta(\lambda_1)\gamma(\lambda_1)\mu(\bx,\lambda_2)-\beta(\lambda_2)\gamma(\lambda_2)\mu(\bx,\lambda_1)\neq 0$ at $\bx\in\Omega$, we can solve the above system of equations to reconstruct
	$\sigma_a(\bx,\lambda_1)$ and $\sigma_a(\bx,\lambda_2)$. We then obtain $\Upsilon(\bx)$, and $\sigma_s(\bx)$. Once we know $\Upsilon(\bx)$ and $\sigma_s(\bx)$, we can reconstruct uniquely $\sigma_a(\bx,\lambda)$ for any $\lambda\in\Lambda$. Now select $\{\lambda_m\}_{m=1}^M$ from $\Lambda$ such that the matrix $\boldsymbol\alpha=(\alpha_k(\lambda_m)), 1\le k\le K, 1\le m\le M$ has rank $K$. We can reconstruct $\{\sigma_a^k(\bx)\}_{k=1}^K$ by solving the system $\sum_{k=1}^K\alpha_k(\lambda_m)\sigma_a^k(\bx)=\sigma(\bx,\lambda_m),\ 1\le m\le M$. This completes the proof.
\end{proof}

\subsection{Reconstruction methods}

The linearized and nonlinear reconstruction methods proposed in the previous section can be adapted to use the multi-spectral interior data. We would not repeat the whole algorithm again but just highlight the main modifications here for the linearized reconstruction with Born approximation. 

We can build an analogue of~\eqref{EQ:LIN INTGRL} 
\begin{equation}\label{EQ:LIN INTGRL SPEC}
	\dfrac{H_\lambda}{\gamma(\lambda)\sigma_{a0}\aver{U_\lambda}}(\bx)=\cI(\Upsilon)+\cL_\lambda^a(\Upsilon\tilde\sigma_a)+\cL_\lambda^s(\Upsilon\tilde\sigma_s),
\end{equation}
where $\cI$ is the identity operator and the operators $\cL^a$ and $\cL^s$ are defined, respectively, as
\begin{eqnarray}
	\label{EQ:La SPEC}
	\cL_\lambda^a(\Upsilon\tilde\sigma_a) &=& \dfrac{\Upsilon \tilde\sigma_a}{\sigma_{a0}}  
	      +\int_\Omega \Upsilon(\by)\tilde\sigma_a(\by,\lambda)
		\dfrac{\aver{U_\lambda G_\lambda}_\bv}{\aver{U_\lambda}_\bv}(\by;\bx) d\by ,\\
	\label{EQ:Ls SPEC}
	\cL_\lambda^s(\Upsilon\tilde\sigma_s) &=& -\alpha(\lambda)\int_\Omega \Upsilon(\by)\tilde\sigma_s(\by) 
		\dfrac{\aver{K(U_\lambda)G_\lambda}_\bv}{\aver{U_\lambda}_\bv}(\by;\bx) d\by .	 
\end{eqnarray}

Let us assume again that we collect the data for ${N_g}$ different illumination patterns and for each illumination patten, we have data for $M$ different wavelengths $\{\lambda_m\}_{m=1}^M$. We denote by $H_{i,\lambda_m}, 1\le i\le {N_g}, 1\le m\le M$ the $i$th data set at wavelength $\lambda_m$. The system of linear integral equations with these ${N_g}\times M$ data sets can be written as
\begin{equation}\label{EQ:LIN INTGRL Sys Spec}
	\left(
	\begin{array}{cccccc}
		\cI & \cL^a_{1,\lambda_1} & 0 & \ldots & 0 & \cL^s_{1,\lambda_1}\\
		\vdots  & \vdots & \vdots\\
		\cI & 0 & 0 & \ldots & \cL^a_{1,\lambda_M} & \cL^s_{1,\lambda_M}\\
		\vdots  & \vdots & \vdots\\
		\cI & \cL^a_{i,\lambda_1} & 0 & \ldots & 0 & \cL^s_{i,\lambda_1}\\
		\vdots  & \vdots & \vdots\\
		\cI & 0 & 0 & \ldots & \cL^a_{i,\lambda_M} & \cL^s_{i,\lambda_M}\\
		\vdots  & \vdots & \vdots\\
		\cI & \cL^a_{{N_g},\lambda_1} & 0 & \ldots & 0 & \cL^s_{{N_g},\lambda_1}\\
		\vdots  & \vdots & \vdots\\
		\cI & 0 & 0 & \ldots & \cL^a_{{N_g},\lambda_M} & \cL^s_{{N_g},\lambda_M}\\
	\end{array}
	\right)
	\left(
	\begin{array}{c}
		\Upsilon\\
		\tilde\sigma_a(\bx,\lambda_1)\\
		\vdots\\
		\tilde\sigma_a(\bx,\lambda_m)\\
		\vdots\\
		\tilde\sigma_a(\bx,\lambda_M)\\
		\tilde\sigma_s(\bx)
	\end{array}
	\right)
	=
	\left(
	\begin{array}{c}
		h_{1,\lambda_1}\\
		\vdots\\
		h_{1,\lambda_M}\\
		\vdots\\
		h_{i,\lambda_1}\\
		\vdots\\
		h_{i,\lambda_M}\\
		\vdots\\
		h_{{N_g},\lambda_1}\\
		\vdots\\
		h_{{N_g},\lambda_M}
	\end{array}
	\right)
\end{equation}
with $\cL^a_{i,\lambda_m}$ and $\cL^s_{i,\lambda_m}$ being the evaluation of the operators defined in~\eqref{EQ:La SPEC} and ~\eqref{EQ:Ls SPEC} respectively at source $g_i$ and wavelength $\lambda_m$. If we introduce the notation 
\begin{eqnarray*}
\cL_i^a & = &\diag (\cL_{i,\lambda_m}^a,\ldots,\cL_{i,\lambda_m}^a,\ldots,\cL_{i,\lambda_M}^a),  \\
\tilde\sigma_a & = & (\tilde\sigma_a(\bx,\lambda_1),\ldots,\tilde\sigma_a(\bx,\lambda_m),\ldots,\tilde\sigma_a(\bx,\lambda_M))^T, \\
h_i & = & (h_{i,\lambda_1},\ldots,h_{i,\lambda_m},\ldots,h_{i,\lambda_M})^T, 
\end{eqnarray*}
then this system is in the exact same form as~\eqref{EQ:LIN INTGRL Sys}. We solve this linear system in least-squares sense again in exactly the same ways as those presented in Section~\ref{SUBSEC:Born}. Once $(\Upsilon(\bx),\{\tilde\sigma_a(\bx,\lambda_m)\}_{m=1}^M,\tilde\sigma_s(\bx))$ are reconstructed, we reconstruct the coefficient components $\{\tilde\sigma_a^k\}_{k=1}^K$ by solving the linear system (in least-squares sense) locally (i.e. at each spatial location $\bx\in\Omega$)
\[
	\left(
	\begin{array}{ccc}
		\alpha_1(\lambda_1) & \ldots & \alpha_K(\lambda_1) \\
		\vdots & \ldots & \vdots \\
		\alpha_1(\lambda_M) & \ldots & \alpha_K(\lambda_M) 
	\end{array}
	\right)
	\left(
	\begin{array}{c}
		\tilde\sigma_a^1(\bx)\\
		\vdots\\
		\tilde\sigma_a^K(\bx)
	\end{array}
	\right)
	=
	\left(
	\begin{array}{c}
		\tilde\sigma_a(\bx,\lambda_1)\\
		\vdots\\
		\tilde\sigma_a(\bx,\lambda_M)
	\end{array}
	\right) .
\] 

\section{Numerical experiments}
\label{SEC:Num}

We now present some numerical simulations with synthetic data to demonstrate the performance of the algorithms that we have presented in the previous sections. To simplify numerical computation, we consider only two-dimensional simulations, although the algorithms we have described are independent of spatial dimension. We will use both $\bx$ and $(x,y)$ to denote a point on the plane. The computational domain is the square $\Omega=(0,2)\times(0,2)$. We denote by $\overline{\Omega}\equiv\Omega\cup\partial\Omega$. We also denote by $\partial\Omega|_L$, $\partial\Omega|_R$, $\partial\Omega|_B$, $\partial\Omega|_T$ the left, right, bottom and top sides of the boundary $\partial\Omega$ respectively. For instance, $\partial\Omega|_L=\{(x,y) \; | \; x=0,y\in(0,2)\}$. The unit sphere $\bbS^1$ is parameterized by an angle $\theta \in [0,\ 2\pi)$ so that the direction vector $\bv$ can be represented as ($\cos\theta$, $\sin\theta$). As before, we will denote by $\aver{u}_\bv$ and $\aver{u}_\theta$ the average of $u$ on $\bbS^1$.

To generate synthetic data, we solve the transport problem on a spatial mesh that is four times as fine as the mesh used to solve the inverse problem. We then compute $H$ in~\eqref{EQ:Data} by averaging in the four neighbor cells. This way, the interior data $H$ contain automatically noise due to the mismatch in spatial discretization. However, we will still call the data constructed this way the ``noiseless'' data. In the case when the parameters that we are interested in are discontinuous in space, there is no way to recover the discontinuity in the coefficients exactly even with noiseless data. 

To get noisy data, we add random noise in the following way. Let $\bH\in\bbR^{N_\Omega}$ ($N_\Omega$ being the total number of grid points) be a vector of noiseless data on the grid, then the noisy data vector $\widetilde{\bH}\in\bbR^{N_\Omega}$ is given by
\begin{equation}
	\widetilde{\bH} = (\cI + \epsilon \mathcal{N}) \bH, \quad \mathcal{N} = \diag (X_1,\ldots,X_{N_\Omega}),
\label{EQN:cone-noise}
\end{equation}
where $X_j$, $j=1,\ldots,N_\Omega$ are independent identically distributed Gaussian random variables with zero mean and unit variance, and $\epsilon$ is the parameter that controls the noise level in the noisy data. For a particular value of the parameter $\epsilon$ we say that the noise level is $\epsilon \cdot 100\%$. For example, for $\epsilon =
 0.05$ we say the noise level is $5\%$.

\subsection{Reconstructions in non-scattering media}
\label{SUBSEC:num-nonscatter}

In this section we present some numerical simulations in non-scattering media following the reconstruction methods presented in Section~\ref{SEC:Non-scattering}. To do the reconstruction, we cover the domain with $100\times 100$ cells of uniform size whose nodes are given as
\begin{displaymath}
	\Omega_\Delta=\{\bx_{i,j}=(x_i,y_j)|\ x_i=i\Delta x,\ y_j=j\Delta y,\ i, j = 0, 1, \ldots, 100\},
\end{displaymath}
with $\Delta x=\Delta y=0.02$. This is four times as fine as the grid used in next section for reconstructions in scattering media.

\subsubsection{Reconstructing $\sigma_a$}

We show now some reconstructions of the absorption coefficient. 

\begin{figure}[ht]
\centering
\includegraphics[angle=0,width=0.23\textwidth]{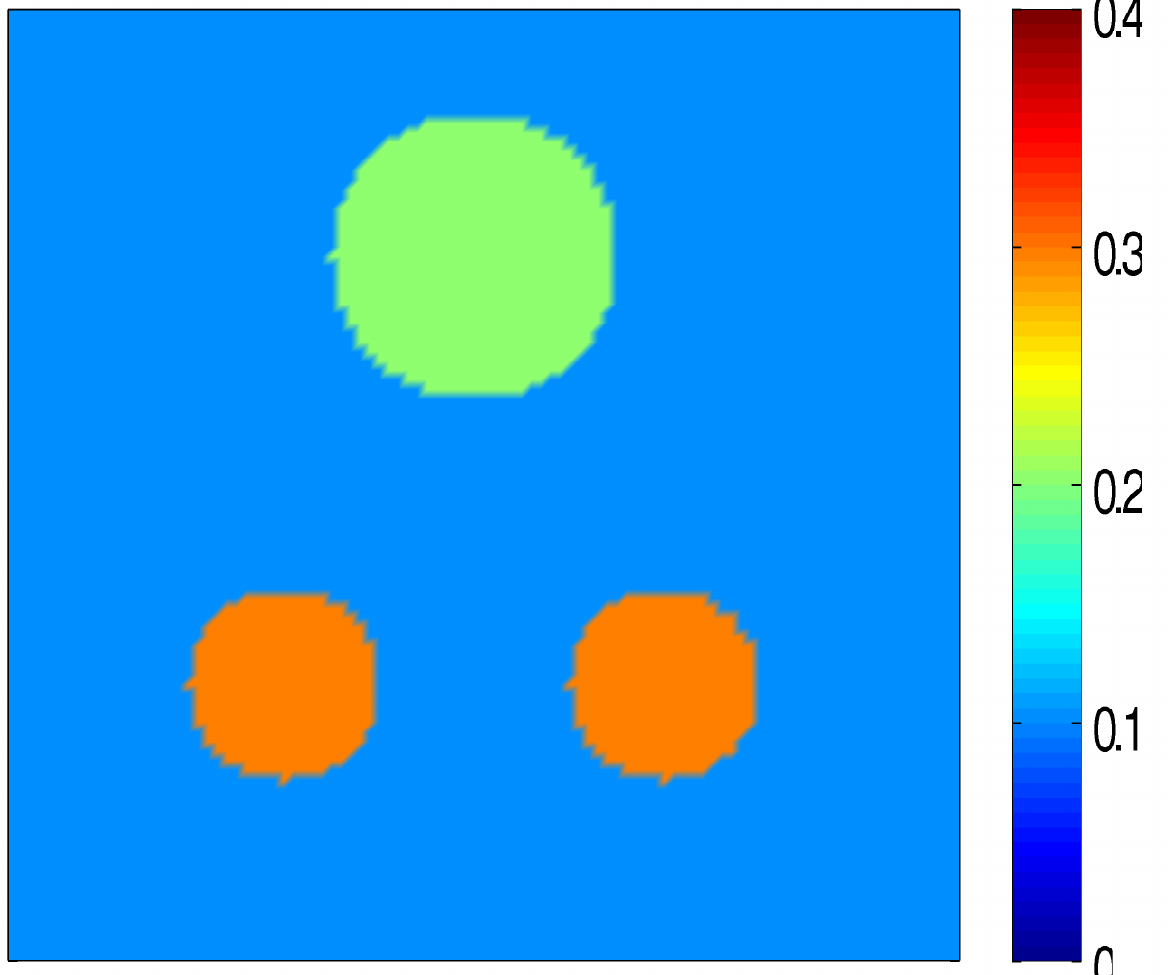} 
\includegraphics[angle=0,width=0.23\textwidth]{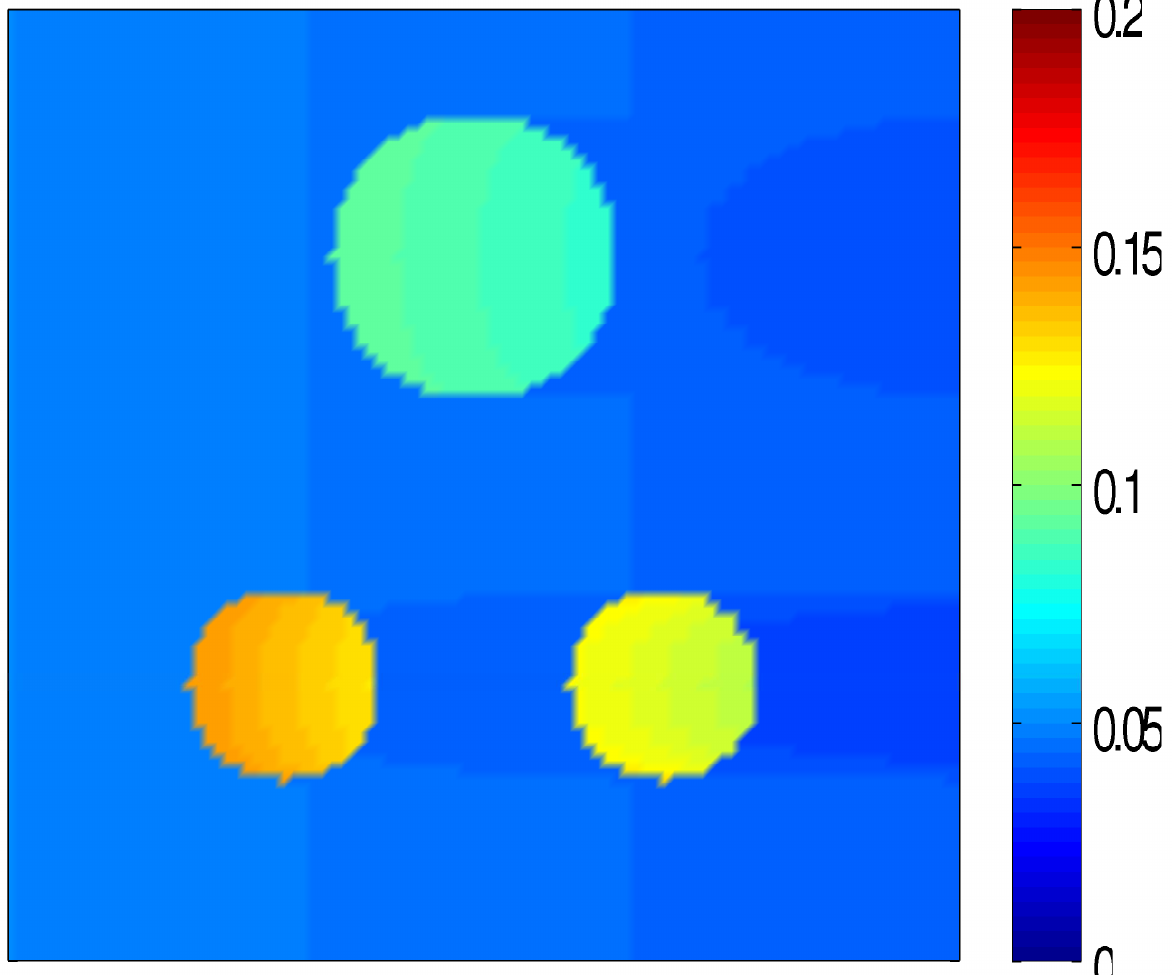} 
\includegraphics[angle=0,width=0.23\textwidth]{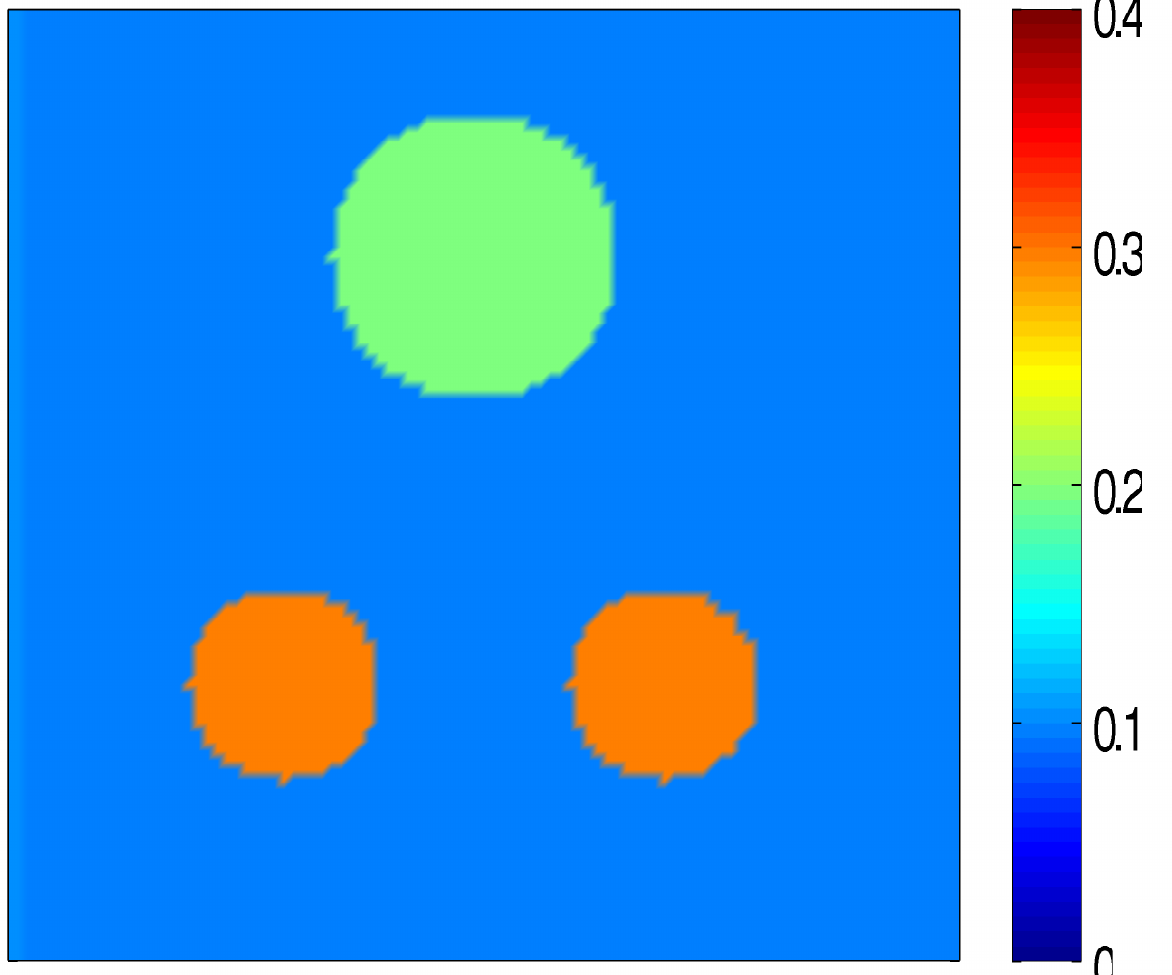} 
\includegraphics[angle=0,width=0.23\textwidth]{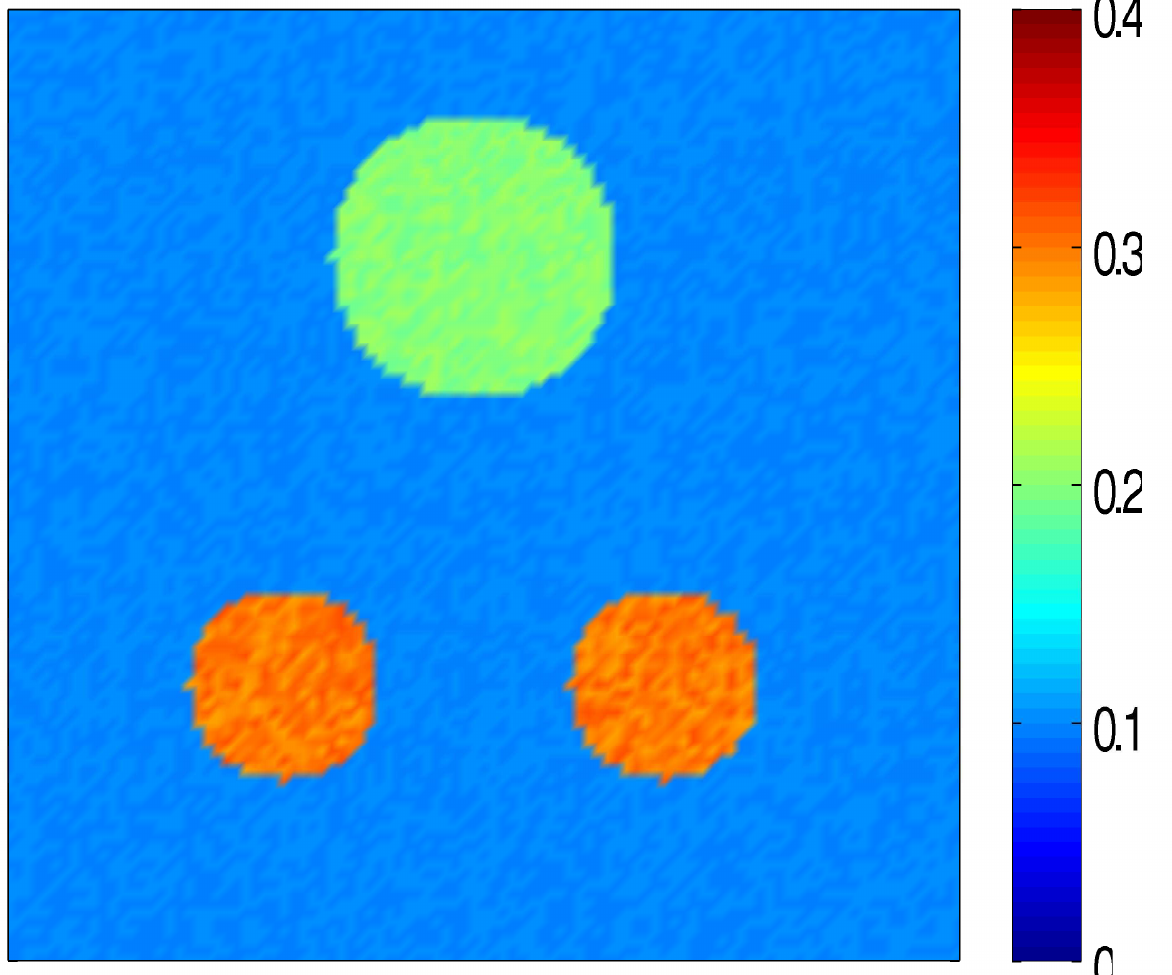} 
\caption{Reconstructions of a piecewise constant absorption coefficient with a collimated source. Left to right: true absorption coefficient $\sigma_a$, interior data $H$, $\sigma_a$ reconstructed with noiseless data and $\sigma_a$ reconstructed with noisy data (noise level $5\%$).}
\label{FIG:Collimated}
\end{figure}
In the first numerical simulation in this group, we perform a reconstruction of a piecewise constant absorption function using a collimated source located on the left side of the boundary pointing inside the domain. The source is $g(\bx,\bv)=\chi_{\partial\Omega|_L}\delta(\bv-\bv')$ with $\bv'=(1,0)$. 
The absorption function consists of a background $\sigma_a=0.1\mbox{ cm}^{-1}$ and three disk inclusions $\Omega_1=\{\bx \in \Omega \;|\; |\bx-(0.6,0.6)|\le 0.2\}$, $\Omega_2=\{\bx \in \Omega \;|\; |\bx-(1.4,0.6)|\le 0.2\}$ and $\Omega_3=\{\bx \in \Omega \;|\; |\bx-(1.0,1.5)|\le 0.3\}$ with values ${\sigma_a}_{|\Omega_1}=0.3\mbox{ cm}^{-1}$, ${\sigma_a}_{|\Omega_2}=0.3\mbox{ cm}^{-1}$ and ${\sigma_a}_{|\Omega_3}=0.2\mbox{ cm}^{-1}$ respectively. The reconstruction results with noiseless and noisy data are presented in Fig.~\ref{FIG:Collimated}. The reconstruction is almost perfect when noiseless data is used. When noisy data ($\epsilon=0.05$) is used, we can clearly see a degeneration of the quality of the reconstruction. However, the degeneration is very small, comparable to the noise level of the data.

\begin{figure}[ht]
\centering
\includegraphics[angle=0,width=0.23\textwidth]{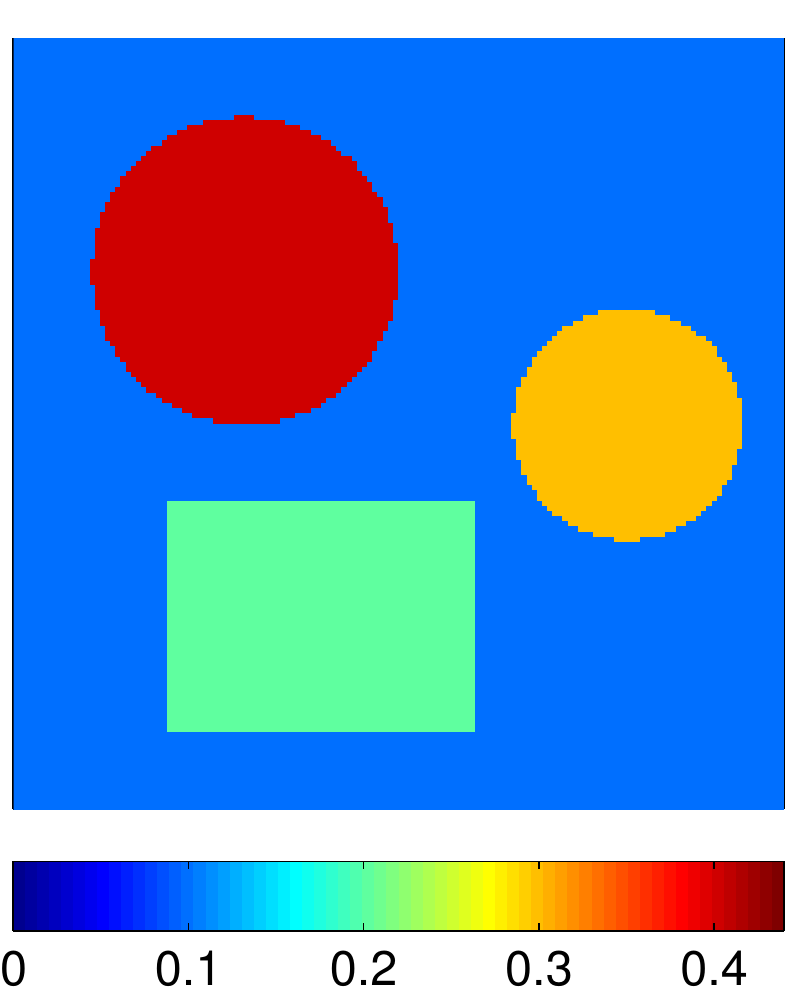} 
\includegraphics[angle=0,width=0.23\textwidth]{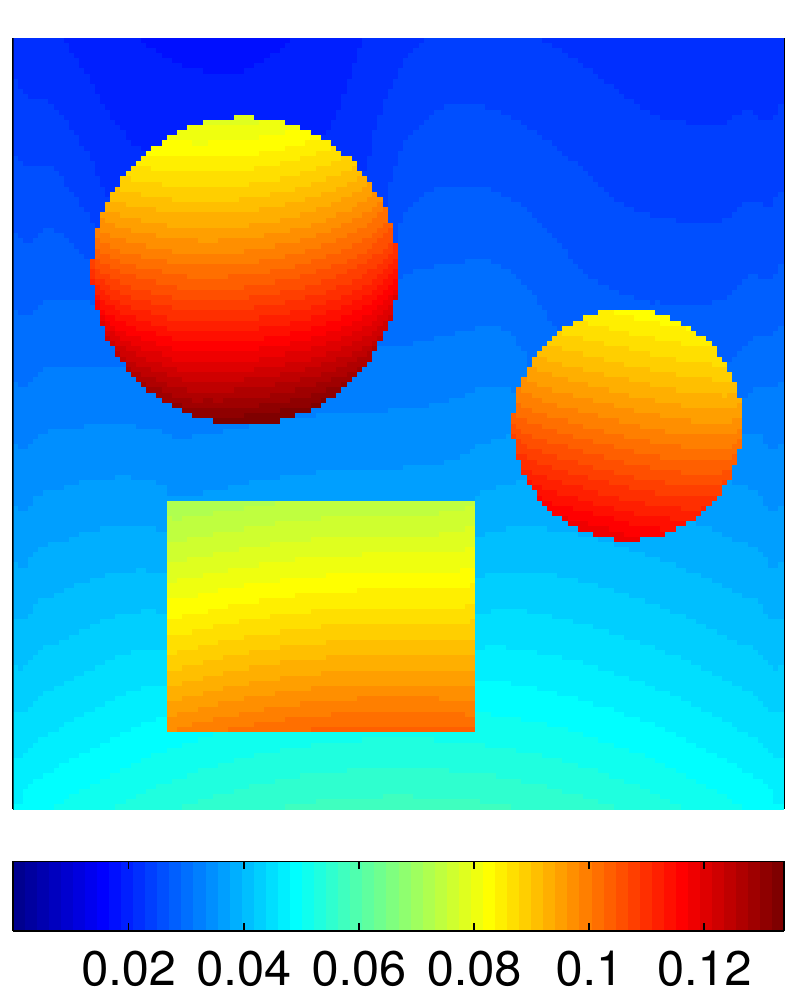} 
\includegraphics[angle=0,width=0.23\textwidth]{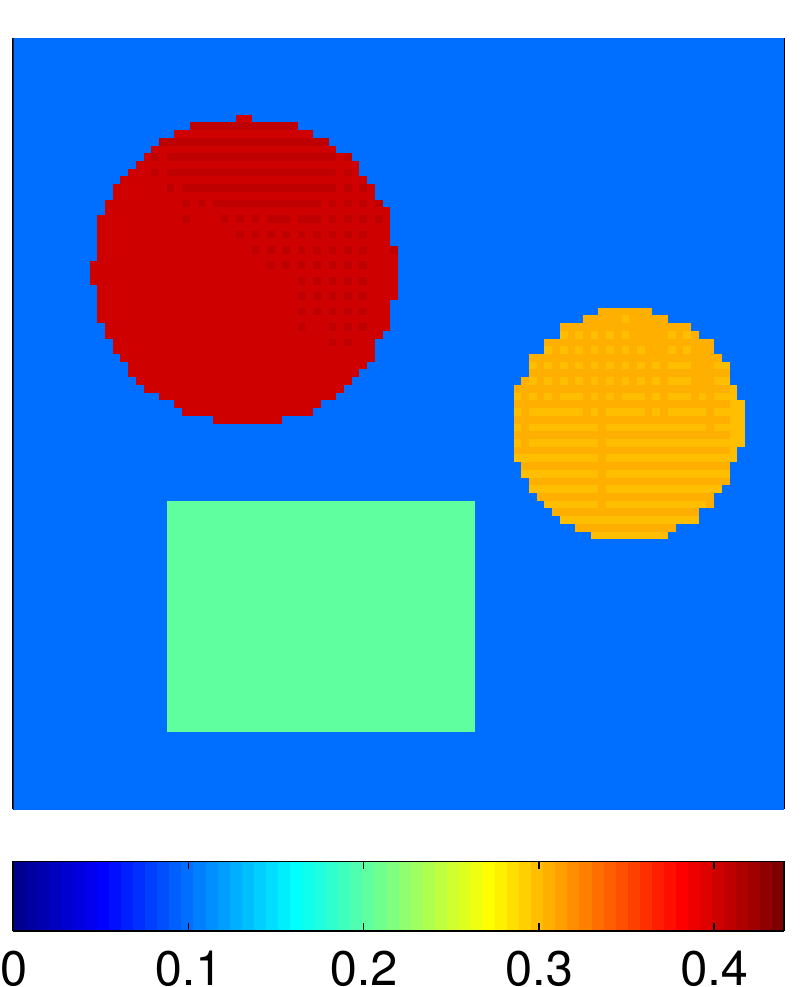}
\includegraphics[angle=0,width=0.23\textwidth]{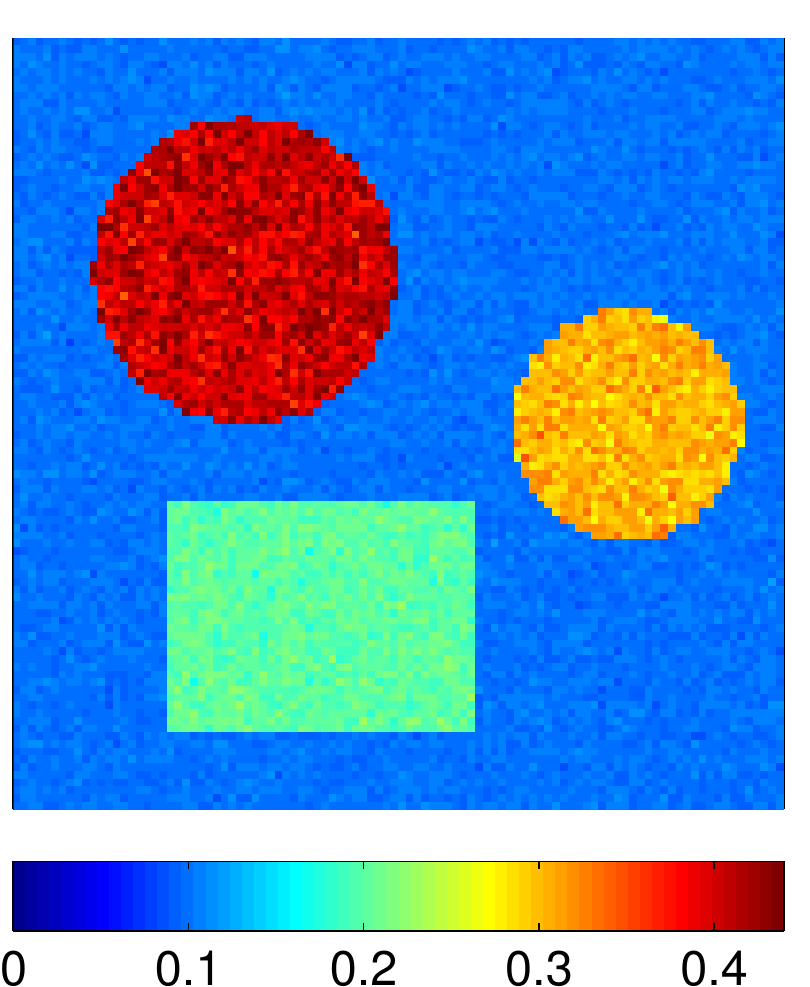}
\caption{Reconstructions of a piecewise constant absorption coefficient with cone-limited source. Left to right: true absorption coefficient $\sigma_a$, interior data $H$, $\sigma_a$ reconstructed with noiseless data and $\sigma_a$ reconstructed with noisy data (noise level $5\%$).}
\label{FIG:cone-pwc}
\end{figure}
Next we perform two reconstructions using a cone-limited source. The setup is as depicted in Fig.~\ref{FIG:cone-square}. The boundary condition $g(\bx,\theta)$ corresponds to a uniform isotropic line source of unit intensity concentrated on a segment $\{(x,y) \;|\; x\in(0,2), y=-2\}$ which results in $g(\bx,\theta)$ being non-zero only on $\partial \Omega|_B$. The half-aperture angle $\theta_0$ for such source is $\pi/4$, so we only keep track of the solution $u(x,y,\theta)$ for $\theta \in [\pi/4, 3\pi/4]$. This segment is uniformly discretized with $50$ nodes to generate the data and with $40$ nodes to solve the inverse problem using the algorithm from Section~\ref{SUBSEC:rec-cone}.

\begin{figure}[ht]
\centering
\includegraphics[angle=0,width=0.23\textwidth]{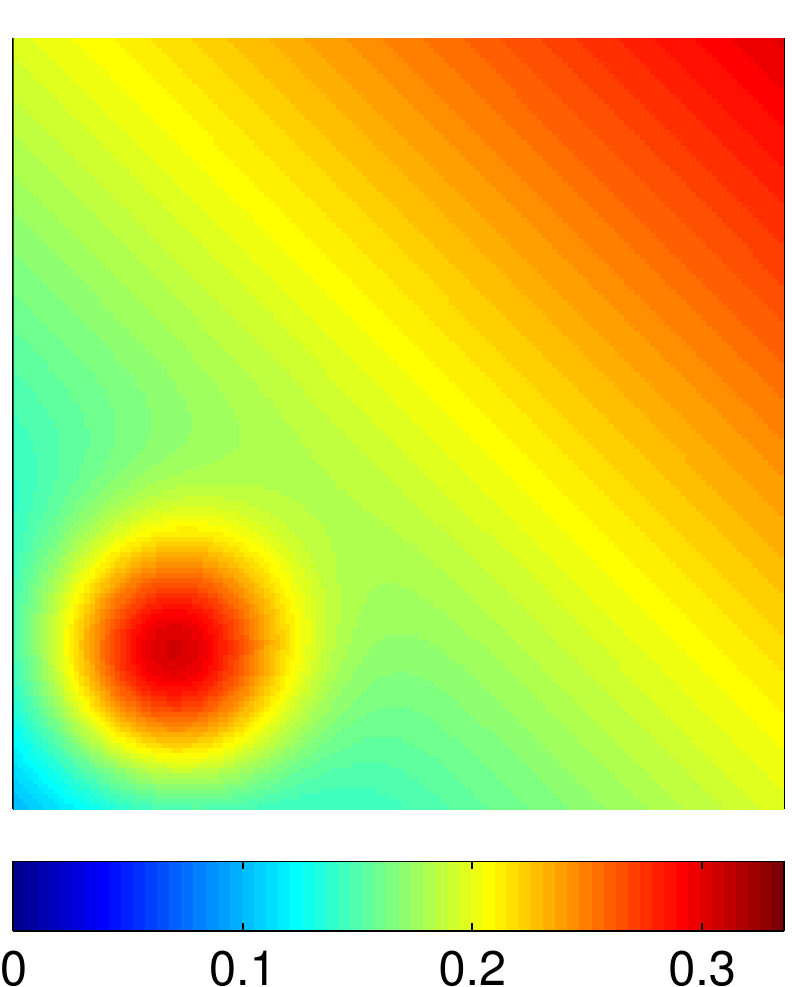}
\includegraphics[angle=0,width=0.23\textwidth]{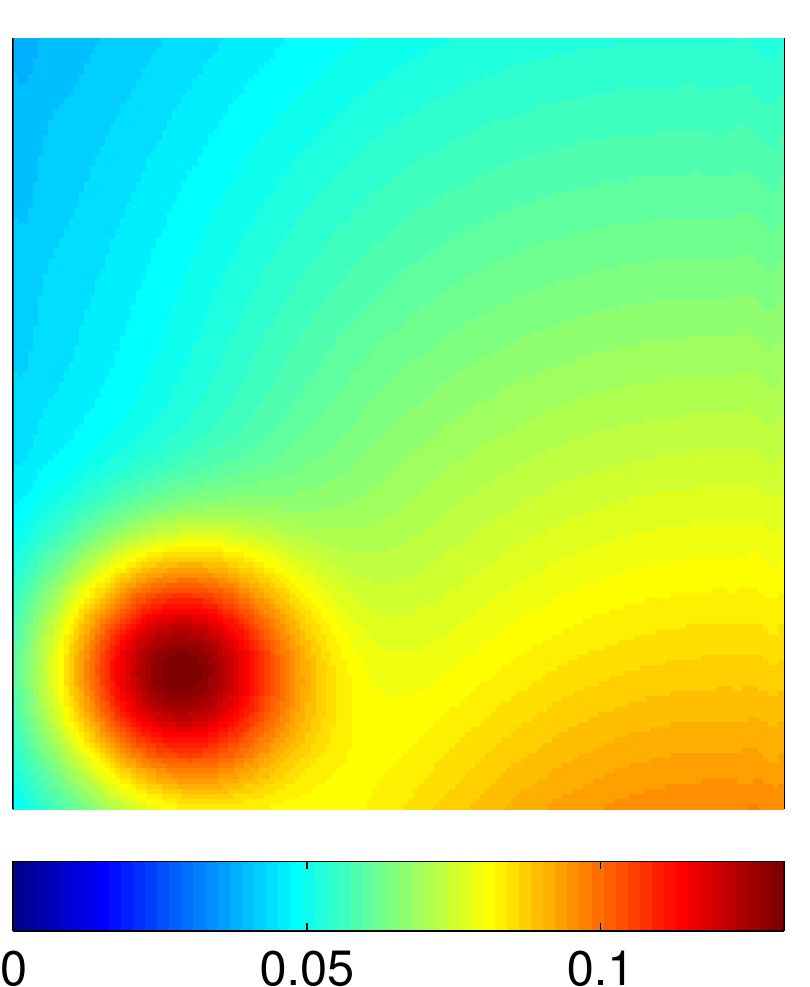}
\includegraphics[angle=0,width=0.23\textwidth]{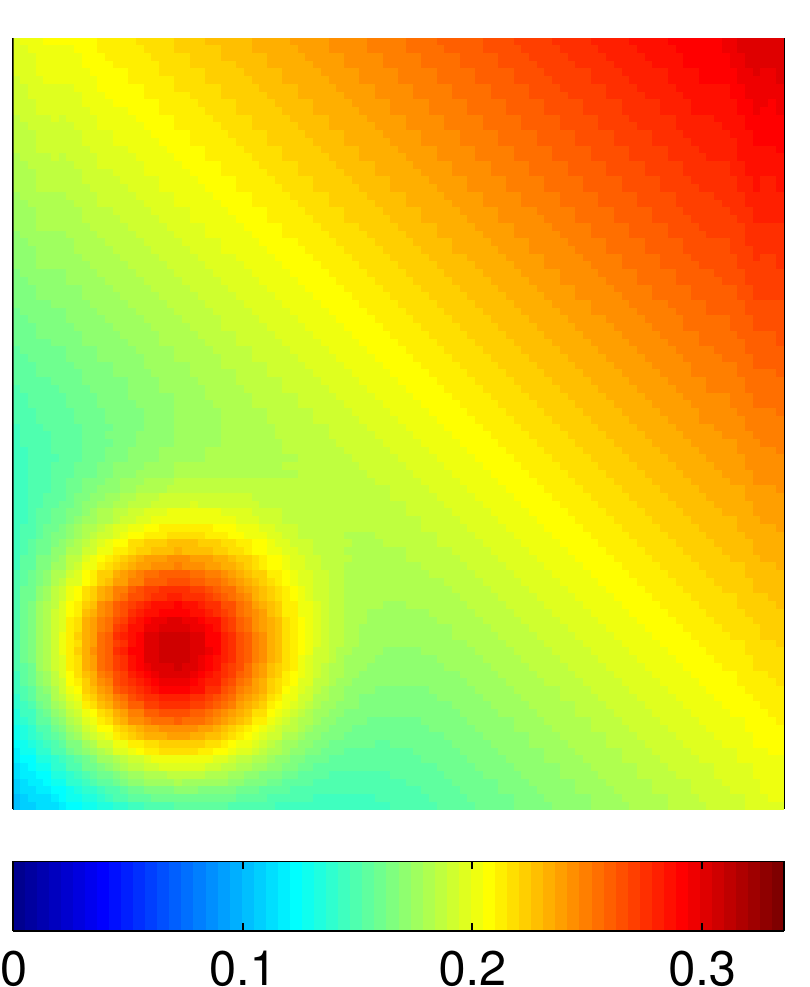}
\includegraphics[angle=0,width=0.23\textwidth]{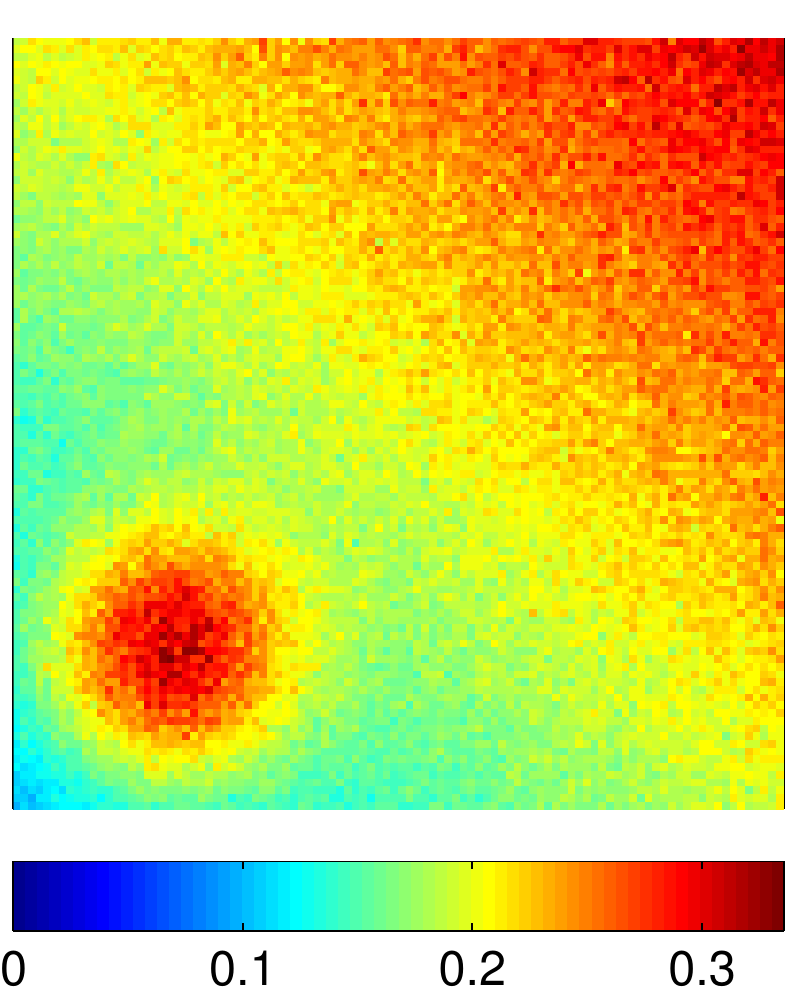}
\caption{Reconstructions of a smooth absorption coefficient with cone-limited source. Left to right: true absorption coefficient $\sigma_a$, interior data $H$, $\sigma_a$ reconstructed with noiseless data and $\sigma_a$ reconstructed with noisy data (noise level $5\%$).}
\label{FIG:cone-smooth}
\end{figure}
We show in Fig.~\ref{FIG:cone-pwc} the reconstruction of a piecewise-constant absorption coefficient with the cone-limited source. The absorption coefficient (in units of $\mbox{cm}^{-1}$) is given by
\begin{equation}\label{EQ:Gas Absorption1}
	\sigma_a (\bx) = 0.1 + 0.1\chi_{\Omega_1}(\bx)+0.2\chi_{\Omega_2}(\bx)+0.3\chi_{\Omega_3}(\bx),
\end{equation}
with the rectangular inclusion $\Omega_1=[0.4,1.2]\times[0.2,0.8]$, the smaller disk inclusion $\Omega_2=\{\bx \in \Omega \;|\; |\bx-(1.6,1.0)|\le 0.3\}$ and the larger disk inclusion $\Omega_3=\{\bx \in \Omega \;|\; |\bx-(0.6,1.4)|\le 0.4\}$. Thus the absorption coefficient in the background is $\sigma_a=0.1\mbox{ cm}^{-1}$ while that in the three inclusions it takes the values ${\sigma_a}|_{\Omega_1}=0.2\mbox{ cm}^{-1}$, ${\sigma_a}|_{\Omega_2}=0.3\mbox{ cm}^{-1}$ and ${\sigma_a}|_{\Omega_3}=0.4\mbox{ cm}^{-1}$ respectively. The quality of the reconstruction is comparable to that in the previous numerical experiment in Fig.~\ref{FIG:Collimated}. Smoother absorption coefficients can be reconstructed with similar quality. To demonstrate that, we show in Fig.~\ref{FIG:cone-smooth} the reconstruction of the absorption coefficient that is a sum of a Gaussian and linear functions given by
\begin{equation}\label{EQ:Gas Absorption2}
	\sigma_a (x,y) 
		= Ax + By + C + D e^{\left.\left((x-0.4)^2+(y-0.4)^2\right)\right/s_D^2}
\end{equation}
where the parameters $A,B,C,D$, and $s_D$ are chosen so that $0.1\mbox{ cm}^{-1} \leq \sigma_a \leq 0.3\mbox{ cm}^{-1}$. We performed reconstructions for many different choices of smooth absorption coefficients. The quality of the reconstruction in Fig.~\ref{FIG:cone-smooth} is representative of the typical reconstruction quality that we obtained in those experiments.

\begin{figure}[ht]
\centering
\includegraphics[angle=0,width=0.4\textwidth]{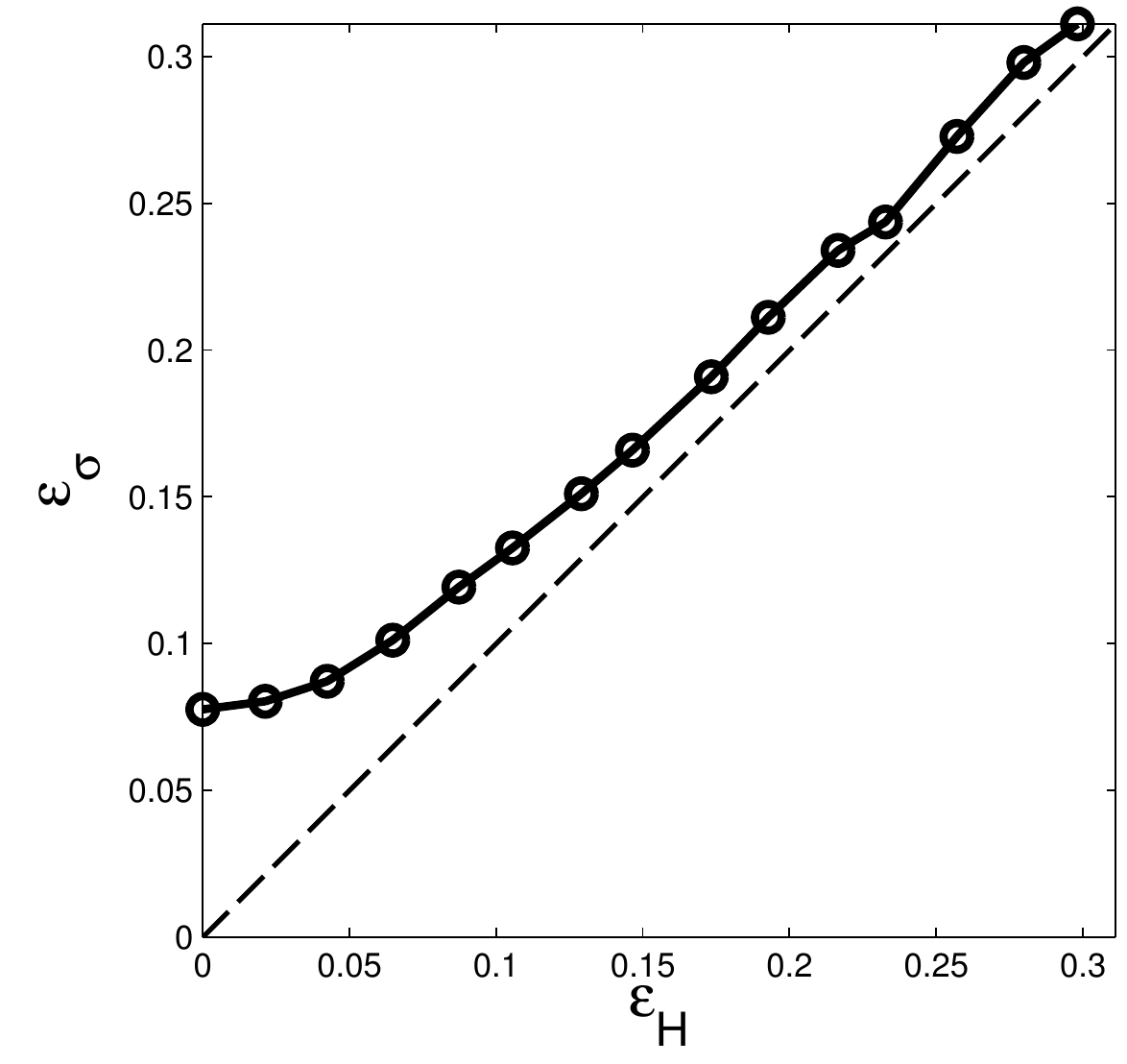}
\includegraphics[angle=0,width=0.4\textwidth]{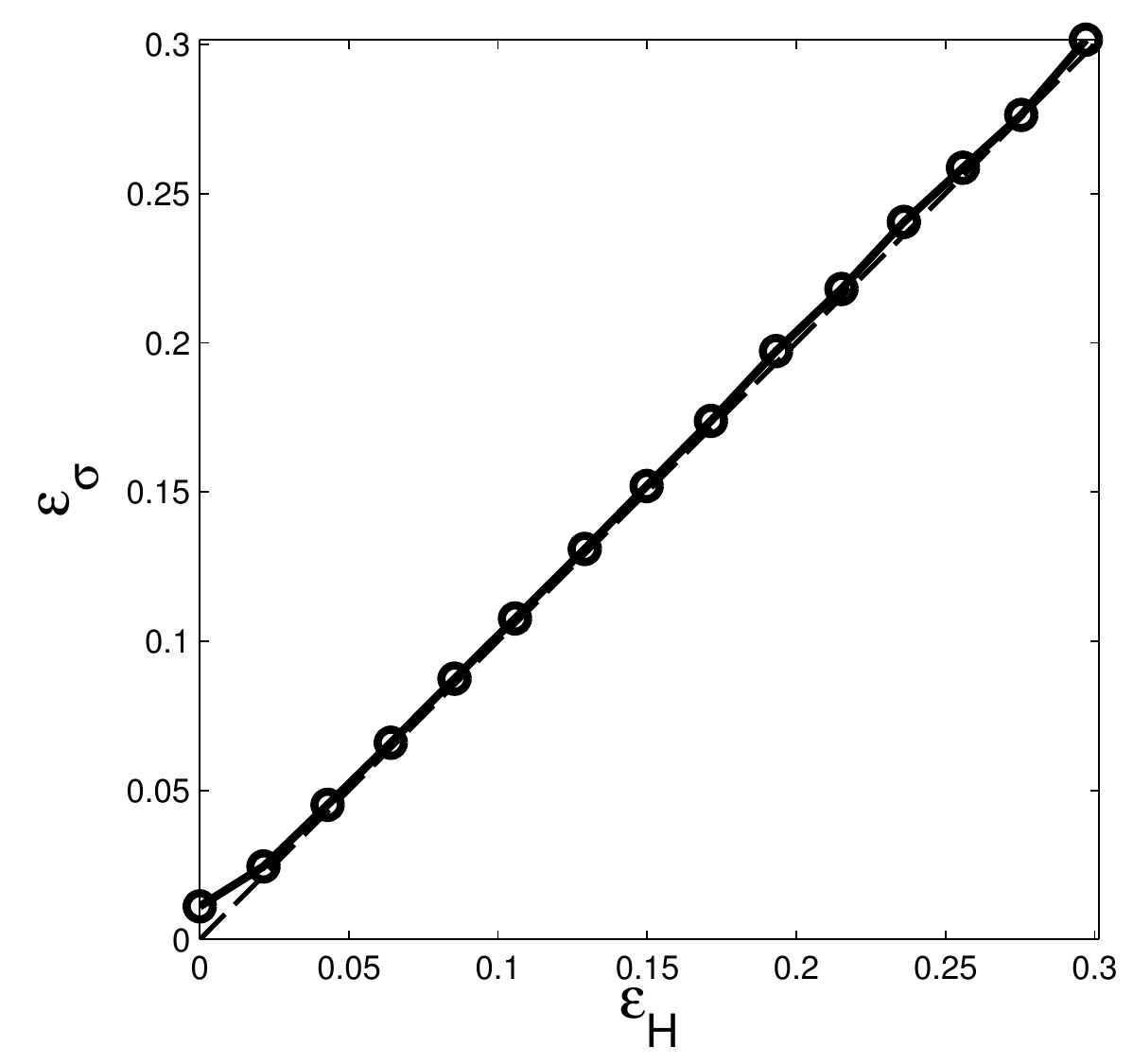}
\caption{Relative error $\cE_\sigma$ in the reconstruction of the absorption coefficient with a cone-limited source versus noise level $\cE_H$ in the data. Noise levels from $0$ to $30\%$. 
Left: reconstruction of piecewise constant absorption coefficient~\eqref{EQ:Gas Absorption1}; Right: reconstruction of smooth absorption coefficient~\eqref{EQ:Gas Absorption2}. 
Perfect linear stability $\cE_\sigma = \cE_H$ is shown as a dashed line for reference.}
\label{FIG:cone-noise-error}
\end{figure}
To characterize the stability of the reconstruction more precisely, we plot in Fig.~\ref{FIG:cone-noise-error} the relative $L^2$ error $\cE_\sigma$ in the reconstruction of the absorption coefficients~\eqref{EQ:Gas Absorption1} and ~\eqref{EQ:Gas Absorption2} versus the noise level $\cE_H$ in the data used, with $\cE_\sigma$ and $\cE_H$ defined respectively as
\begin{equation}
	\cE_\sigma = \frac{\|\tilde{\boldsymbol\sigma}_a-\boldsymbol\sigma_a\|_{l^2}}{\|\boldsymbol\sigma_a\|_{l^2}}, \quad\mbox{and}\quad
	\cE_H = \frac{\|\tilde\bH -\bH\|_{l^2}}{\|\bH\|_{l^2}}
\end{equation}
where $\boldsymbol \sigma_a\in \bbR^{N_\Omega}$ and $\tilde{\boldsymbol \sigma}_a\in \bbR^{N_\Omega}$ are the true and reconstructed absorption vectors respectively. Numerically the method appears to
have linear stability, with the piecewise constant case being slightly worse than the smooth one. This is typically due to
an imperfect resolution of the boundaries of the inclusions. Also, there is some residual error in the noiseless ($\epsilon=0$)
case due to the mismatch between the fine grid and the reconstruction grid.

\subsubsection{Reconstructing ($\Upsilon$, $\sigma_a$)}

To reconstruct both the Gr\"uneisen coefficient and the absorption coefficient, we use data collected from two collimated sources located on the left and right sides of the boundary respectively, $g_1(\bx,\bv)=\chi_{\partial\Omega|_L}\delta(\bv-\bv')$ and $g_2(\bx,\bv)=\chi_{\partial\Omega|_R}\delta(\bv+\bv')$ with $\bv'=(1,0)$. 

\begin{figure}[ht]
\centering
\includegraphics[angle=0,width=0.23\textwidth]{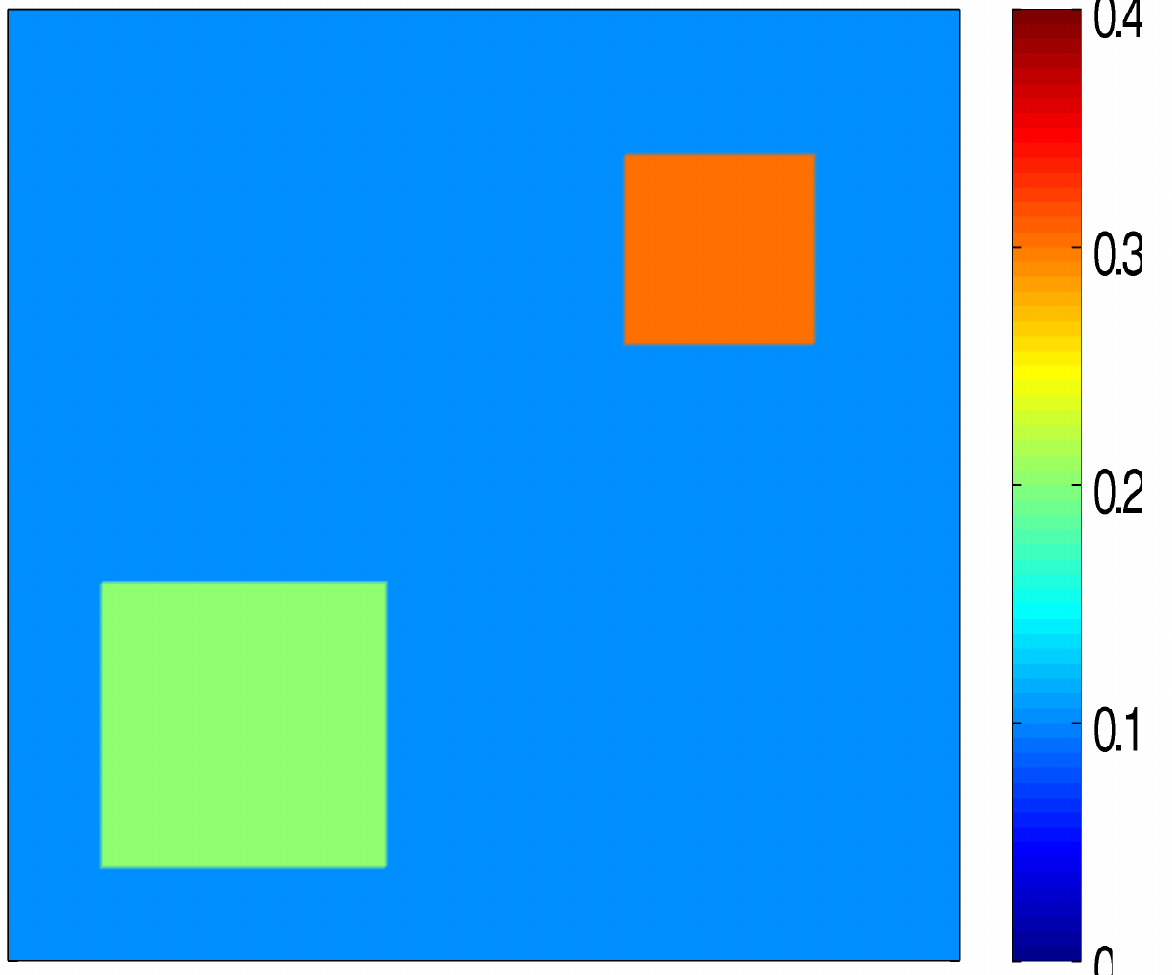} 
\includegraphics[angle=0,width=0.23\textwidth]{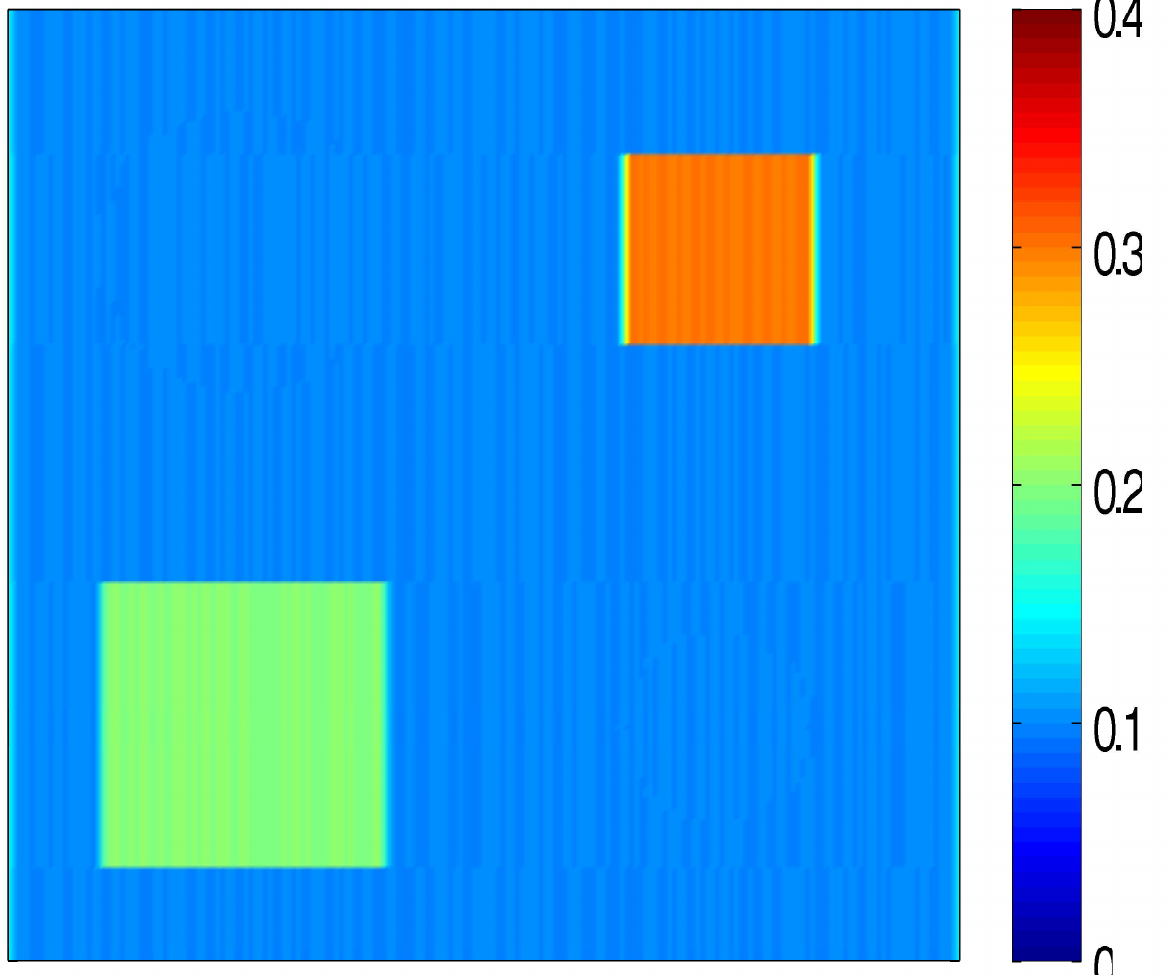} 
\includegraphics[angle=0,width=0.23\textwidth]{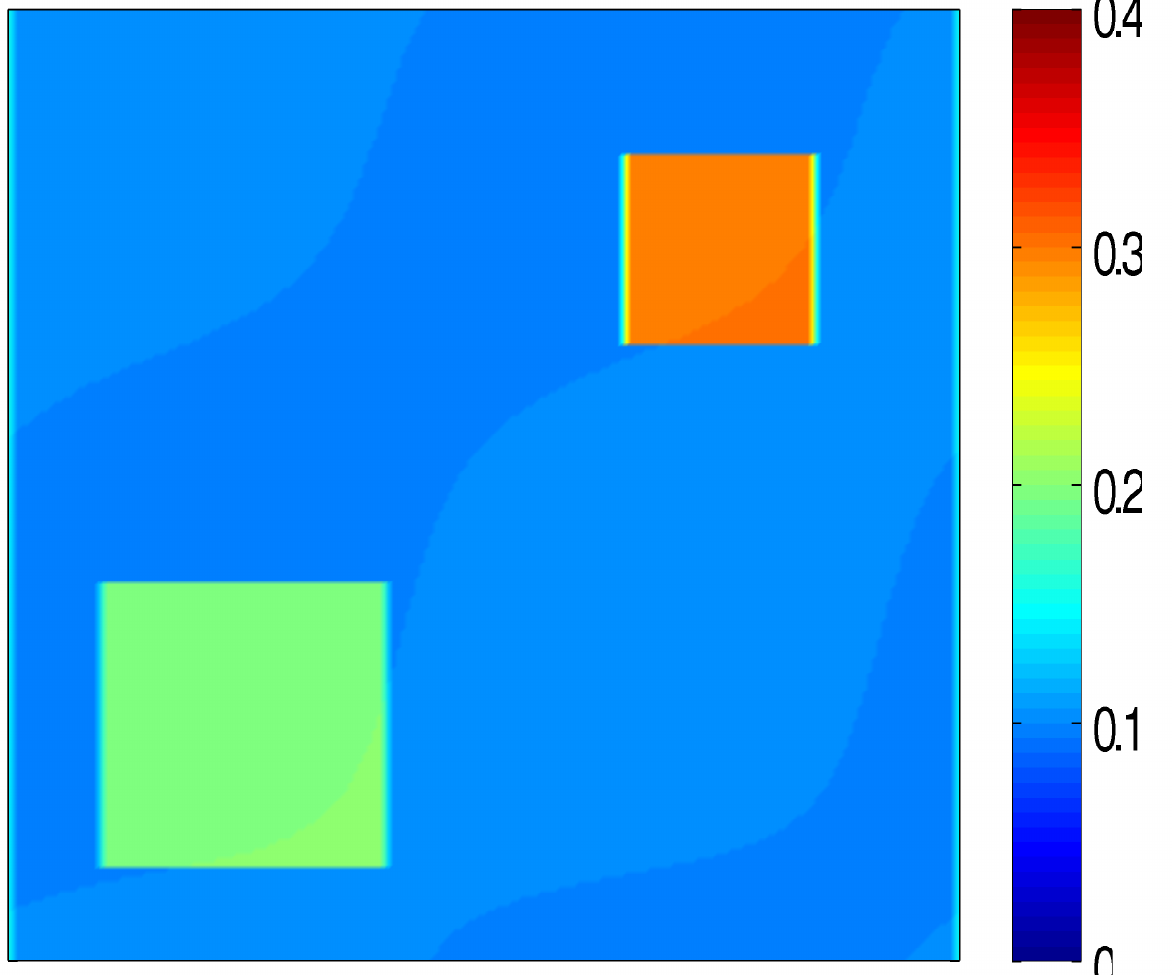}\\ 
\includegraphics[angle=0,width=0.23\textwidth]{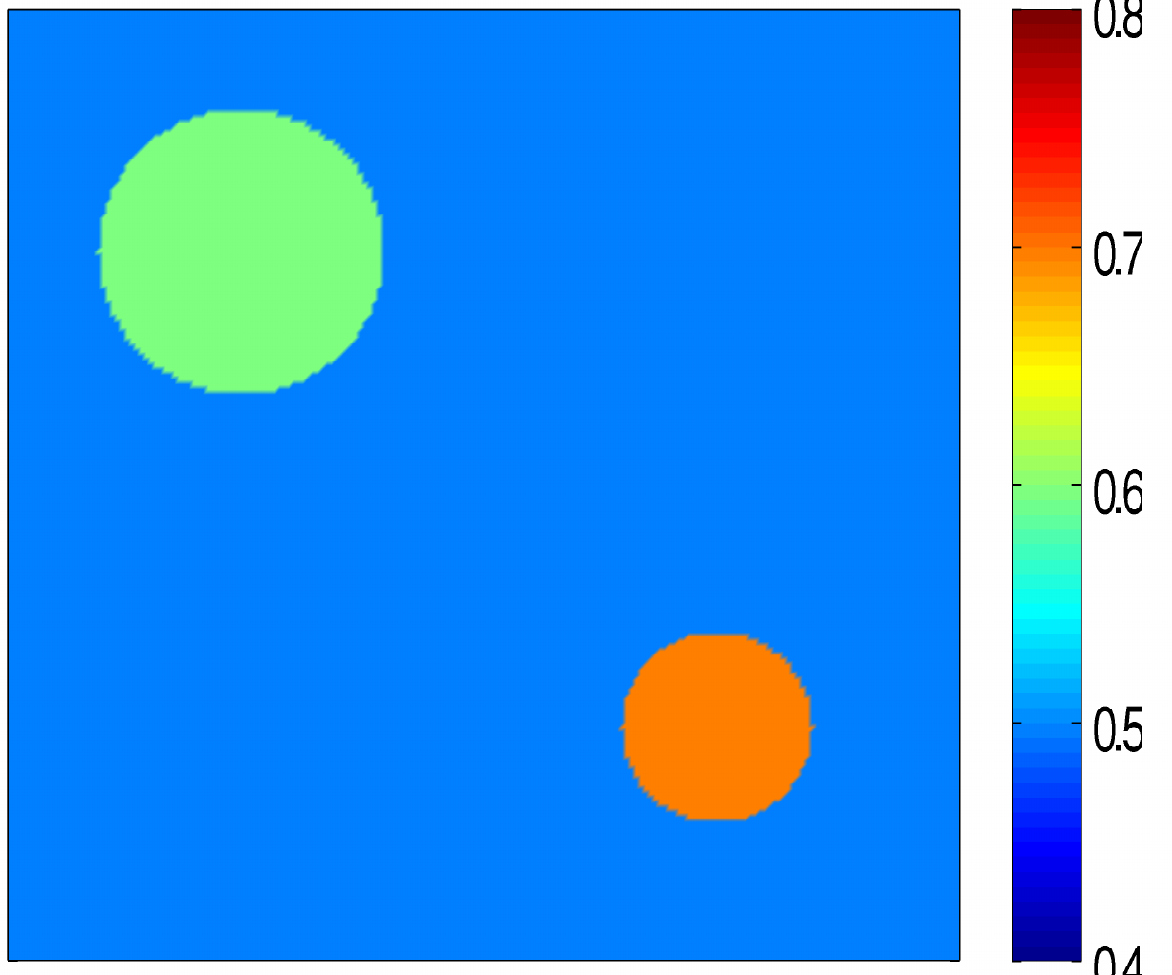} 
\includegraphics[angle=0,width=0.23\textwidth]{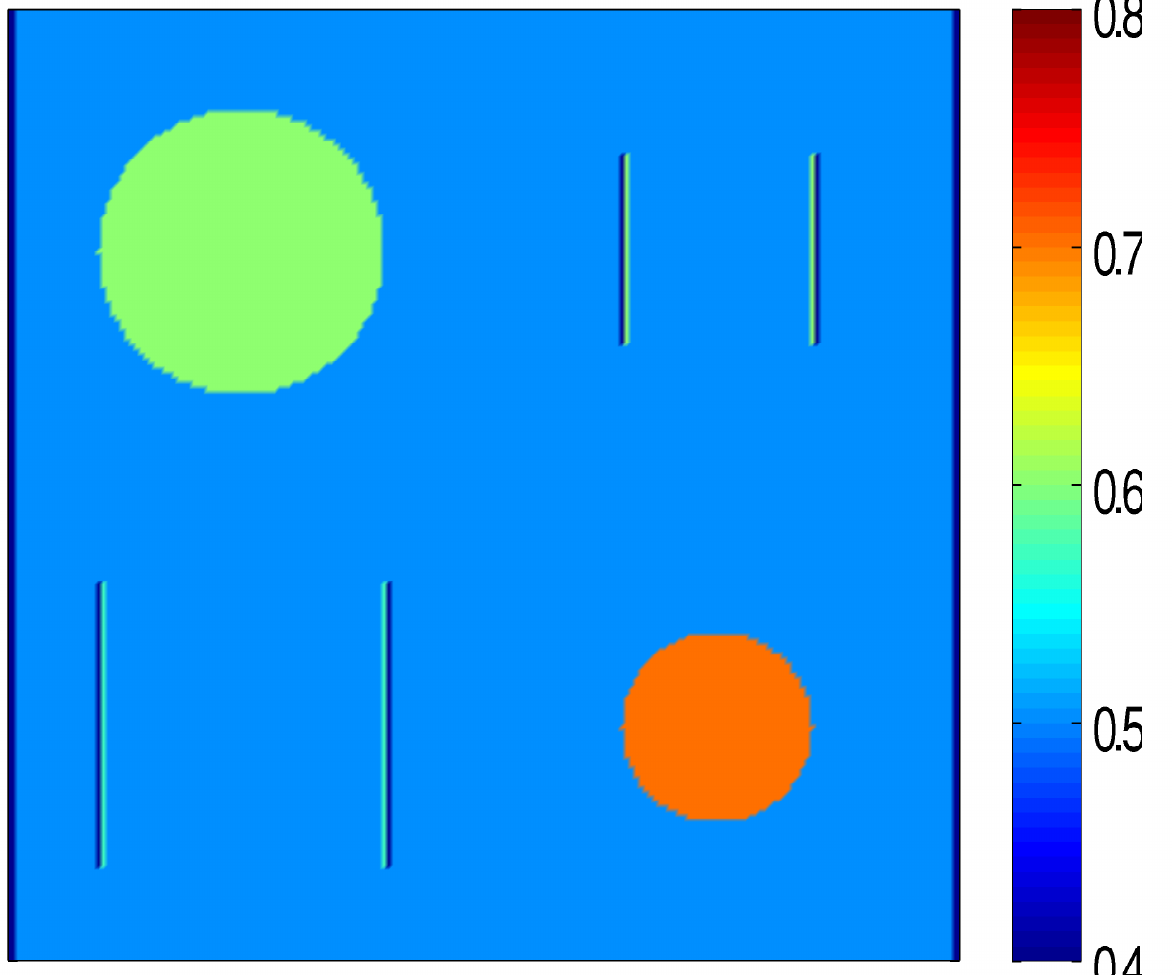} 
\includegraphics[angle=0,width=0.23\textwidth]{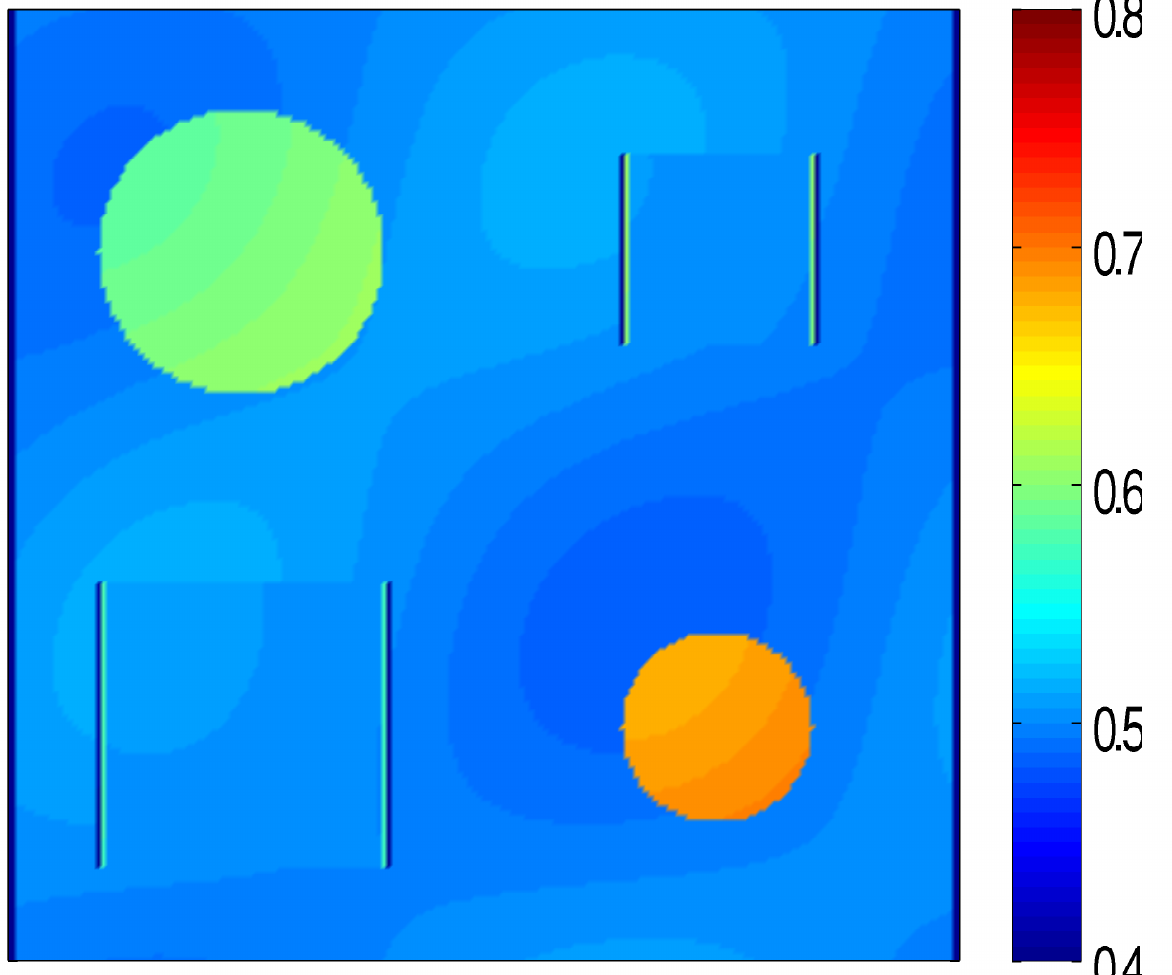} 
\caption{Reconstructions of $\sigma_a$ (top row) and $\Upsilon$ (bottom row) using a pair of collimated sources. Left to right: true coefficients, coefficients reconstructed with noiseless data and coefficients reconstructed with noisy data. The noisy data contains $5\%$ random noise ($\epsilon=0.05$).}
\label{FIG:Collimated GammaMua}
\end{figure}
The background absorption coefficient is $\sigma_a=0.1$ cm$^{-1}$ and the background Gr\"uneisen coefficient is $\Upsilon=0.5$. There are four inclusions, two for the absorption coefficient located in $\Omega_1=\{\bx \in \Omega \;|\; |\bx-(0.5,0.5)|\le 0.3\}$ and $\Omega_2=\{\bx \in \Omega \;|\; |\bx-(1.5,1.5)|\le 0.2\}$ respectively and two for the Gr\"uneisen coefficient located in $\Omega_3=\{\bx \in \Omega \;|\; |\bx-(0.5,1.5)|\le 0.2\}$ and $\Omega_4=\{\bx \in \Omega \;|\; |\bx-(1.5,0.5)|\le 0.3\}$ respectively. The coefficients inside the inclusions are ${\sigma_a}|_{\Omega_1}=0.2$ cm$^{-1}$, ${\sigma_a}|_{\Omega_2}=0.3$ cm$^{-1}$, $\Upsilon|_{\Omega_3}=0.6$ and ${\Upsilon}|_{\Omega_4}=0.7$ respectively. We perform the reconstruction with both noiseless data and noisy data polluted with $5$\% additive random noise. The reconstruction results are presented in Fig.~\ref{FIG:Collimated GammaMua}. Other than the phantom inclusions in the reconstructed Gr\"uneisen coefficient at the locations of the inclusions of the absorption coefficient, the quality of the reconstructions is very high, comparable to the previous reconstructions in the cases of one unknown coefficient. The phantom inclusions in the reconstructed Gr\"uneisen coefficients are caused by the inaccuracy of the reconstruction of the absorption coefficient at the boundary of square inclusions. This inaccuracy originates from the differentiation of the quantity $\ln\frac{H_2}{H_1}$ which contains noise coming from mismatch between the forward and inversion grids. For the reconstruction from noisy data, the noise added to the synthetic data in random but have only low frequency components. If the data contain very high frequency components as in the previous cases, numerical differentiation of the quantity $\ln\frac{H_2}{H_1}$ in the algorithm would yield even larger noise in the reconstruction that would bury the true coefficients. Averaging from multiple reconstructions would be necessary to get a clean image. We would not address this issue in detail now but leave it to future study. 

\subsection{Reconstructions in scattering media}
\label{SUBSEC:ScatterNumSetup}

Here we present some numerical simulations for scattering media, i.e. the case when $\sigma_s\neq 0$. In this case, we do not have analytical reconstruction formulas to work with. We thus use the linearized reconstruction and minimization-based nonlinear reconstruction schemes.

\subsubsection{Numerical setup}

The transport equations are solved with a finite volume scheme for the spatial variable and a discrete ordinate method for the angular variable. The domain is covered by $50\times 50$ cells of uniform size whose nodes are given by
\begin{displaymath}
	\Omega_h=\{\bx_{i,j}=(x_i,y_j)|\ x_i=i\Delta x,\ y_j=j\Delta y,\ i, j = 0, 1, \ldots, 50\},
\end{displaymath}
with $\Delta x=\Delta y=0.04$.  We discretize $\bbS^1$ into $128$ uniformly distributed directions with
identical quadrature weight:
\begin{displaymath}
	\bbS^1_{\Delta\theta}=\{\bv_k:\bv_k=(k-1)*\Delta\theta,\ k=1, \ldots, 128\},
\end{displaymath}
where $\Delta\theta=2\pi/128$. Note that this discretization is not as fine as that used in the previous case when scattering is not presented in the problem. This is only done here to save computational cost. It should not be regarded as a limitation of the reconstruction algorithms for scattering media.

The scattering kernel is chosen as the Henyey-Greenstein phase function~\cite{Arridge-IP99,HeGr-AJ41,WeVa-Book95}
\begin{equation}\label{EQ:HG}
	\cK(\bv,\bv')=\dfrac{1}{2\pi}\frac{1-\eta^2}{(1+\eta^2-2\eta\bv\cdot\bv')^{3/2}},
\end{equation}
where $\eta\in[0,1]$ is the anisotropy factor, which measures how peaked forward the phase function is. The larger $\eta$ is, the more forward the scattering.  The anisotropy factor is often used to define the so-called effective scattering coefficient through $\sigma_s'=(1-\eta)\sigma_s$. 

For more details on the forward solver, we refer to our previous publications~\cite{ReAbBaHi-OL04,ReBaHi-SIAM06}. 

\subsubsection{Reconstructing $\sigma_a$}
\label{SUBSEC:1Coeff}

\begin{figure}[ht]
\centering
\includegraphics[angle=0,width=0.25\textwidth]{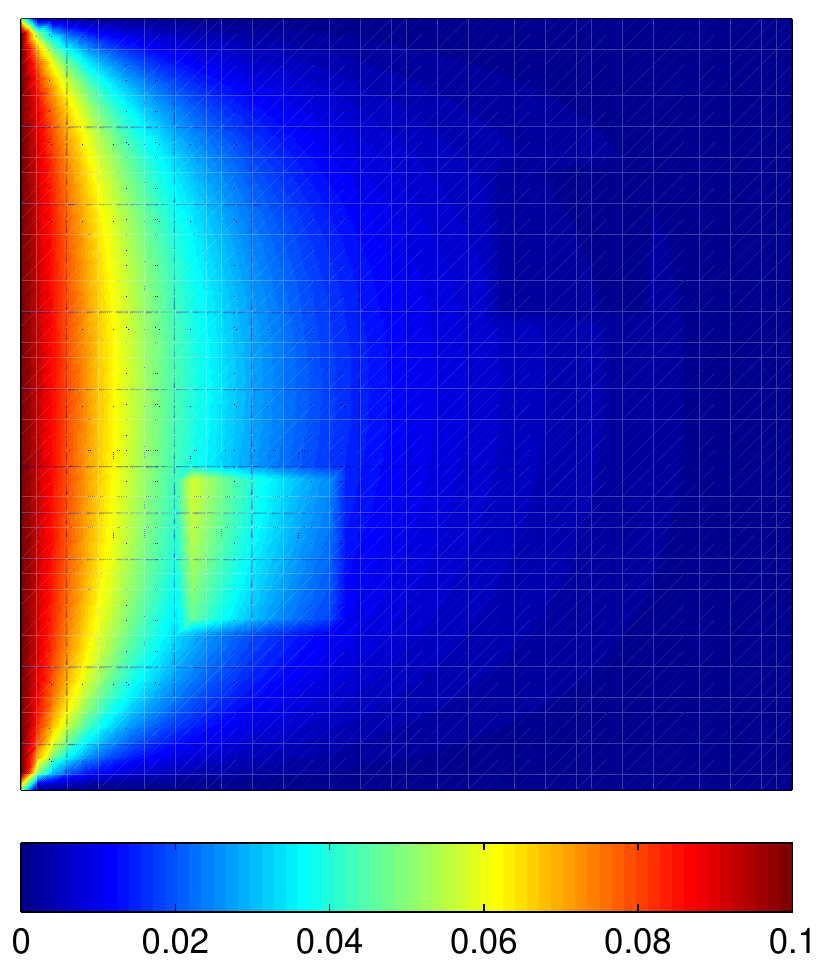}
\includegraphics[angle=0,width=0.25\textwidth]{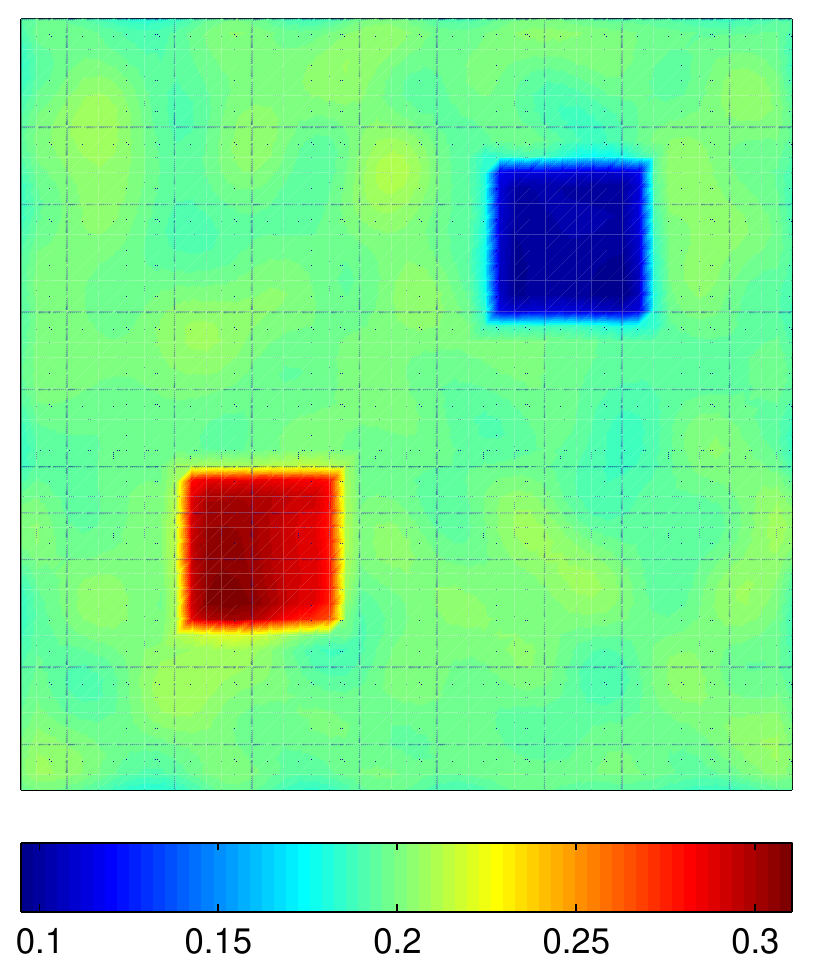}
\includegraphics[angle=0,width=0.25\textwidth]{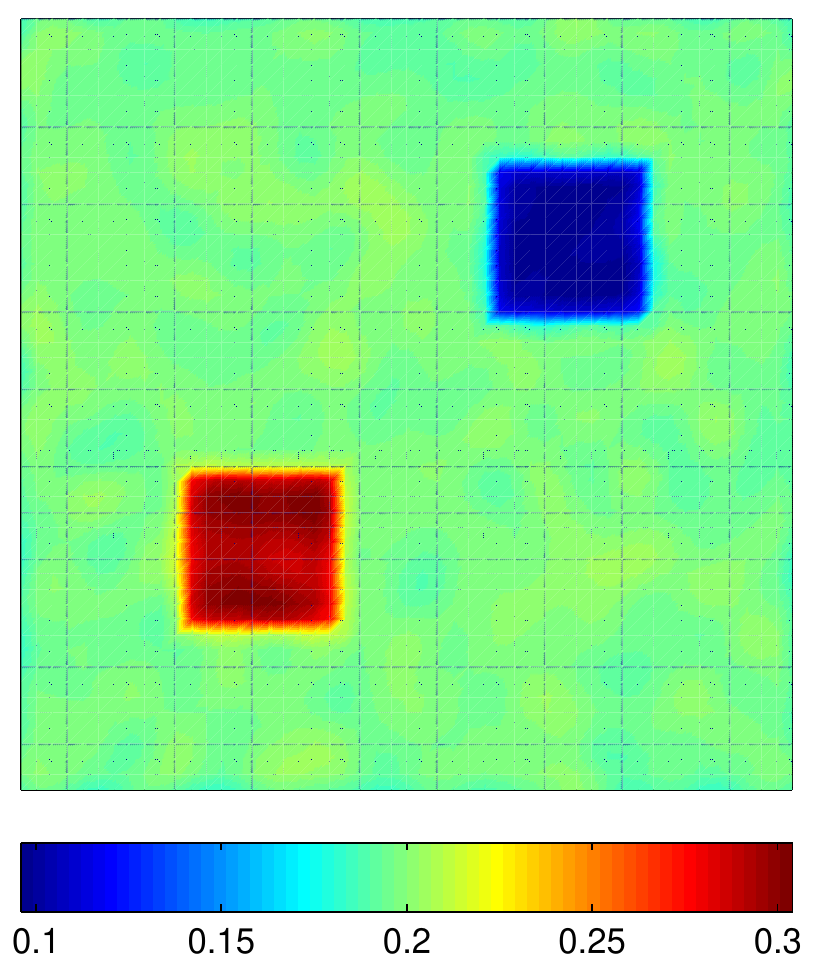}\\
\includegraphics[angle=0,width=0.25\textwidth]{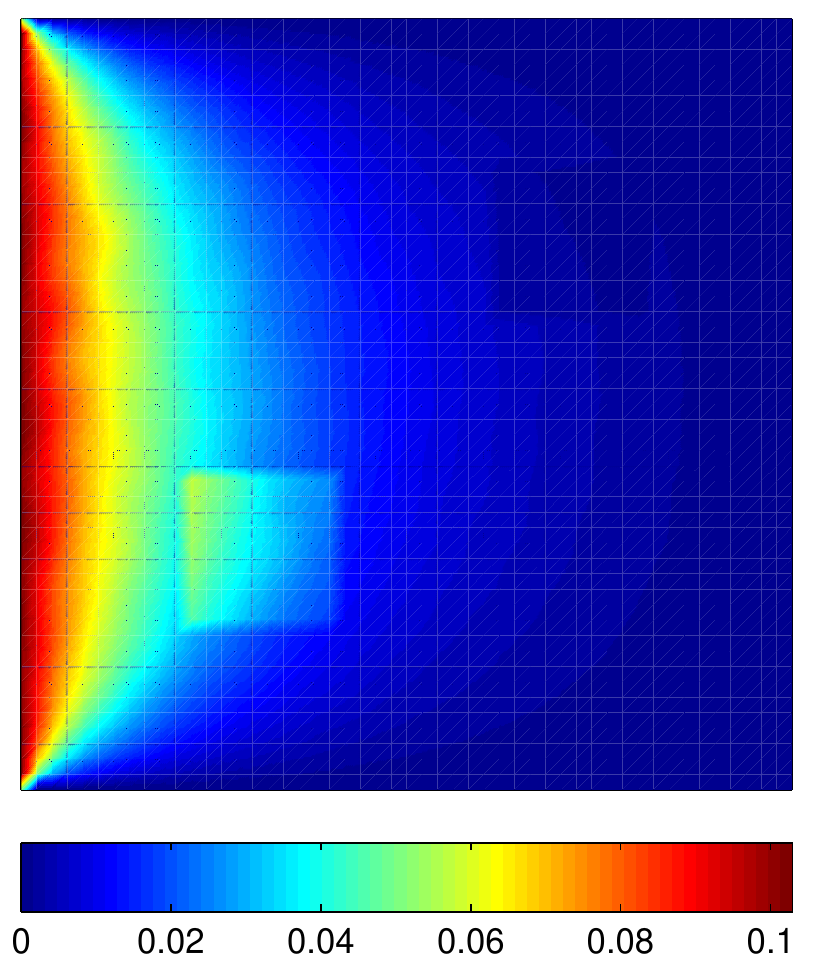}
\includegraphics[angle=0,width=0.25\textwidth]{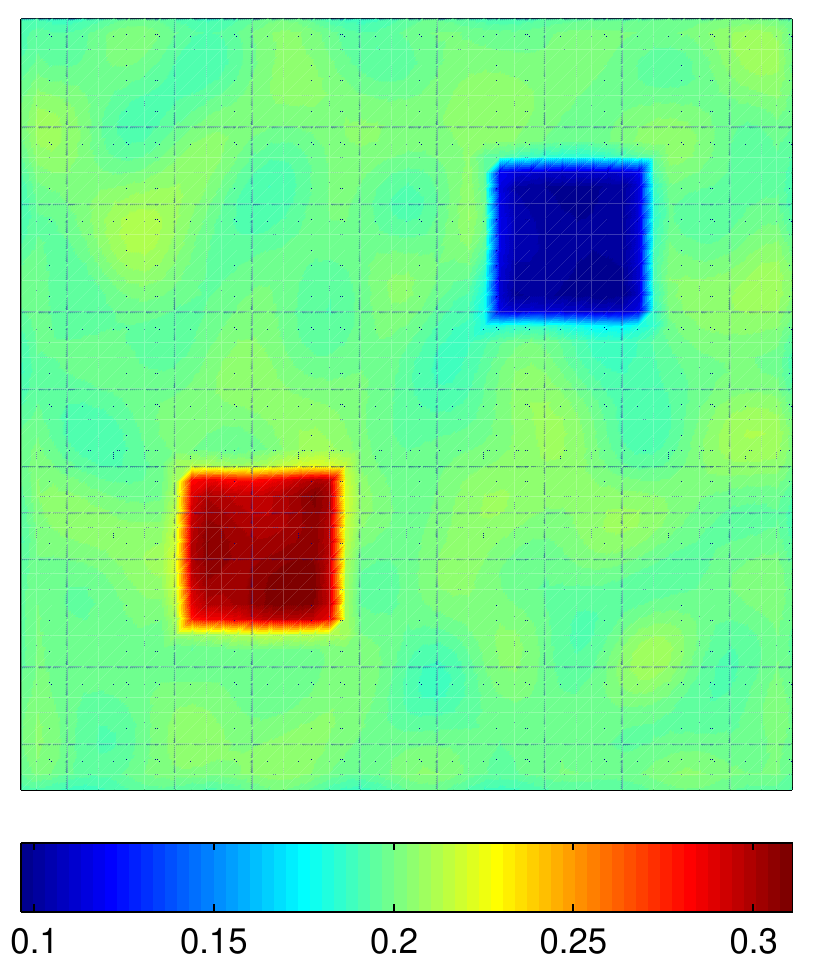}
\includegraphics[angle=0,width=0.25\textwidth]{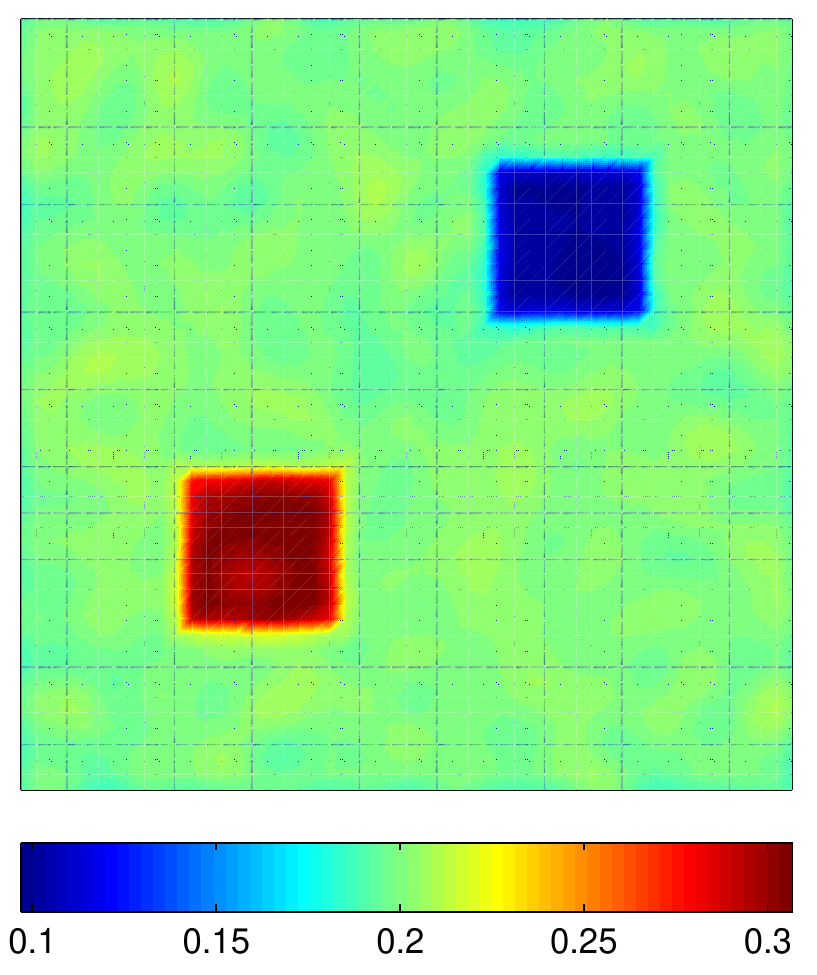}
\caption{Reconstructions of the absorption coefficient. Top row: the data $H=\Upsilon\sigma_a \aver{u}_\bv$ (left), $\sigma_a$ reconstructed using Born approximation (middle) and $\sigma_a$ reconstructed using nonlinear iteration (right) in an isotropic medium; Bottom row: same as the top row but in an anisotropic medium.}
\label{FIG:Mua}
\end{figure}
We first reconstruct the absorption coefficient, assuming that both the scattering coefficient and the Gr\"uneisen coefficients are known. The setup is as follows. The medium consists of a homogeneous background absorbing medium with absorption coefficient $\sigma_a=0.2$ cm$^{-1}$ and two absorbing inclusions. The first inclusion occupies $\Omega_1=[0.4,0.8]\times[0.4,0.8]$ with the absorption coefficient $\sigma_a=0.3$ cm$^{-1}$. The second inclusion occupies $\Omega_2=[1.2,1.6]\times [1.2,1.6]$ with the absorption coefficient $\sigma_a=0.1$ cm$^{-1}$. We first performed two reconstructions with the linearization method: a reconstruction in an isotropically scattering medium with $\eta=0$, $\sigma_s=8$ cm$^{-1}$ and a reconstruction in an anisotropically scattering medium with $\eta=0.9$, $\sigma_s=80$ cm$^{-1}$. The Gr\"uneisen coefficient is always $\Upsilon=0.5$. The synthetic data used in these reconstructions is noiseless in the sense described at the beginning of this section. The results of the reconstructions are shown in Fig.~\ref{FIG:Mua}. The quality of the reconstructions is very high (although shown on a coarser grid than those in the previous section). Note that the fast decay of field from the line source makes the second inclusion hardly visible in the data $H=\Upsilon\sigma_a\aver{u}_\bv$ plots. This is one of the main reason that the quantitative step of PAT is needed. The reconstructions using the nonlinear reconstruction algorithm are presented in the right column of Fig.~\ref{FIG:Mua}.

\subsubsection{Reconstructing ($\Upsilon$, $\sigma_a$)}
\label{SUBSEC:2Coeff}

\begin{figure}[ht]
\centering
\includegraphics[angle=0,width=0.48\textwidth]{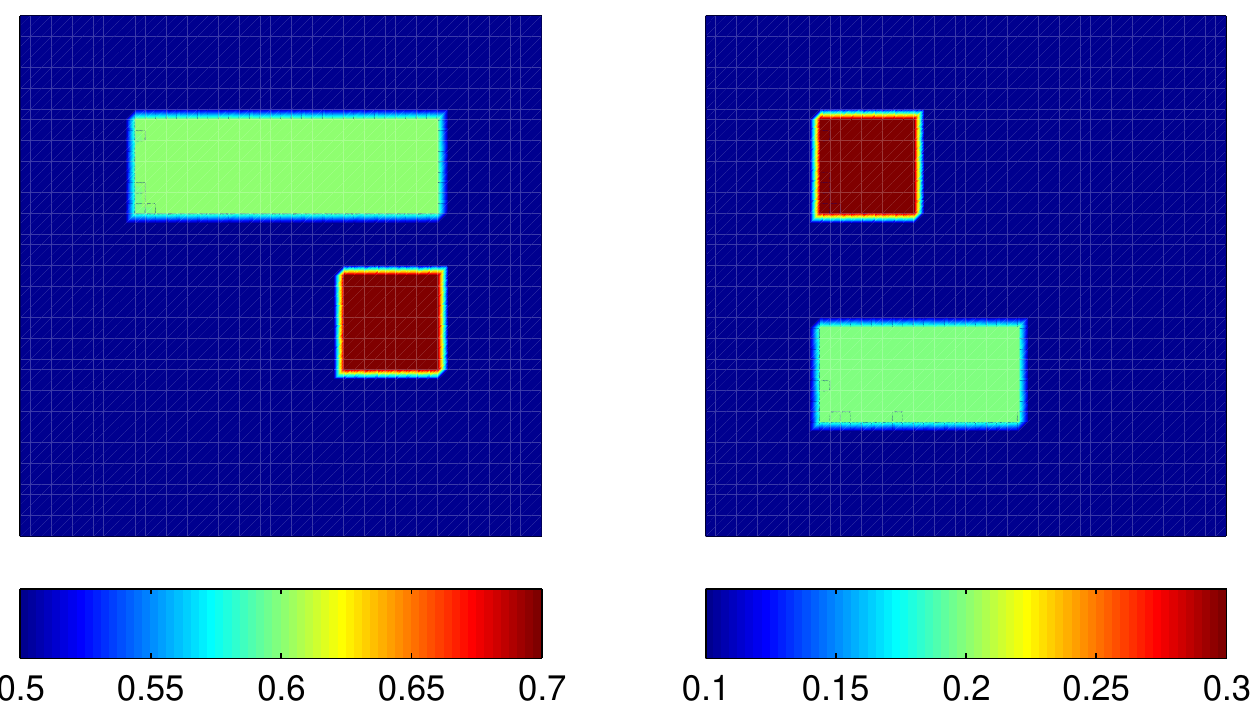}
\includegraphics[angle=0,width=0.48\textwidth]{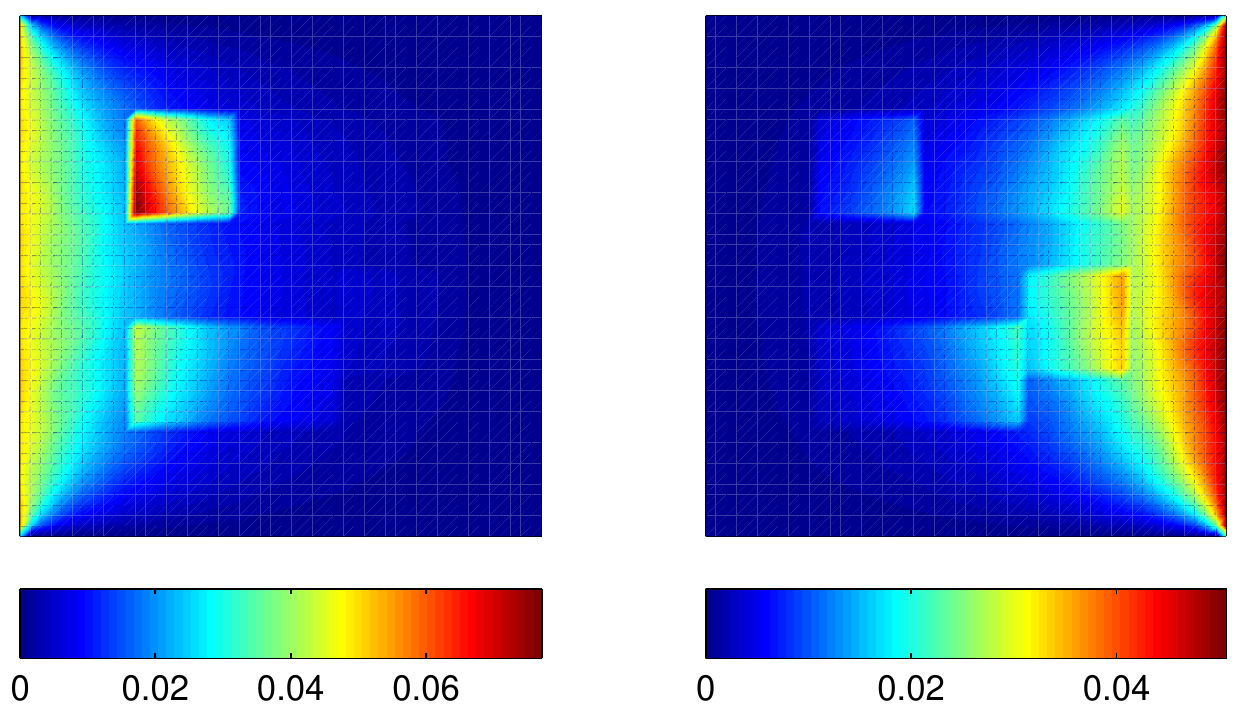}\\
\includegraphics[angle=0,width=0.48\textwidth]{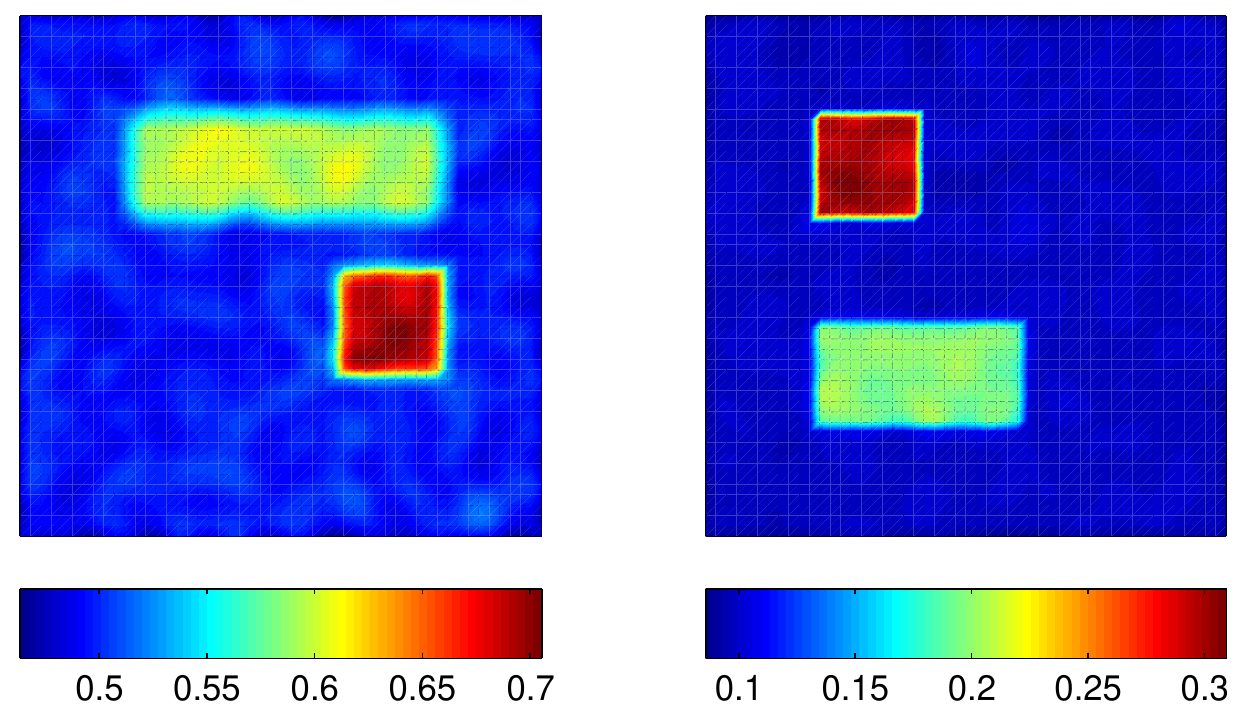}
\includegraphics[angle=0,width=0.48\textwidth]{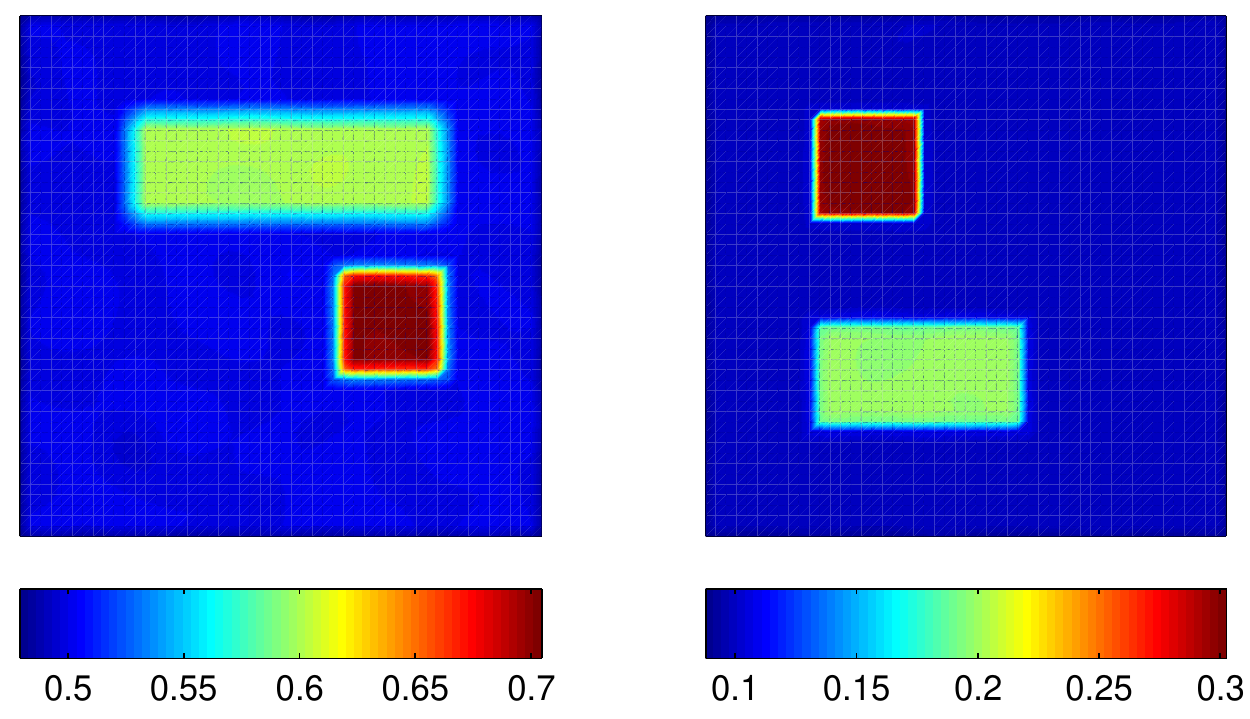}
\caption{Reconstructions of the Gr\"uneisen and the absorption coefficients. Top left: true coefficients ($\Upsilon$, $\sigma_a$) Top right: two data sets ($H_1$, $H_2$) used in the reconstruction; Bottom left: reconstructed ($\Upsilon$, $\sigma_a$) with two data sets; Bottom right: reconstructed ($\Upsilon$, $\sigma_a$) with eight data sets.}
\label{FIG:UpsilonMua}
\end{figure}
We now perform numerical simulations where we reconstruct both the Gr\"uneisen coefficient and the absorption coefficient. The scattering coefficient is fixed at $\sigma_s=8$ cm$^{-1}$ and the anisotropic factor $\eta=0$. The background absorption coefficient is $\sigma_a=0.1$ cm$^{-1}$ and the background Gr\"uneisen coefficient is $\Upsilon=0.5$. There are four inclusions, two for the absorption coefficient located in $\Omega_1=[0.4,1.2]\times[0.4,0.8]$ and $\Omega_2=[0.4, 0.8]\times [1.2, 1.6]$ respectively and two for the Gr\"uneisen coefficient located in $\Omega_3=[0.4,1.6]\times[1.2,1.6]$ and $\Omega_4=[1.2,1.6]\times[0.6,1.0]$ respectively. 
The coefficients inside the inclusions are ${\sigma_a}|_{\Omega_1}=0.2$ cm$^{-1}$, ${\sigma_a}|_{\Omega_2}=0.3$ cm$^{-1}$, $\Upsilon|_{\Omega_3}=0.6$ and ${\Upsilon}|_{\Omega_4}=0.7$ respectively. We perform the reconstruction with data polluted with 5\% additive random noise added the same way as in the previous section. The coefficients ($\Gamma$,$\sigma_a$) reconstructed with two data sets generated with sources $g_1(\bx,\bv)=\chi_{\partial\Omega_L}$ and $g_2(\bx,\bv)=\chi_{\partial\Omega_R}$ are shown in Fig.~\ref{FIG:UpsilonMua} (the left two plots in the bottom row). The quality of the reconstruction is quite high again even though it is slightly lower than that in the reconstructions in the non-scattering regime, such as those shown in Fig.~\ref{FIG:Collimated GammaMua}. The quality of the reconstruction can be improved by averaging out noise in the data through the usage of additional data sets. This is demonstrated in Fig.~\ref{FIG:UpsilonMua} (the right two plots in the bottom row) where we show the reconstruction of the coefficients ($\Upsilon$,$\sigma_a$) using eight data sets. The eight sources used are the rotation of the source $g_1(\bx)=\chi_{\partial\Omega_L}$ and $g_2=\chi_{\partial\Omega_L}\delta(\bv-(1,0))$ over $0$, $\pi/2$, $\pi$, and $3\pi/2$ around the center of the square domain $\Omega$.

\subsubsection{Reconstructing ($\sigma_a$, $\sigma_s$)}
\label{SUBSEC:2Coeff2}

\begin{figure}[ht]
\centering
\includegraphics[angle=0,width=0.48\textwidth]{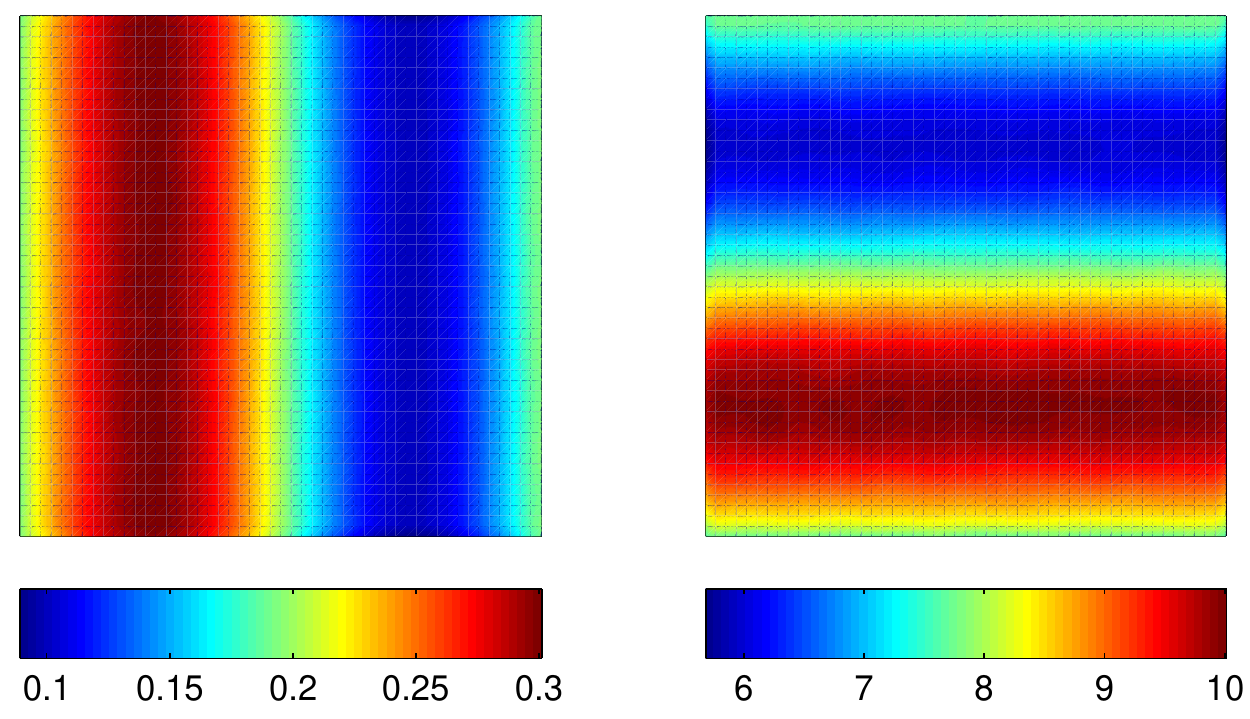}
\includegraphics[angle=0,width=0.48\textwidth]{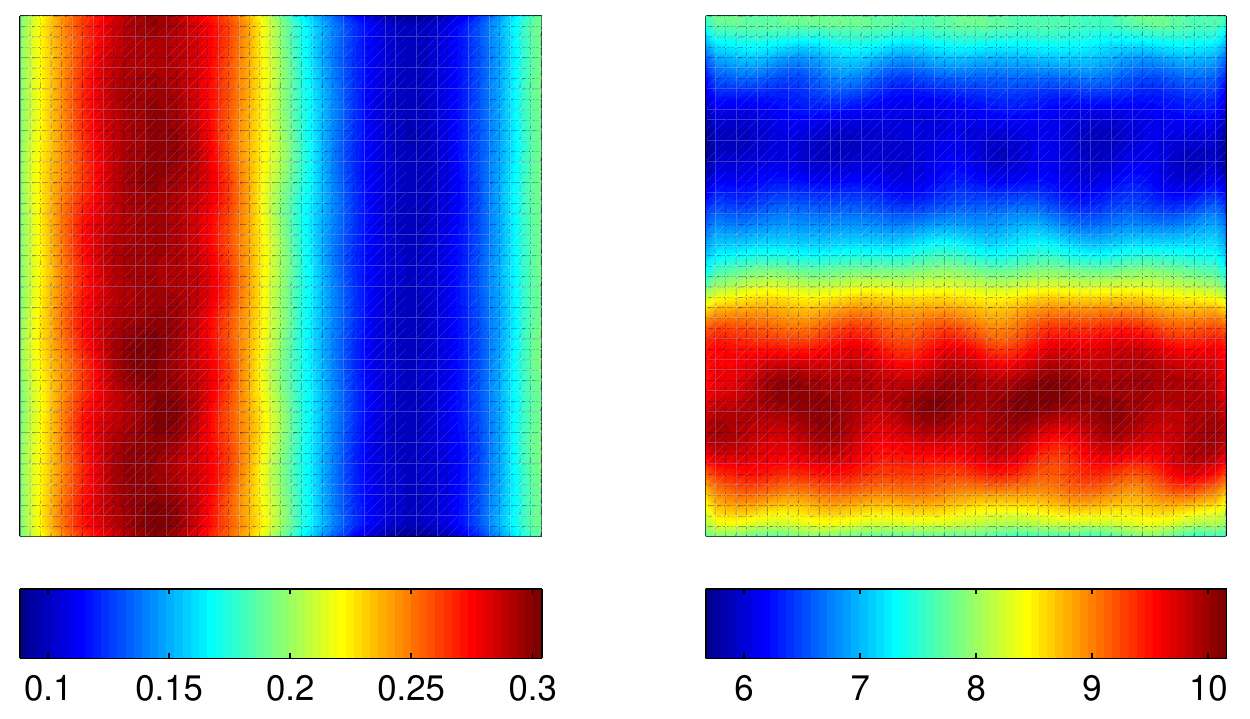}\\
\includegraphics[angle=0,width=0.48\textwidth]{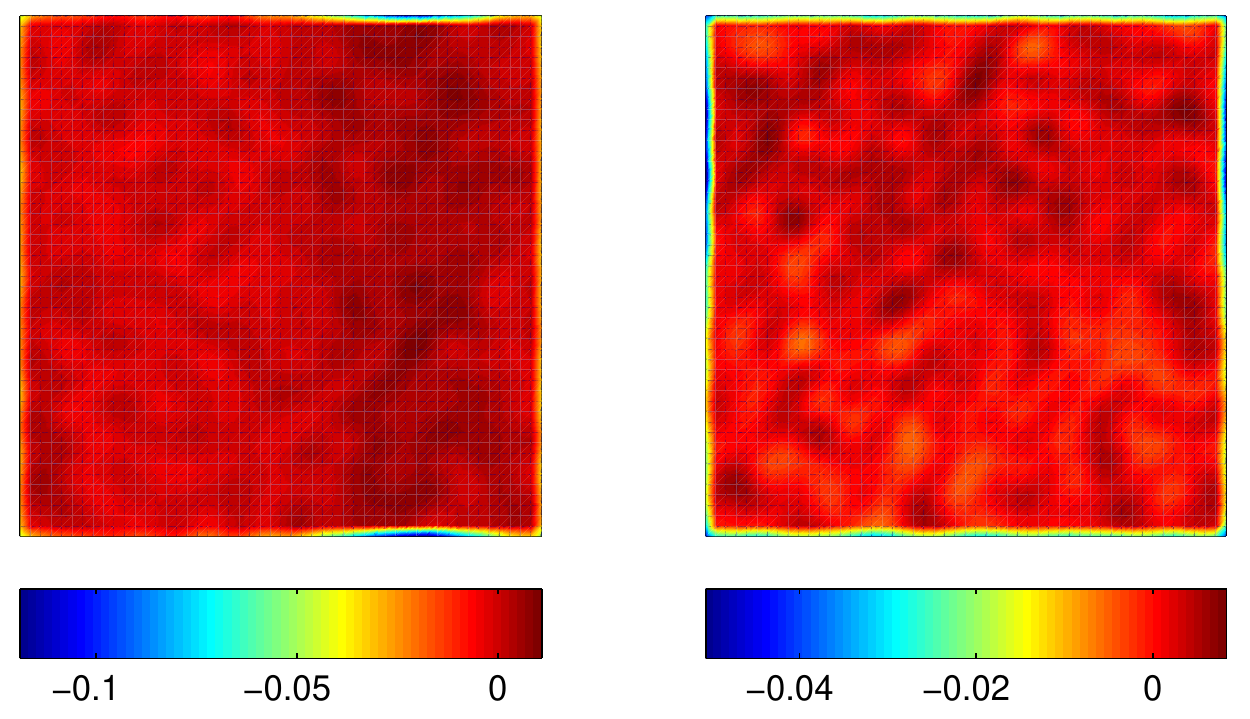}
\includegraphics[angle=0,width=0.48\textwidth]{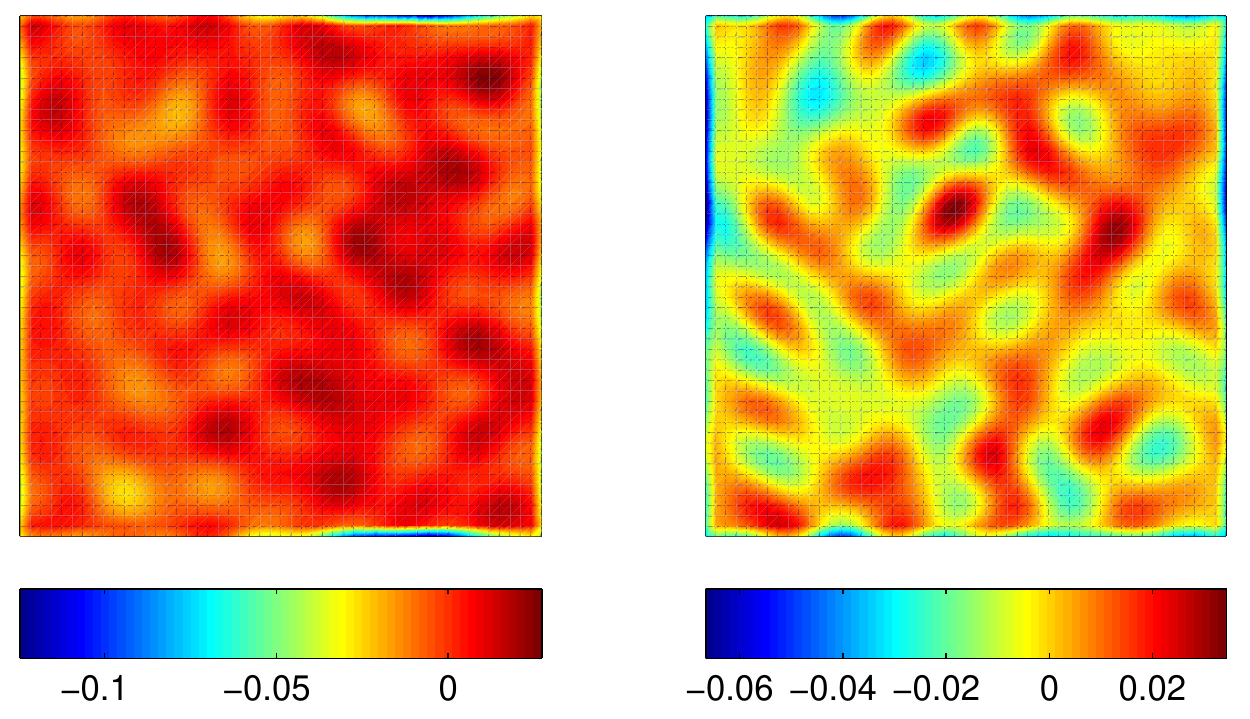}
\caption{Reconstructed absorption and scattering coefficients given in~\eqref{EQ:Abs Coeff1}. 
Top row: ($\sigma_a$, $\sigma_s$) reconstructed with noiseless (left two plots) and noisy (right two plots) data; Bottom row: the corresponding relative differences between reconstructed and real coefficients.}
\label{FIG:MuaMus-C}
\end{figure}
We now fix the Gr\"uneisen coefficient $\Upsilon=1$ and attempt to reconstruct both the absorption and the scattering coefficient. We show in Fig.~\ref{FIG:MuaMus-C} the reconstructions of the objective absorption and scattering coefficients:
\begin{equation}\label{EQ:Abs Coeff1}
	\sigma_a(x,y)=0.2+0.1 \sin(\pi x), \qquad 
	\sigma_s(x,y)=8.0+2.0 \sin(\pi y) .
\end{equation}
Shown are reconstructed ($\sigma_a$,$\sigma_s$) (top row) with eight noiseless and noisy data sets and the corresponding pointwise relative error (bottom row), defined for $\sigma_a$ (resp. $\sigma_s$) as $\frac{\tilde{\boldsymbol\sigma}_a-\boldsymbol\sigma_a}{\boldsymbol\sigma_a}$ (resp. $\frac{\tilde{\boldsymbol\sigma}_s-\boldsymbol\sigma_s}{\boldsymbol\sigma_s}$), in the reconstructions. Except for the relatively large error on some part of the boundary, the error in the reconstructions is comparable to the noise level in the data. This is also confirmed in the reconstruction of piecewise constant coefficients such as those shown in Fig.~\ref{FIG:MuaMus-D}. The piecewise constant absorption coefficient but the smooth scattering coefficient are given as
\begin{equation}\label{EQ:Abs Coeff2}
	\sigma_a(x,y)=0.2+0.1 \sin(\pi x), \qquad 
	\sigma_s(x,y)=8.0+2.0 \sin(\pi y)
\end{equation}
where the first absorbing inclusion is located in $\Omega_1=[0.4,0.8]\times[0.8,1.2]$ and the second absorbing inclusion is located in $\Omega_2=[1.4,1.6]\times[0.2,1.8]$. For piecewise constant coefficient, the reconstruction error occurs on the boundary of the inclusions. This is true also in the noiseless data case due to the fact that the synthetic data is generated by averaging quantities on a finer mesh. The Tikhonov regularization we employed in the numerical schemes also contribute to the smoothing on the boundary of the inclusions.
\begin{figure}[ht]
\centering
\includegraphics[angle=0,width=0.48\textwidth]{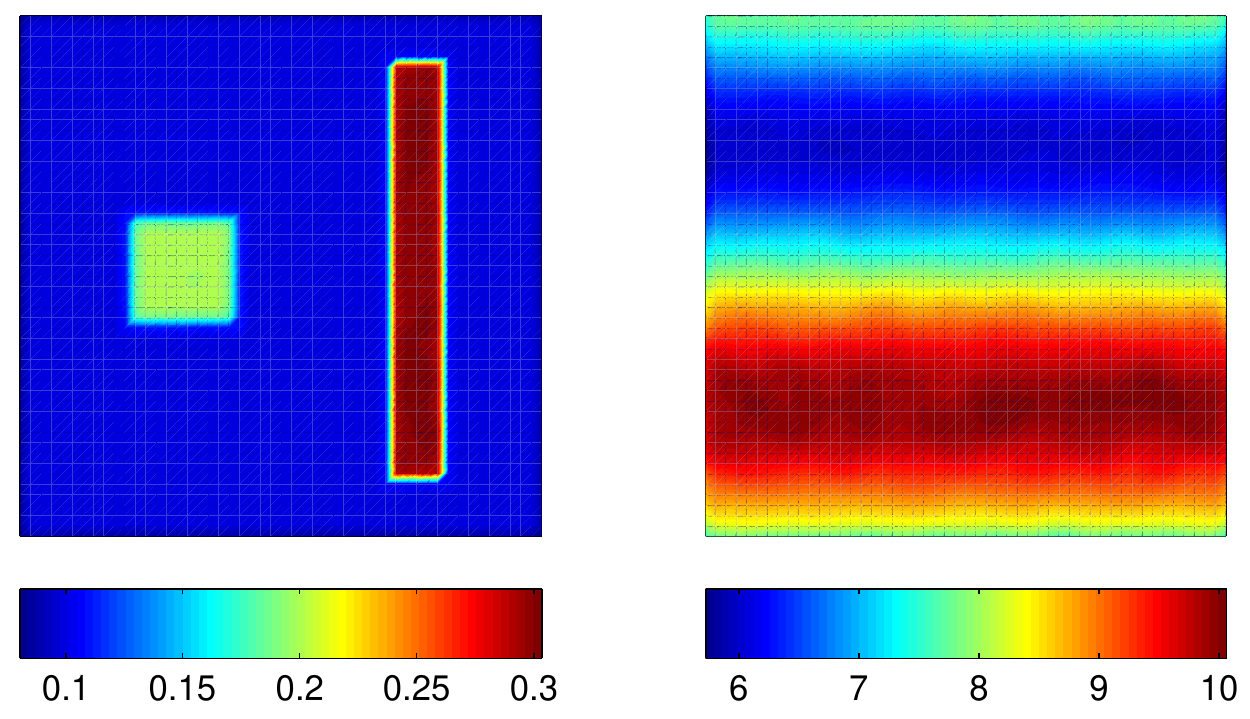}
\includegraphics[angle=0,width=0.48\textwidth]{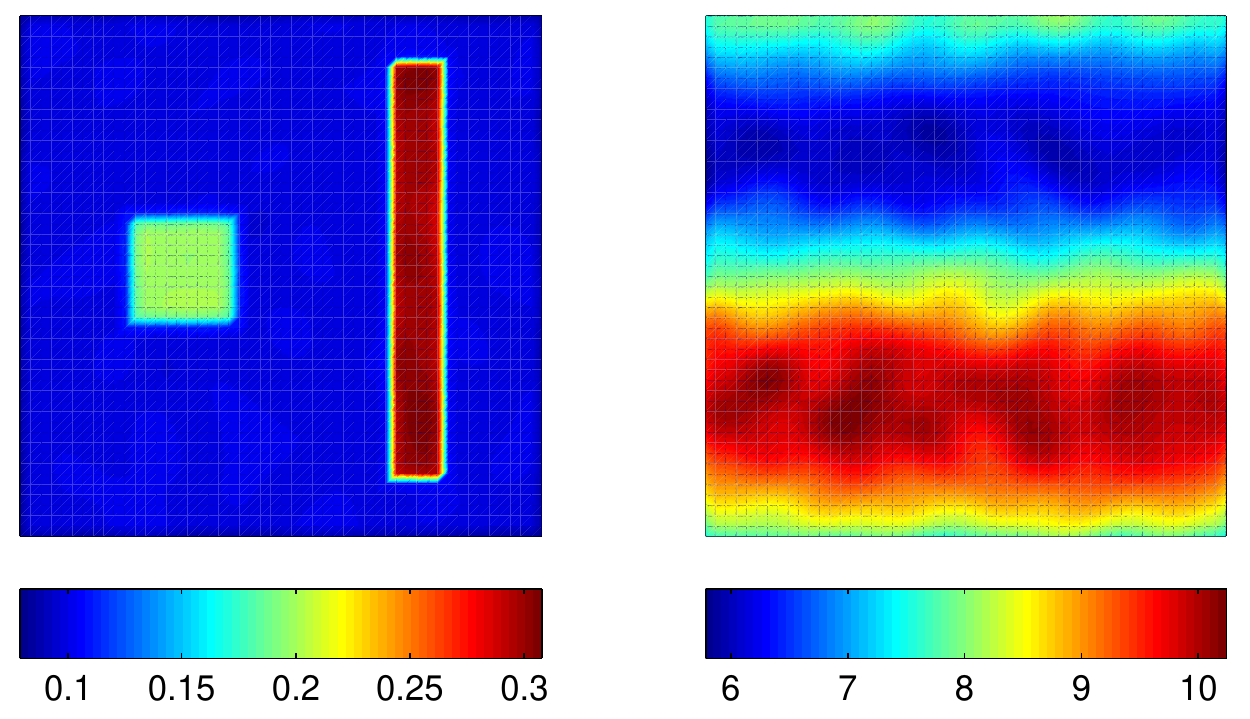}\\
\includegraphics[angle=0,width=0.48\textwidth]{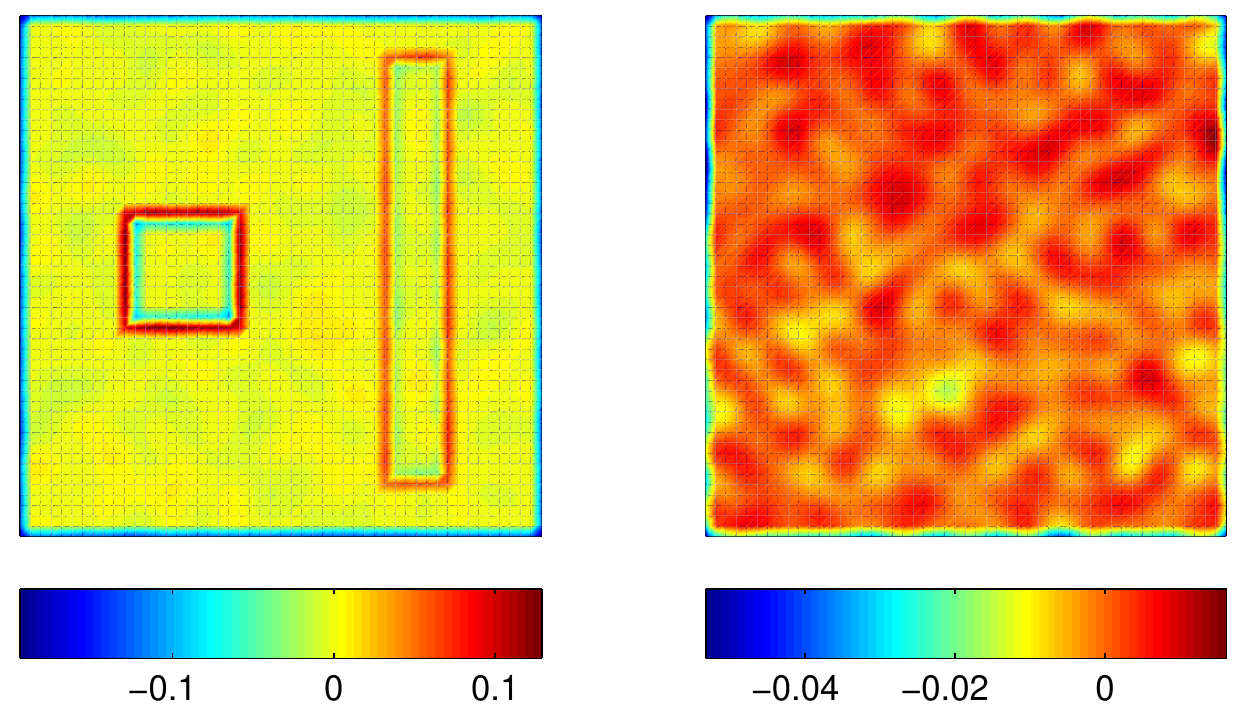}
\includegraphics[angle=0,width=0.48\textwidth]{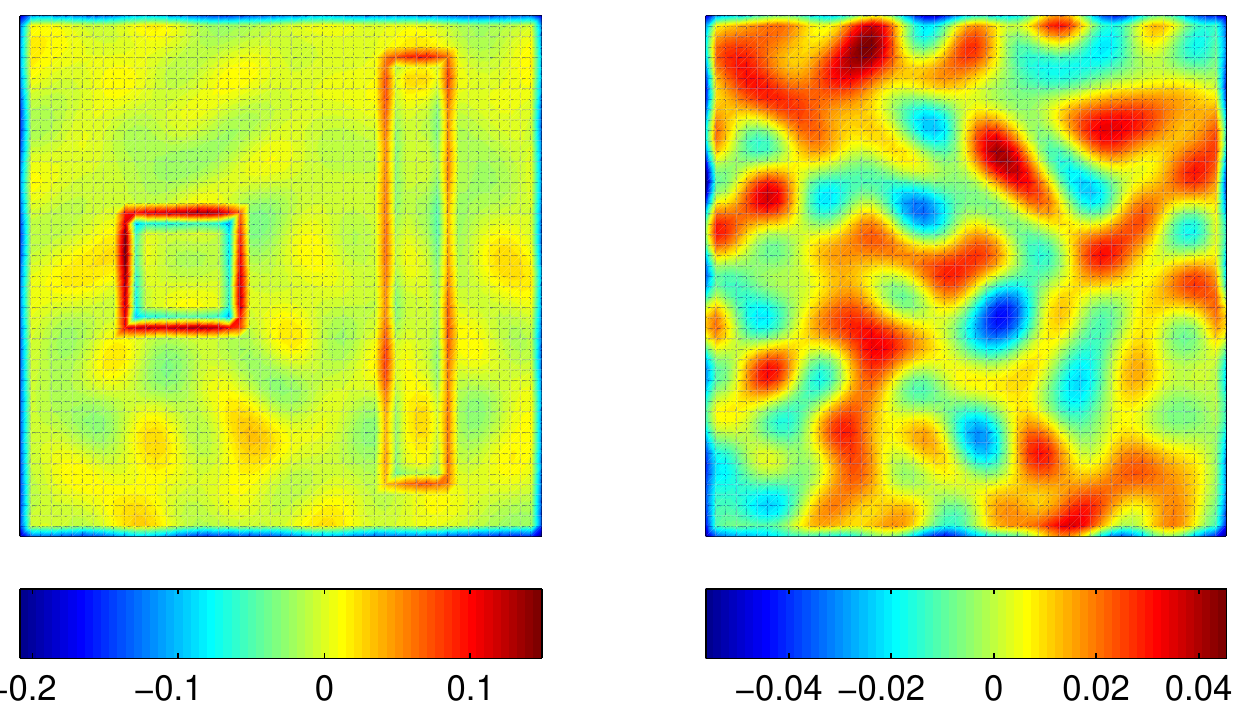}
\caption{Reconstructed absorption and scattering coefficients given in~\eqref{EQ:Abs Coeff2}. 
Top row: ($\sigma_a$, $\sigma_s$) reconstructed with noiseless (left two plots) and noisy (right two plots) data; Bottom row: the corresponding relative differences between reconstructed and real coefficients.}
\label{FIG:MuaMus-D}
\end{figure}

\subsubsection{Reconstructing ($\Upsilon$, $\sigma_a$, $\sigma_s$) with multi-spectral data}
\label{SUBSEC:3Coeff}

\begin{figure}[ht]
\centering
\includegraphics[angle=0,width=0.96\textwidth]{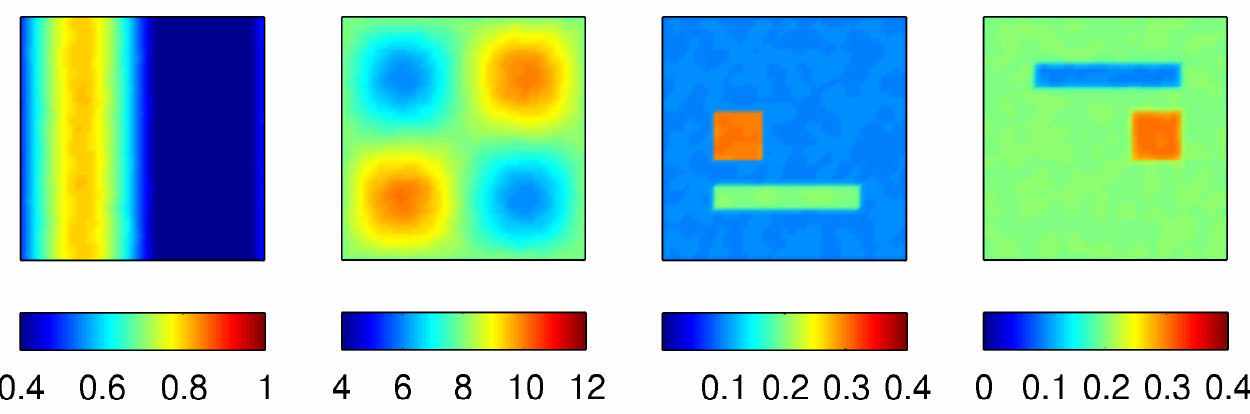}
\includegraphics[angle=0,width=0.96\textwidth]{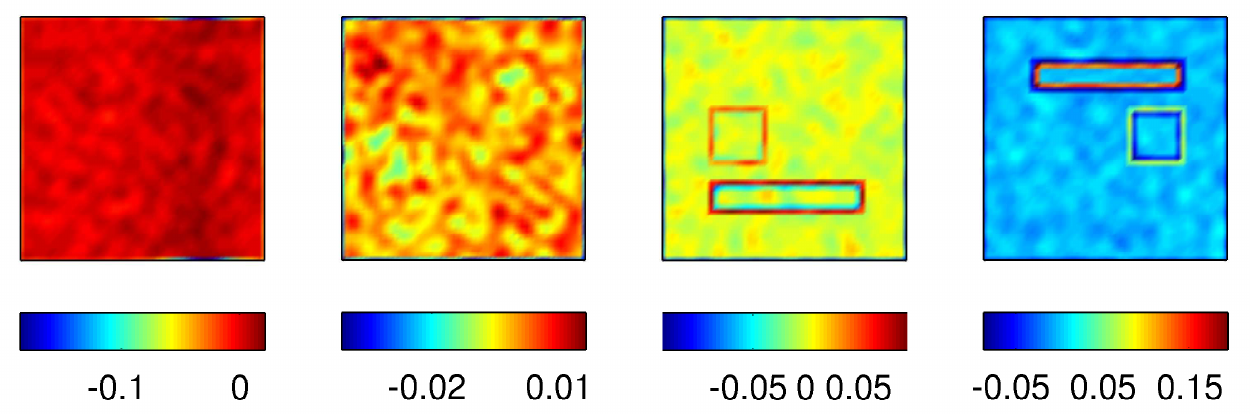}
\caption{Simultaneous reconstruction of the absorption and the scattering and the Gr\"uneisen coefficients with multi-spectral data. Top row: reconstructed $\Upsilon$, $\sigma_s$, $\sigma_a^1$ and $\sigma_a^2$; Bottom row: the corresponding pointwise relative error in the reconstructions. The data used in this simulation are noiseless.}
\label{FIG:Spec-C}
\end{figure}
The last numerical simulation is devoted to the simultaneous reconstructions of all three coefficients with multi-spectral interior data. The coefficients take the forms given in~\eqref{EQ:Model II}. To simplify the presentation, we consider only an absorption coefficient that has two components, i.e. $K=2$. The spectral components of the coefficients are given as follows:
\begin{equation}
	\alpha_1(\lambda)=\dfrac{\lambda}{\lambda_0},\quad  \alpha_2(\lambda)=\dfrac{\lambda_0}{\lambda}, \quad \beta(\lambda)= \left( \dfrac{\lambda}{\lambda_0} \right)^{3/2}, \quad \gamma(\lambda) = 1,
\end{equation}
where the normalization wavelength $\lambda_0$ is included to control the amplitude of coefficients. 
These weight functions are by no means what they should exactly be in practical applications. However, the specific forms do not have impact on the results of the reconstruction. The spatial components of the coefficients are given as
\begin{equation}
	\begin{array}{c}
		\Upsilon(\bx)=0.5+0.4\tanh(4x-4), \qquad \sigma_s(\bx)=8.0+2.0\sin(\pi x)\sin(\pi y)\\
		\sigma_a^1(\bx)= 0.1+0.1\chi_{\Omega_1}+0.2\chi_{\Omega_2}, \qquad \sigma_a^2(\bx)=0.2-0.1\chi_{\Omega_3}+0.1\chi_{\Omega_4}
	\end{array}
\end{equation}
where $\Omega_1=[0.4, 1.6]\times[0.4, 0.6]$, $\Omega_2=[0.4, 0.8]\times[0.8, 1.2]$, $\Omega_3=[0.4, 1.6]\times[1.4, 1.6]$, and $\Omega_4=[1.2, 1.6]\times[0.8, 1.2]$. The anisotropic factor $\eta=0$. We performanced numerical reconstructions using four wavelength-dependent sources. For each source, we have data for four different wavelengths. The results of the reconstructions using noiseless data are presented in Fig.~\ref{FIG:Spec-C} and those using noisy data are shown in Fig.~\ref{FIG:Spec-D}. We observe similar reconstruction qualities (as can be seen in the plots of the pointwise relative error) to the two-coefficient cases in the previous sections.
\begin{figure}[ht]
\centering
\includegraphics[angle=0,width=0.96\textwidth]{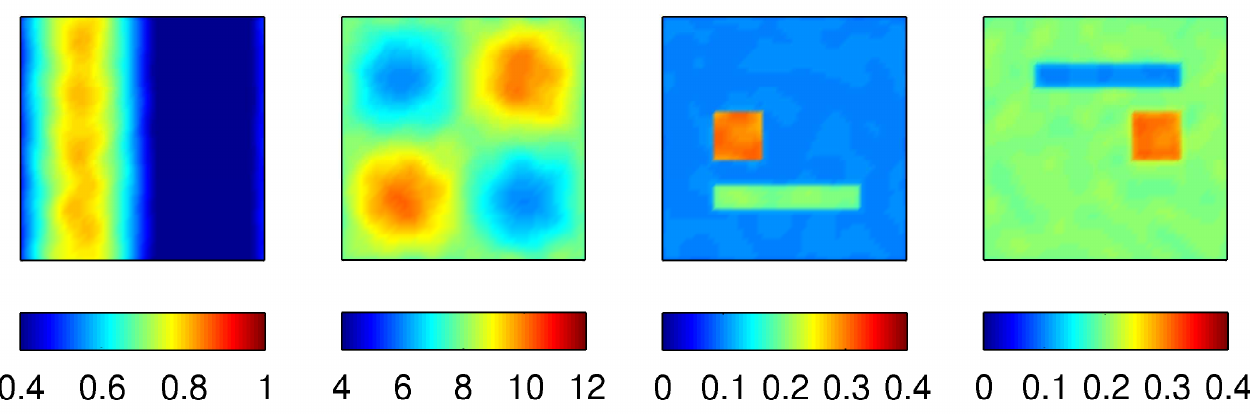}
\includegraphics[angle=0,width=0.96\textwidth]{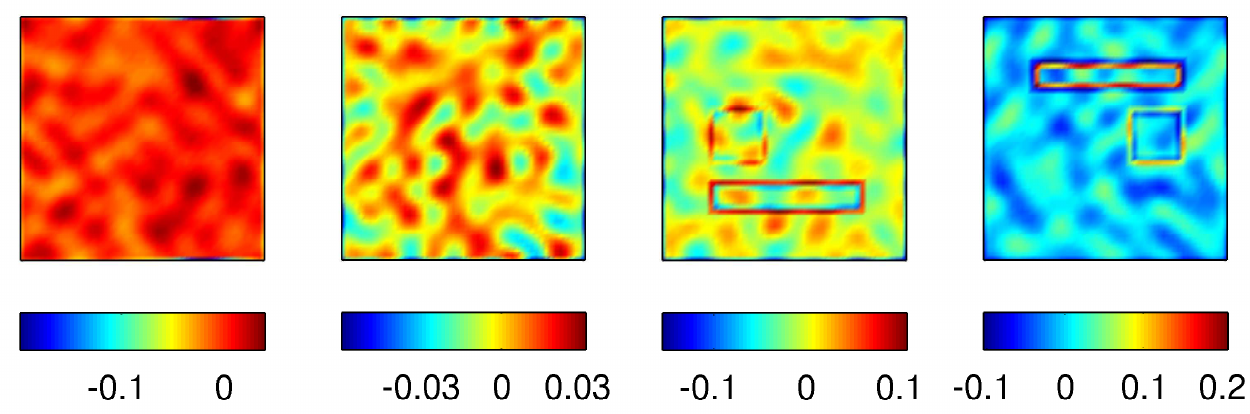}
\caption{Same as in Fig.~\ref{FIG:Spec-C} except that the data used contain 5\% random noise.}
\label{FIG:Spec-D}
\end{figure}

\section{Concluding remarks}
\label{SEC:Concl}

We studied the quantitative photoacoustic tomography problem with the radiative transport model, aiming at reconstructing multiple physical coefficients simultaneously using the data collected from multiple illuminations. We showed that in non-scattering absorbing media, we can reconstruct both the absorption and the Gr\"uneisen coefficients simultaneously in a stable manner using only two sets of interior data. Moreover, in this case we derived explicit reconstruction formulas for the problem with particular choices of illuminations (collimated, point and cone-limited sources). In scattering media, we show, based on the result in~\cite{BaJoJu-IP10}, that one can reconstruct two of the absorption, the scattering and the Gr\"uneisen coefficients stably when more data, i.e., data given by the full albedo operator, is available. To reconstruct all three coefficients simultaneously, we proposed to use interior data collected at different optical wavelengths. We show that with some realistic \emph{a priori} knowledge, mainly the knowledge of the spectral denpence of the coefficients, we can reconstruct stably the spatial component of all three coefficients simultaneously.

Besides the analytical reconstruction strategy for the non-scattering problem, we proposed a linearized reconstruction method based on Born approximation to the original inverse problem as well as a nonlinear reconstruction method based on numerical minimization techniques for QPAT for scattering media. We show numerically that the reconstruction is very stable. In fact, our numerical experiments showed that, assuming that the data measured with ultrasound is accurate enough, the nonlinear least-square formulation of the inverse transport problems with interior data problem behaves like a convex optimization problem and thus can be solved efficiently and accurately; see the numerical reconstructions in Section~\ref{SEC:Num}. 

There are many important issues in quantitative PAT that need to be addressed in the future. For instance, in the theory developed in~\cite{BaRe-IP11,BaUh-IP10,BaRe-IP12} for the diffusion model, one can reconstruct both the absorption and the scattering coefficients with only two ``well-chosen'' illuminations, with little \emph{a priori} assumptions on the coefficients. It would be interesting to see how to generalize that to the radiative transport model~\eqref{EQ:ERT} with $\sigma_s\neq 0$. An equally important question is whether or not we can derive any analytical reconstruction procedure, similar to those developed in diffusion-based theory~\cite{BaUh-IP10,BaRe-IP11,BaRe-IP12}, for the inverse transport problem in scattering media. 

On the numerical side, when the noise present in the data is significant, we need to regularize the reconstruction. The regularization we adopt here is the usual Tikhonov regularization which yields smooth solutions among all possibilities. In certain applications, we might know \emph{a priori} that the coefficients to be reconstructed are piecewise constant, such as those presented in some of the numerical simulations in Section~\ref{SEC:Num}. In this case, alternative regularization strategies such as the total variation (TV) regularization might be more appropriate. In~\cite{BaRe-IP11,GaZh-JBO10,GaZhOs-Prep10}, $L^1$ regularization for non-smooth coefficients have been considered in the diffusion case. It would be interesting to see the performance of that in the transport case. We plan to explore this issue in the future.

\section*{Acknowledgment}

This work is partially supported by National Science Foundation (NSF) through grant DMS-0914825 and a faculty development award from the University of Texas at Austin. 

\bibliography{/home/ren/Academic/Bibliography/BIB-KR,/home/ren/Academic/Bibliography/BIB-YH}
\bibliographystyle{siam}

\end{document}